\documentclass{fundam}
\pdfoutput=1

 \usepackage{url}
 \usepackage{hyperref}

\newcommand{\defref}[1]{Definition~\ref{#1}}
\newcommand{\thmref}[1]{Theorem~\ref{#1}}
\newcommand{\lemref}[1]{Lemma~\ref{#1}}
\newcommand{\exref}[1]{Example~\ref{#1}}
\newcommand{\corref}[1]{Corollary~\ref{#1}}

\newcommand{\propref}[1]{Proposition~\ref{#1}}

\usepackage{graphicx}
\usepackage{inputenc}
\usepackage{comment}

\usepackage{wrapfig}
\usepackage{microtype}

\usepackage{amsfonts,amsmath,amssymb}
\usepackage{color}
\usepackage{stmaryrd}
\usepackage{tikz}
\usetikzlibrary{positioning,shapes.multipart,arrows.meta}
\usepackage{forloop}
\usepackage{multirow}

 \usepackage{listings-rust,paralist}

\usepackage[ruled, vlined]{algorithm2e}
\SetKwProg{Fn}{Function}{}{end}
\SetAlgoLongEnd

\graphicspath{{.}{figures/}}

\usepackage{etoolbox}

\usepackage{float}
\usepackage{xcolor}
 \definecolor{chameleond}{HTML}{4E9A06}
 \definecolor{skyblued}{HTML}{204A87}
 \definecolor{bluue}{HTML}{0047AB}
\definecolor{colorF}{HTML}{8bc28c} 
\definecolor{colorB}{HTML}{f18aad} 
\definecolor{colorCdark}{HTML}{ea6759} 
\definecolor{colorC}{HTML}{f4b3ac} 
\definecolor{colorD}{HTML}{f88f58} 
\definecolor{colorE}{HTML}{f3c65f} 
\definecolor{colorAdark}{HTML}{6667ab} 
\definecolor{colorA}{HTML}{b2b3d5} 

\definecolor{Salade}{HTML}{228B22}%
\colorlet{Salad}{Salade!90!black}%
\definecolor{Bleuh}{HTML}{15508F}%
\colorlet{Oranj}{orange!70!red}%
\colorlet{LightSalad}{colorF}%
\colorlet{LightBleuh}{colorA}%
\colorlet{LightOranj}{colorD}%
\newlength\theWidth
\newlength\theHeight
\newlength\theDepth
\newcommand\SetupBox[2]{%
  \settowidth\theWidth{#2}%
  \settoheight\theHeight{#2}%
  \settodepth\theDepth{#2}%
  \expandafter\newsavebox{#1}%
  \expandafter\savebox{#1}{\parbox[c][\theHeight]{\theWidth}{#2}}%
}%

\newcommand\maxparallel{\rule[-2pt]{2pt}{10pt}}

\newcommand{\asym}{f}                

\newcommand{\sym}[1]{\mathsf{#1}}    
\newcommand{\FZero}{\sym{Zero}}
\newcommand{\FS}{\sym{S}}
\newcommand{\FNil}{\sym{Nil}}
\newcommand{\FTree}{\sym{Tree}}
\newcommand{\FCons}{\sym{Cons}}
\newcommand{\Fplus}{\sym{plus}}
\newcommand{\Fsize}{\sym{size}}
\newcommand{\Fdoubles}{\sym{doubles}}
\newcommand{\FCom}{\sym{Com}}
\newcommand{\FComPar}{\sym{ComPar}}

\newcommand{\Fa}{\sym{a}}
\newcommand{\Fb}{\sym{b}}
\newcommand{\Fc}{\sym{c}}
\newcommand{\Fd}{\sym{d}}

\newcommand{\Ff}{\sym{f}}
\newcommand{\Fg}{\sym{g}}

\newcommand{\Fminus}{\sym{-}}

\newcommand{\Fs}{\sym{S}}
\newcommand{\Fzero}{\sym{0}}

\newcommand{\Fleq}{\sym{leq}}
\newcommand{\Ftrue}{\sym{True}}
\newcommand{\Ffalse}{\sym{False}}
\newcommand{\Fmod}{\sym{mod}}
\newcommand{\Fif}{\sym{if}}
\newcommand{\Fcond}{\sym{cond}}
\newcommand{\Fmax}{\sym{max}}

\newcommand{\tup}[1]{#1^{\sharp}}
\newcommand{\FPLUS}{\tup{\Fplus}}
\newcommand{\FSIZE}{\tup{\Fsize}}

\newcommand{\FDOUBLES}{\tup{\Fdoubles}}

\newcommand{\FD}{\tup{\Fd}}

\newcommand{\cost}{\textit{cost}}

\newcommand{\hole}{\Box}   
\newcommand{\sbot}{\alpha} 
\newcommand{\sDT}{\beta}  

\newcommand{\Dom}{\mathit{Dom}}
\newcommand{\Pol}{{\mathcal{P}ol}}

\newcommand{\Signature}{\ensuremath{\Sigma}}      
\newcommand{\DefSyms}{\ensuremath{\Signature_d}}      
\newcommand{\ConSyms}{\ensuremath{\Signature_c}}      
\newcommand{\symsof}[1]{\Signature^{#1}}
\DeclareMathOperator{\Dh}{dh}                     
\newcommand{\Cplx}{\mathit{Cplx}}                 
\newcommand{\rt}{\mathrm{root}}                   
\newcommand{\tsize}[1]{\ensuremath{\lvert#1\rvert}}

\newcommand{\irc}[1]{\mathrm{irc}_{#1}}
\newcommand{\ircR}{\irc{\RR}}
\newcommand{\pirc}[1]{\mathrm{pirc}_{#1}}
\newcommand{\pircR}{\pirc{\RR}}

\newcommand{\maxanf}[1]{#1\!\Downarrow} 

\newcommand{\detup}{\delta} 

\newcommand{\detupTRS}{\SSS/((\DD\setminus\SSS)\cup\RR)}

\newcommand{\msdc}[2]{\msdcSym(#1;#2)}
\newcommand{\msdcSym}{\mathit{MSDC}}

\newcommand{\lmsdc}{\langle} 
\newcommand{\rmsdc}{\rangle} 

\renewcommand{\msdc}[2]{\lmsdc #2 \rmsdc \in \msdcSym(#1)}

\newcommand{\TT}{\mathcal{T}}
\newcommand{\RR}{\mathcal{R}}
\newcommand{\DD}{\mathcal{D}}
\newcommand{\SSS}{\mathcal{S}}

\newcommand{\RRA}{{\RR_1}}  
\newcommand{\RRB}{{\RR_2}}  

\newcommand{\OO}{\mathcal{O}}
\newcommand{\Oh}{\OO}
\newcommand{\VV}{\mathcal{V}}
\newcommand{\BasicTerms}{\TT_{\mathrm{basic}}}    
\newcommand{\SharpTerms}{\TT^\sharp}              

\newcommand{\Pos}{\mathcal{P}\!\mathit{os}}       
\newcommand{\PosDef}{\Pos_d}                      

\newcommand{\DT}{\mathit{DT}}                     
\newcommand{\DTpar}{\mathit{PDT}}      

\newcommand{\ito}{
  \mathrel{\smash{\stackrel{\raisebox{2pt}{\scriptsize $\mathsf{i}\:$}}%
      {\smash{\rightarrow}}}_{\RR}}
}

\newcommand{\itodetup}{
  \mathrel{\smash{\stackrel{\raisebox{2pt}{\scriptsize $\mathsf{i}\:$}}%
      {\smash{\rightarrow}}}_{\detup(\langle \DD, \SSS, \RR \rangle)}}
}
\renewcommand{\itodetup}{
  \mathrel{\smash{\stackrel{\raisebox{2pt}{\scriptsize $\mathsf{i}\:$}}%
      {\smash{\rightarrow}}}_{\SSS/((\DD\setminus\SSS)\cup\RR)}}
}
\renewcommand{\itodetup}{\someito{\SSS/((\DD\setminus\SSS)\cup\RR)}}

\newcommand{\itosdetup}{
  \mathrel{\smash{\stackrel{\raisebox{2pt}{\scriptsize $\mathsf{i}\:$}}%
      {\smash{\rightarrow^*}}}_{\SSS/((\DD\setminus\SSS)\cup\RR)}}
}
\renewcommand{\itosdetup}{\someitos{\SSS/((\DD\setminus\SSS)\cup\RR)}}

\newcommand{\itop}{
  \mathrel{\smash{\stackrel{\raisebox{2pt}{\scriptsize $\mathsf{i}\:\:\:$}}%
      {\smash{\rightarrow^+}}}_{\RR}}
}
\renewcommand{\itop}{\someitop{\RR}}

\newcommand{\itos}{
  \mathrel{\smash{\stackrel{\raisebox{2pt}{\scriptsize $\mathsf{i}\:$}}%
      {\smash{\rightarrow^*}}}_{\RR}}
}
\renewcommand{\itos}{\someitos{\RR}}

\newcommand{\oldpto}{
  \mathrel{\smash{\stackrel{\raisebox{2pt}{\scriptsize $\:\:\ \:\:$}}%
      {\smash{\longrightarrow\hspace*{-14pt}{\parallel}\hspace*{8pt}}}}_{\RR}}
}

\newcommand{\pito}{
  \mathrel{\smash{\stackrel{\raisebox{2pt}{\scriptsize $\:\:\ \mathsf{i}\:$}}%
      {\smash{\longrightarrow\hspace*{-14pt}{\maxparallel}\hspace*{8pt}}}}_{\RR}}
}

\newcommand{\pitos}{
  \mathrel{\smash{\stackrel{\raisebox{2pt}{\scriptsize $\:\:\ \mathsf{i}\:$}}%
      {\smash{\longrightarrow\hspace*{-14pt}{\maxparallel}\hspace*{8pt}}}}^*_{\RR}}
}

\newcommand{\pitosbelow}{
  \mathrel{\smash{\stackrel{\raisebox{2pt}{\scriptsize $\:\:\ \mathsf{i}$}}%
      {\smash{\longrightarrow\hspace*{-14pt}{\maxparallel}\hspace*{8pt}}}}^*_{\RR,>\varepsilon}}
}

\newcommand{\someto}[1]{
  \mathrel{\smash{\rightarrow}_{#1}}
}
\newcommand{\sometos}[1]{
  \mathrel{\smash{\rightarrow}_{#1}^{*}}
}
\newcommand{\someito}[1]{
  \mathrel{\smash{\stackrel{\raisebox{2pt}{\scriptsize $\mathsf{i}\:$}}%
      {\smash{\rightarrow}}}_{#1}}
}
\newcommand{\someitos}[1]{
  \mathrel{\smash{\stackrel{\raisebox{2pt}{\scriptsize $\mathsf{i}\:$}}%
      {\smash{\rightarrow}}}_{#1}^{*}}
}
\newcommand{\someitop}[1]{
  \mathrel{\smash{\stackrel{\raisebox{2pt}{\scriptsize $\mathsf{i}\:$}}%
      {\smash{\rightarrow}}}_{#1}^{+}}
}
\newcommand{\someitom}[1]{
  \mathrel{\smash{\stackrel{\raisebox{2pt}{\scriptsize $\mathsf{i}\:$}}%
      {\smash{\rightarrow}}}_{#1}^{m}}
}

\newcommand{\aprove}{\textsc{AProVE}\xspace}
\newcommand{\tct}{\textsc{TcT}\xspace}

\newcommand{\starexec}{\textsc{StarExec}\xspace}
\newcommand{\cofloco}{\textsc{CoFloCo}\xspace}
\newcommand{\raml}{\textsc{RAML}\xspace}
\newcommand{\coq}{\textsc{Coq}\xspace}
\newcommand{\isabelleHOL}{\textsc{Isabelle/HOL}\xspace}

\newcommand{\nat}{\mathbb{N}}

\newcommand{\cp}[2]{#1 \rtimes #2}

\newcommand{\ocp}[2]{#1 \Join #2}

\definecolor{darkred}{HTML}{AA0000}

\newcommand{\journal}[1]{\textcolor{darkred}{#1}}
\renewcommand{\journal}[1]{#1}

\SetupBox\sizeoftwohightree{$\Fsize(\FTree(\FZero, \FNil, \FTree(\FZero, \FNil, \FNil)))$}%
\SetupBox\sizeofnil{$\Fsize(\FNil)$}%
\SetupBox\sizeofonehightree{$\Fsize(\FTree(\FZero, \FNil, \FNil))$}%
\SetupBox\pluszerozero{$\Fplus(\FZero, \FZero)$}%
\SetupBox\pluszeroone{$\Fplus(\FZero, \FS(\FZero))$}%

\begin{document}


\setcounter{page}{121}
\publyear{24}
\papernumber{2191}
\volume{192}
\issue{2}

\finalVersionForARXIV


\title{On Complexity Bounds and Confluence of Parallel \newline  Term Rewriting\thanks{
  This work was partially funded by the French National Agency of Research in
  the CODAS Project (ANR-17-CE23-0004-01).
  For Open Access purposes, our extended authors' accepted
  manuscript~\cite{versionArxivJournal} of this paper is available
  under Creative Commons CC BY licence.}}


\author{Tha\"{i}s Baudon \\
LIP (UMR CNRS/ENS Lyon/UCB Lyon1/Inria) \\
Lyon, France
\and
Carsten Fuhs\thanks{Address for correspondence: Birkbeck, University of London, London, United Kingdom.}
 \\
Birkbeck, University of London \\
London, United Kingdom
\and
Laure Gonnord\thanks{Also work:  LIP (UMR CNRS/ENS Lyon/UCB Lyon1/Inria), Lyon, France.}
\\
LCIS, University Grenoble  Alpes\\
Valence ,  France
}

\maketitle

\runninghead{T.\ Baudon et al.}{On Complexity Bounds and Confluence of Parallel Term Rewriting}

\vspace*{-5mm}
\begin{abstract}
  We revisit parallel-innermost term rewriting as a model of parallel
  computation on inductive data structures and provide a corresponding
  notion of runtime complexity parametric in the size of the start
  term.  We propose automatic techniques to derive both upper and
  lower bounds on parallel complexity of rewriting that enable a
  direct reuse of existing techniques for sequential complexity. \journal{Our
  approach to find lower bounds requires confluence of the
  parallel-innermost rewrite relation,
  thus we also provide effective sufficient criteria for proving
  confluence. }The applicability and the precision of the method are
  demonstrated by the relatively light effort in extending the program
  analysis tool \aprove and by experiments on numerous benchmarks from
  the literature.
\end{abstract}


\section{Introduction}
\label{sec:intro}
Automated inference of complexity bounds for parallel computation has
seen a surge of attention in recent years \cite{BaillotPiESOP21,BaillotPiCONCUR21,GallagherLOPSTR19,AlbertTOCL18,HoffmannESOP15,HoffmannICFP18}. While
techniques and tools for a variety of computational models have been
introduced, so far there does not seem to be any paper in this area for
complexity of \emph{term rewriting} with
parallel
evaluation
strategies.
This paper addresses this gap in the literature. We consider
term rewrite systems (TRSs) as \emph{intermediate representation} for
programs with \emph{pattern-matching} operating on \emph{algebraic data
  types}
like the one depicted in \autoref{lst:rustsize}.

\begin{figure}[H]
\begin{lstlisting}[language=caml]
let rec size : tree -> int = function
 | Node (_, left, right) -> 1 + size left + size right
 | Empty -> 0
\end{lstlisting}
\vspace*{-2ex}
\caption[Size of Tree in OCaml]{Tree size computation in OCaml}
\label{lst:rustsize}
\end{figure}

In this particular example, the recursive calls
\lstinline[language=caml]{size left} and
\lstinline[language=caml]{size right} can be done
in parallel.
Building on previous work on parallel-innermost
rewriting~\cite{parallelRewriting,innermostOrdering} and first ideas about parallel
complexity~\cite{wst16trs}, we propose a new notion of Parallel
Dependency Tuples that captures such a behaviour, and methods to
compute both upper and lower \emph{parallel complexity bounds}.

Bounds on parallel complexity can provide insights about the
potentiality of parallelisation: if sequential and parallel complexity
of a function
(asymptotically) coincide, this information can be useful for a
parallelising compiler to refrain from parallelising the evaluation of
this function.
Moreover, evaluation of TRSs (as a simple
functional programming language) in massively
parallel settings such as GPUs is currently a topic of active research
\cite{gpu_trs,gpu_trs_journal}. In this context, a static analysis of parallel complexity can be
helpful to determine whether
to rewrite on a
(fast, but not very parallel) CPU or on a (slower, but massively parallel) GPU.

In this paper, we provide techniques for the synthesis
  of both upper and lower bounds for the parallel-innermost
  runtime complexity of TRSs.
  We motivate our focus on innermost rewrite strategies by the fact
  that innermost rewriting is closely related to call-by-value
  evaluation strategies as used by many programming languages,
  such as OCaml, Scala, Rust, C++, \ldots

  Our approach to finding
  lower bounds requires that the input TRS is \emph{confluent}. This
  means essentially that computations with the TRS have deterministic
  results. Thus, we also provide efficiently checkable
  sufficient criteria for proving confluence of
  parallel-innermost rewriting, which are of interest
  both as an ingredient for our complexity analysis and in their own right.
  These criteria
  capture the confluence of TRSs corresponding
  to programs with deterministic small-step semantics,
  the motivation of this work.

\smallskip
This paper is an extended journal version of a conference paper
published at LOPSTR 2022 \cite{lopstr2022}. We make the following
additional contributions over the conference version:
\begin{compactitem}
\item
We provide additional explanations, examples, and discussion
throughout the paper.
\item
We state a stronger new sufficient criterion for confluence of
parallel-innermost rewriting (\thmref{thm:pito_confluent_with_cps}).
\item
We have run more extended experiments.
\item
We provide proofs for all our theorems.
\end{compactitem}

\section{Illustrating example}
\label{sec:example}

We illustrate our approach informally with the help
of an example.

\medskip
Consider the following functional program, given as a
term rewrite system with the following rewrite rules.
\[
\begin{array}{rcl@{\hspace*{5ex}}|@{\hspace*{5ex}}rcl}
\Fdoubles(\FZero) & \to & \FNil &
  \Fd(\FZero) & \to & \FZero \\
\Fdoubles(\FS(x)) & \to & \FCons(\Fd(\FS(x)), \Fdoubles(x)) &
  \Fd(\FS(x)) & \to & \FS(\FS(\Fd(x)))
\end{array}
\]
Here we use the constructor symbols $\FZero$ and $\FS$
to represent natural numbers (with $\FZero$ for $0$ and
$\FS(x)$ for $x+1$).
Lists are represented via the constructors $\FNil$ and $\FCons$.
Then the function $\Fd$ computes the double of a natural number:
\[
  \Fd(\FS(\FS(\FZero))) \ito  \FS(\FS(\Fd(\FS(\FZero))))
                        \ito  \FS(\FS(\FS(\FS(\Fd(\FZero)))))
                        \ito  \FS(\FS(\FS(\FS(\FZero))))
\]
In other words, $2 \cdot 2 = 4$.
The function $\Fdoubles$
takes a number $n$ and computes a term representing the list
$[2n, 2(n-1), \ldots, 4, 2]$.
To evaluate the start term $\Fdoubles(\FS(\FZero))$,
we need four rewrite steps with a sequential evaluation strategy.
In the following rewrite sequence, coloured boxes indicate terms that are
reduced by a given step, called \emph{redexes}.\footnote{Boxes that keep the same colour throughout one
or more rewrite steps indicate terms that are reduced multiple (consecutive) times.}
\[
\begin{array}{lcl}
\Fdoubles(\FS(\FZero)) &\ito&
\FCons(\colorbox{colorF}{$\Fd(\FS(\FZero))$}, \Fdoubles(\FZero)) \\
&\ito&
\FCons(\FS(\FS(\colorbox{colorF}{$\Fd(\FZero)$})), \Fdoubles(\FZero)) \\
&\ito&
\FCons(\FS(\FS(\FZero)), \colorbox{colorD}{$\Fdoubles(\FZero)$}) \\
&\ito&
\FCons(\FS(\FS(\FZero)), \FNil)
\end{array}
\]
We want to get an upper bound on the number of evaluation steps
with our program for the general case.
We first consider existing methods for the classic case of
a call-by-value strategy in a \emph{sequential} model:
evaluate only a single redex at a time, as in the examples above.
We consider start terms where a defined
function is called on \emph{data terms},
i.e., terms that use only constructor symbols.
Here these terms have the form $\Fd(t)$ and $\Fdoubles(t)$,
where $t$ may contain only the constructor symbols
$\FZero$, $\FS$, $\FNil$, $\FCons$, and variables.
Our upper bound will be parametric in the size $n$ of the start
term: larger start terms usually have higher runtimes.

\emph{Dependency Tuples} \cite{DependencyTuple} are a standard technique
for finding such upper bounds for sequential evaluation.
The idea is to ``desugar''
the program by grouping all the function calls of a rule
together (hence the name ``tuple'') and then to analyse
how big the function call tree can become.

\medskip
For example, the rule
\[
\Fdoubles(\FS(x)) \to \FCons(\Fd(\FS(x)), \Fdoubles(x))
\]
has the dependency tuple
\[
\FDOUBLES(\FS(x)) \to \FCom_2(\FD(\FS(x)), \FDOUBLES(x))
\]
that groups the function calls $\FD(\FS(x))$ and $\FDOUBLES(y)$
into a single tuple, using a fresh constructor symbol
$\FCom_2$ for a tuple of 2 arguments
Here we ignore the constructor context $\FCons(...)$
from the original rule -- evaluating it does not cost anything.
We use the $\sharp$ symbol to indicate which
function calls must be ``paid for'' in the analysis.

\medskip
This dependency tuple says, informally:
\[
\cost(\Fdoubles(\FS(x))) = 1 + \cost(\Fd(\FS(x))) + \cost(\Fdoubles(x))
\]
So evaluating a call to $\Fdoubles(\FS(x)))$ with that
rule ``costs'' us $1$ step (for using the rule itself) $+$
the cost of evaluating the call $\Fd(\FS(x))$ $+$
the cost of evaluating the call $\Fdoubles(x)$.

\medskip
\emph{Polynomial interpretations} \cite{Lankford75} are
the working horse for finding upper bounds on the complexity
of such complexity problems represented by Dependency Tuples.
A polynomial interpretation $\Pol$ maps function symbols
to polynomial functions over the natural numbers, and extends naturally to terms.
If we can find an interpretation $\Pol$
such that, among other requirements,
\[
\Pol(\FDOUBLES(\FS(x))) \geq 1 + \Pol(\FD(\FS(x))) + \Pol(\FDOUBLES(x))
\]
then the highest degree of a polynomial in $\Pol$
for a symbol $\tup{f}$ is also an upper bound for
the size of the call tree possible for these
Dependency Tuples with the given program: the polynomial
function \emph{overapproximates} the cost function.

Such polynomial interpretations can be found automatically
using modern constraint solvers \cite{satPolo,smtPolo}.
For our example, we would find an interpretation of degree 2.
This tells us that the complexity for evaluating
a function in our program using a sequential call-by-value
strategy is bounded by $\mathcal{O}(n^2)$
for $n$ as the size of the start term.
The above is just an informal overview -- \autoref{sec:irc}
provides a formal introduction to this
approach to analysing complexity for
\emph{sequential} evaluation.

The bound is tight: from
$\Fdoubles(\FS(\FS(\ldots\FS(\FZero)\ldots)))$, we get
linearly many calls to the linear-time function $\Fd$
on arguments of size linear in the start term.
So this is as good as it gets\ldots for a sequential
evaluation strategy.

\medskip
Now, how about a \emph{parallel} call-by-value evaluation
strategy, where we evaluate function calls that happen at
independent positions \emph{at the same time} rather than
one after the other?
It turns out that with a \emph{parallel} strategy,
we can evaluate our start term
$\Fdoubles(\FS(\FZero))$ in just \emph{three} steps (the two coloured redexes are reduced in parallel):
\[
\begin{array}{lcl}
\Fdoubles(\FS(\FZero)) &\pito&
\FCons(\colorbox{colorF}{$\Fd(\FS(\FZero))$}, \colorbox{colorD}{$\Fdoubles(\FZero)$}) \\
&\pito&
\FCons(\FS(\FS(\colorbox{colorF}{$\Fd(\FZero)$})), \FNil) \\
&\pito&
\FCons(\FS(\FS(\FZero)), \FNil)
\end{array}
\]
Can we expect such speed-ups in the number of steps
also in the general case?
It turns out that, for our example, the answer is \emph{yes}.
To prove this automatically, we revisit Dependency Tuples
using a little trick. The reason why parallel evaluation
is faster here is that the rule
\[
\Fdoubles(\FS(x)) \to \FCons(\Fd(\FS(x)), \Fdoubles(x))
\]
makes two function calls at independent positions, which
then get evaluated \emph{in parallel}.
Whichever of the two calls to $\Fd(\FS(x))$ and to $\Fdoubles(x)$
finishes last is responsible for the overall cost of using
this rule and evaluating its function calls in parallel.
We express this by introducing two \emph{separate}
Dependency Tuples for our rewrite rule:
\[
\begin{array}{rcl}
\FDOUBLES(\FS(x)) &\to& \FCom_1(\FD(\FS(x)))\\
\FDOUBLES(\FS(x)) &\to& \FCom_1(\FDOUBLES(x))
\end{array}
\]
These two Dependency Tuples are enough to capture the worst case:
either the call to $\FD(\FS(x))$ takes longer to evaluate
than the call to $\FDOUBLES(x)$ (then the first one represents
the worst case), or it does not (then the second one represents
the worst case).

\medskip
We can now reuse the same analysis machinery as before
to search for a polynomial interpretation $\Pol$ that
solves the following constraints:
\[
\begin{array}{rcl}
\Pol(\FDOUBLES(\FS(x))) &\geq& 1 + \Pol(\FCom_1(\FD(\FS(x))))\\
\Pol(\FDOUBLES(\FS(x))) &\geq& 1 + \Pol(\FCom_1(\FDOUBLES(x)))
\end{array}
\]
Here our constraint solvers find a solution already
for a parametric interpretation with templates of degree 1.
This tells us that the complexity for evaluating
a function in our program using a \emph{parallel} call-by-value
strategy is bounded by $\mathcal{O}(n)$
for $n$ as the size of the start term, a strictly better
bound than is possible for sequential evaluation.

\medskip
However, this example is very benign: here \emph{all} function
calls triggered by a rule are at independent positions.
How about a function to compute the size of a tree with
a rule like the following?
\[
\Fsize(\FTree(v, l, r)) \to \FS(\Fplus(\Fsize(l), \Fsize(r)))
\]
Here $\Fplus$ must wait for the calls to $\Fsize$ to finish,
and we cannot evaluate $\Fplus$ in parallel with the calls to
$\Fsize$. As we shall see in \autoref{sec:para_complex},
a refinement of our method can be used to also deal with
such more complicated structural dependencies between function calls.

The above example shows that our analysis of parallel complexity
can reuse machinery from sequential complexity analysis
provided by the Dependency Tuple framework.
But Dependency Tuples are just one method out of a plethora
for analysis of sequential complexity of term rewrite systems.
Can we not go back from Dependency Tuples to term rewriting
systems whose \emph{sequential} complexity we can then analyse
in order to get upper bounds for the \emph{parallel}
complexity of our \emph{original} program?
\autoref{sec:dt_to_irc} shows us how this is possible.

\medskip
We would also like to find \emph{lower} bounds on
the complexity of parallel evaluation.
It turns out that if we are dealing with a \emph{confluent}
input program (roughly speaking, a program with deterministic
results), \autoref{sec:dt_to_irc} lets us reuse methods
to find lower bounds for sequential complexity of term rewriting
systems to get answers for lower bounds for the \emph{parallel}
complexity of our \emph{original} program.
But to be able to apply these methods, we must somehow
know that the input program is indeed confluent.
This is why \autoref{sec:confluence} introduces methods
for analysing \emph{confluence} of a term rewrite system
in our parallel evaluation strategy.
Then \autoref{sec:expe} gives experimental
evidence of the practicality of our methods
on large standard benchmark sets.
We discuss related work \journal{and conclude}
in~\autoref{sec:related}.

\emph{Limitations.}
Our approach to complexity analysis of parallel-innermost rewriting
is transformational, by generating problem instances that can be
handled by existing complexity analysis tools for (sequential)
innermost rewriting \cite{aprove-tool,tct} as backends.
Therefore, the precision of our analysis is
limited by the precision of these backend tools, and
improvements
to their precision should carry over directly
also to the analysis of parallel-innermost rewriting.

Moreover, while our contributions aim to be applicable to
complexity analysis of programming languages with algebraic data types that
use innermost/call-by-value evaluation strategies,
we present them in the setting of first-order term rewriting
without built-in data types, for parallel-innermost runtime complexity.
Extensions to first-order term rewriting with logical
constraints \cite{lctrs13,lctrsTocl},
to higher-order rewriting without \cite{thesisKop}
and with \cite{esop24lcstrs} logical constraints,
and to rewrite strategies that rewrite at even more positions
simultaneously than parallel-innermost rewriting \cite{gpu_trs_journal}
would be natural next steps.
Similarly, we anticipate extensions from term rewriting
to programming languages with call-by-value evaluation strategies,
such as OCaml or Scala.

\section{Term rewriting and innermost runtime complexity}
\label{sec:irc}

We assume basic familiarity with term rewriting (see, e.g.,
\cite{BaaderNipkow}) and recall standard definitions to fix notation, which we illustrate in Example~\ref{ex:size1}.
As customary for analysis of runtime complexity of rewriting,
we consider terms as \emph{tree-shaped} objects,
without sharing of subtrees.

\medskip
We first define \emph{Term Rewrite Systems} and \emph{Innermost Rewriting}.
$\TT(\Signature, \VV)$ denotes the set of \emph{terms} over a
finite signature $\Signature$ and the set of variables $\VV$.
For a term $t$, its
\emph{size} $\tsize{t}$
is defined by:
(a)
if $t \in \VV$, then
$\tsize{t} = 1$;
(b)
if $t = f(t_1, \ldots, t_n)$, then
$\tsize{t} = 1 + \sum_{i = 1}^n \tsize{t_i}$.
The set $\Pos(t)$ of the \emph{positions} of a term $t$
is defined by:
(a) if $t \in \VV$, then $\Pos(t) = \{ \varepsilon \}$,
and (b) if $t = f(t_1,\ldots,t_n)$, then
$\Pos(t) = \{ \varepsilon \} \cup
 \bigcup_{1 \leq i \leq n}\{ i.\pi \mid \pi \in \Pos(t_i) \}$.
The position $\varepsilon$ is
the \emph{root position}
of term $t$.

\medskip
If $t = \asym(t_1,\ldots,t_n)$,
$\rt(t) = \asym$ is the \emph{root symbol} of $t$.
The \emph{(strict) prefix order} $>$ on positions is
the strict partial order given by:
$\tau > \pi$ iff there exists $\pi' \neq \varepsilon$ such that
$\pi.\pi' = \tau$.
Two positions $\pi$ and $\tau$ are \emph{parallel} iff neither
$\pi > \tau$ nor $\pi = \tau$ nor $\tau > \pi$ hold.
For $\pi \in \Pos(t)$, $t|_\pi$ is the subterm of $t$
at position $\pi$, and we write $t[s]_\pi$ for the term that results
from $t$ by replacing the subterm $t|_\pi$ at position $\pi$ by the term $s$.
A \emph{context} $C[]$ is a term that contains exactly one occurrence of a
  special symbol $\hole$. Similar to $t[s]_\pi$ (but omitting
  the position $\pi$ because it is implied by the sole
  occurrence of $\hole$), we write $C[s]$ for the term
  obtained from replacing $\hole$ by the term $s$.

A substitution $\sigma$ is a mapping from $\VV$ to
$\TT(\Signature, \VV)$ with finite domain
$\Dom(\sigma) = \{ x \in \VV \mid \sigma(x) \neq x \}$.
We write
$\{x_1 \mapsto t_1; \ldots; x_n \mapsto t_n\}$ for
a substitution $\sigma$ with $\sigma(x_i) = t_i$ for
$1 \leq i \leq n$ and $\sigma(x) = x$ for $x \in \VV$ with $x\neq x_i$.
We extend substitutions to terms by
$\sigma(f(t_1,\ldots,f_n)) = f(\sigma(t_1),\ldots,\sigma(t_n))$.
We may write $t\sigma$ for $\sigma(t)$.

\medskip
For a term $t$, $\VV(t)$ is the set of variables
in
$t$.
A \emph{term rewrite system (TRS)} $\RR$ is a set of rules
   $\{ \ell_1 \to r_1,  \ldots, \ell_n \to r_n \}$
  with
    $\ell_i, r_i \in \TT(\Signature, \VV)$,
    $\ell_i \not\in \VV$,
  and $\VV(r_i) \subseteq \VV(\ell_i)$ for all $1 \leq i \leq n$.
The \emph{rewrite relation} of $\RR$ is
$s \to_\RR t$ iff
  there are
   a rule $\ell \to r \in \RR$,
   a position $\pi \in \Pos(s)$,
   and a substitution $\sigma$
  such that
   $s = s[\ell\sigma]_\pi$ and
   $t = s[r\sigma]_\pi$.
  Here, $\sigma$ is called the \emph{matcher} and the term $\ell\sigma$
  the \emph{redex} of the rewrite step.
  If  no proper subterm of  $\ell\sigma$ is
  a possible redex,
  $\ell\sigma$ is an \emph{innermost redex}, and the rewrite step
   is an \emph{innermost rewrite step}, denoted by $s \ito t$.

 $\DefSyms^{\RR}=\{ f \mid f(\ell_1,\ldots,\ell_n) \to r \in \RR \}$ and
 $\ConSyms^{\RR}= \Signature \setminus \DefSyms^{\RR}$
 are the \emph{defined}  and \emph{constructor} symbols of $\RR$.
 We may
 also just write $\DefSyms$ and $\ConSyms$.
 The set of positions with defined symbols of $t$ is $\PosDef(t) = \{ \pi \mid \pi \in \Pos(t), \rt(t|_\pi) \in \DefSyms \}$.

 For a relation $\to$, $\to^+$ is its transitive closure
 and $\to^*$ its reflexive-transitive closure. An object $o$ is a \emph{normal
 form} (also: in normal form)
 w.r.t.\ a relation $\to$ iff there is no $o'$ with $o \to o'$.
 A relation $\to$ is \emph{confluent} iff $s \to^* t$ and
 $s \to^* u$ implies that
 there exists an object $v$ with
 $t \to^* v$ and $u \to^* v$.
 A relation $\to$ is \emph{terminating} iff
 there is no infinite sequence $t_0 \to t_1 \to t_2 \to \cdots$.

\begin{example}[$\Fsize$]
\label{ex:size1}
Consider the TRS $\RR$ with the following rules modelling the code of~\autoref{lst:rustsize}.\vspace*{-1mm}
\[
\begin{array}{rcl@{\hspace*{5ex}}|@{\hspace*{5ex}}rcl}
\Fplus(\FZero, y) &\to& y & \Fsize(\FNil) &\to &\FZero \\
\Fplus(\FS(x), y) &\to& \FS(\Fplus(x, y)) &
\Fsize(\FTree(v, l, r)) &\to& \FS(\Fplus(\Fsize(l), \Fsize(r)))
\end{array}\vspace*{-1mm}
\]
Here $\DefSyms^{\RR} = \{ \Fplus, \Fsize \}$ and
$\ConSyms^{\RR} = \{ \FZero, \FS, \FNil, \FTree \}$.

\medskip
First, consider the term  $t=\FS(\Fplus(\FZero, \FS(\FZero)))$.
Its size is $5$. Its positions are $\Pos(t)=\{\varepsilon, 1, 1.1, 1.2, 1.2.1\}$.
In $t$, at position $\pi=1$ we have the subterm $t|_1 = \Fplus(\FZero, \FS(\FZero))$.
This term matches the first rule of our TRS $\Fplus(\FZero, y) \to y$ with
the substitution $\sigma=\{y \mapsto \FS(\FZero)\}$. We can therefore reduce
$t$ to $\FS(\FS(\FZero))$.

Beginning from $t'=\Fsize(\FTree(\FZero, \FNil, \FTree(\FZero, \FNil, \FNil)))$, we have the following innermost rewrite sequence, where
the used innermost redexes are put in coloured boxes:
\[
\begin{array}{rl}
&\colorbox{colorF}{\usebox\sizeoftwohightree}\\
\ito & \FS(\Fplus(\colorbox{colorB}{\usebox\sizeofnil}, \usebox\sizeofonehightree))\\
\ito & \FS(\Fplus(\FZero, \colorbox{colorE}{\usebox\sizeofonehightree}))\\
\ito & \FS(\Fplus(\FZero, \FS(\Fplus(\colorbox{colorD}{\usebox\sizeofnil}, \usebox\sizeofnil))))\\
\ito & \FS(\Fplus(\FZero, \FS(\Fplus(\FZero, \colorbox{colorA}{\usebox\sizeofnil}))))\\
\ito & \FS(\Fplus(\FZero, \FS(\colorbox{colorE}{\usebox\pluszerozero})))\\
\ito & \FS(\colorbox{colorF}{\usebox\pluszeroone})\\
\ito & \FS(\FS(\FZero))
\end{array}
\]
This rewrite sequence uses 7 steps to reach a normal
form as the result of the computation.
\end{example}

Our objective is to provide static bounds on the
length of the longest
rewrite sequence from terms of a specific size.
Here we
use innermost evaluation strategies,
which closely correspond to call-by-value strategies used in many
programming languages.
We focus
on rewrite sequences that start with \emph{basic terms},
corresponding to function calls where a function is applied to data
objects. The
resulting
notion of complexity for term
rewriting is known as \emph{innermost runtime complexity}.

\begin{definition}[Derivation Height $\Dh$, Innermost Runtime Complexity $\irc{}$ \cite{Hirokawa08IJCAR,DependencyTuple}]
  \label{def:rc}
  For all $P \subseteq \nat \cup \{\omega \}$, $\sup\, P$ is the
least upper bound of $P$, where $\sup\, \emptyset = 0$
and $\omega$ is the smallest infinite ordinal, i.e.,
$\omega > n$ holds for all $n \in \nat$.
The \emph{derivation height} of a term $t$ w.r.t.\ a relation $\to$ is 
\eject

\noindent  the length of the longest sequence of $\to$-steps
from $t$:
$\Dh(t, \to) = \sup \{ e \mid \exists\, t' \in \TT(\Signature, \VV).\;
t \to^e t' \}$ where $\to^e$ is the
$e$\textsuperscript{th}
iterate of $\to$.
If $t$ starts an infinite $\to$-sequence, we write $\Dh(t, \to) =
\omega$.

A term $f(t_1, \ldots, t_k)$ is \emph{basic (for a TRS $\RR$)}
iff $f\in\DefSyms^{\RR}$ and $t_1, \dots, t_k \in \TT(\ConSyms^{\RR}, \VV)$.
$\BasicTerms^{\RR}$ is the set of basic terms
for a TRS $\RR$.
For $n \in \mathbb{N}$,
  the \emph{innermost runtime complexity} function
  is
  $\ircR(n) = \sup \{ \Dh(t, {\ito}) \mid t \in \BasicTerms^{\RR}, \tsize{t} \leq
n \}$.
\end{definition}

 Many automated techniques have been proposed
\cite{Hirokawa08IJCAR,DependencyTuple,HirokawaMoser14,ava:mos:16,naa:fro:bro:fuh:gie:17,LowerBounds,MoserS20}
to analyse $\ircR$ and compute bounds on it.  We build on Dependency Tuples
\cite{DependencyTuple}, originally designed to find upper bounds
for (sequential) innermost runtime complexity.  A central
idea is to group all function calls\footnote{Here we use the term
  ``function call'' for a subterm $f(t_1,\ldots,t_n)$
  with a defined symbol $f$ at its root
  to capture the corresponding intuition from functional programming.
  In contrast to most standard functional programming languages,
  in term rewriting it is possible that such a function call can be
  evaluated in several ways (non-determinism), or not at all, so that
  $f$ need not describe a (total or even partial) function from terms to terms in the
  mathematical sense.}
by a rewrite rule \emph{together}
rather than to separate them\journal{, in contrast to Dependency Pairs}
for proving termination
\cite{DependencyPairs}. We use \emph{sharp terms} to represent these
function calls.

\begin{definition}[Sharp Terms $\SharpTerms$]
For every $f \in \DefSyms$, we introduce a fresh symbol
$\tup{f}$ of the same arity, called a \emph{sharp symbol}.
For a term $t = f(t_1,\ldots,t_n)$ with
$f \in \DefSyms$, we define $\tup{t} = \tup{f}(t_1,\ldots,t_n)$.
For all other terms $t$, we define $\tup{t} = t$.
$\SharpTerms = \{ \tup{t} \mid t \in \TT(\Signature, \VV),
\rt(t) \in \DefSyms \}$ denotes the set of \emph{sharp terms}.
\end{definition}

To get an upper bound for sequential complexity, we
``count'' how often
each rewrite rule is used.  The idea is
that when a rule $\ell \to r$ is used,
the cost (i.e., number of rewrite steps for the evaluation)
of the function call to the instance of $\ell$
is $1$ $+$ the sum
of the costs of all the function calls in the resulting instance of $r$,
counted separately.
To group $n$ function calls together, we use ``compound symbols''
$\FCom_n$ of arity $n$, which intuitively represent the sum of
the runtimes of their arguments.

\begin{definition}[Dependency Tuple, DT \cite{DependencyTuple}]
\label{def:dt}
A \emph{dependency tuple (DT)} is a rule of the form
$\tup{s} \to \FCom_n(\tup{t}_1,\ldots,\tup{t}_n)$
where $\tup{s}, \tup{t}_1,\ldots,\tup{t}_n \in \SharpTerms$.
Let $\ell \to r$ be a rule with $\PosDef(r) = \{ \pi_1, \ldots, \pi_n \}$
and $\pi_1 \gtrdot \ldots \gtrdot \pi_n$
where $\gtrdot$ is the standard lexicographic order on positions.
Then $\DT(\ell \to r) = \tup{\ell} \to \FCom_n(\tup{r|}_{\pi_1},\ldots,\tup{r|}_{\pi_n})$.\footnote{The
original definition of Dependency Tuples \cite{DependencyTuple}
allows for using an \emph{arbitrary} total order instead of
the lexicographic order on positions for $\gtrdot$.
The theory presented in this paper would work also with the
original definition.
The order $\gtrdot$ must be total to ensure
that the function $\DT$ is well defined w.r.t.\ the order of
the arguments of $\FCom_n$, so the (partial!)\ prefix order $>$
is not sufficient here.
}
For a TRS $\RR$, let $\DT(\RR) = \{ \DT(\ell \to r) \mid \ell \to r \in \RR \}$.
\end{definition}

\begin{example}
\label{ex:sizeDTs}
For $\RR$ from \exref{ex:size1}, $\DT(\RR)$ consists of the
following DTs:
 \[
  \begin{array}{rcl}
 \FPLUS(\FZero, y) & \to &\FCom_0 \\[-0.5pt]
 \FPLUS(\FS(x), y) & \to &\FCom_1(\FPLUS(x, y)) \\[-0.5pt]
 \FSIZE(\FNil) & \to &\FCom_0 \\[-0.5pt]
 \FSIZE(\FTree(v, l, r)) &\to& \FCom_3(\FSIZE(l), \FSIZE(r), \FPLUS(\Fsize(l), \Fsize(r)))
 \end{array}
 \]
\indent   Intuitively, the DT
$\FSIZE(\FTree(v, l, r)) \to
   \FCom_3(\FSIZE(l), \FSIZE(r), \FPLUS(\Fsize(l), \Fsize(r)))$
distils the information about the function calls that
we need to$\,$ ``count'' $\,$from the right-hand side$\,$ of the original

\vfill\eject
\noindent rewrite rule
$\Fsize(\FTree(v, l, r)) \to \FS(\Fplus(\Fsize(l), \Fsize(r)))$,
and the DT represents this information in a more structured way:
\begin{inparaenum}
\item the constructor context $\FS(\Box)$
on the right-hand side that is not needed for
counting function calls is removed;
\item now all the function calls
that need to be counted are present as the \emph{sharp terms}
$\FSIZE(l)$, $\FSIZE(r)$, and $\FPLUS(\Fsize(l), \Fsize(r))$;
\item
these sharp terms are direct arguments of the new compound symbol
$\FCom_3$; and
\item the function calls
$\Fsize(l)$ and $\Fsize(r)$ below $\FPLUS$ on the right-hand side
are now ignored for their direct contribution
to the cost (their cost is accounted for via $\FSIZE(l)$
and $\FSIZE(r)$), but are considered
only for their normal forms from innermost
evaluation that will be used for evaluating $\FPLUS$ in the
recursive call.
\end{inparaenum}
\end{example}

To represent the number of
rewrite steps
used in the worst case
to reduce a sharp term
according to a set of DTs and a TRS $\RR$,
\emph{chain trees} are used \cite{DependencyTuple}.
Intuitively, a chain tree for some sharp term is a dependency tree
of the computations involved in evaluating this term.
Each node represents a computation
(the function calls represented by the DT with its
special syntactic structure)
on some arguments (defined by the substitution).

We now use a \emph{tree} structure to represent
the computation represented by the Dependency Tuples
rather than a rewrite \emph{sequence} or a linear \emph{chain}
(as with Dependency Pairs). The reason is that we now have
\emph{several} function calls on the right-hand sides of our
Dependency Tuples, and we want to trace their computations
in the tree \emph{independently}. Thus, each function
call gives rise to a new subtree.

Each actual innermost rewrite sequence with $\RR$
will have a corresponding chain tree constructed using
$\DT(\RR)$, with $\RR$ used implicitly for the calls
to helper functions inside the sharp terms. This will make
chain trees useful for finding bounds on the maximum length
of rewrite sequences with $\RR$.

\begin{definition}[Chain Tree \cite{DependencyTuple}]
Let $\DD$ be a set of DTs and $\RR$ be a TRS.
Let $T$ be a (possibly infinite) tree where each node is labelled with
a DT $\tup{q}\! \to \FCom_n(\tup{w}_1,\ldots,\tup{w}_n)$ from $\DD$ and
a substitution $\nu$, written
$(\tup{q}\! \to \FCom_n(\tup{w}_1,\ldots,\tup{w}_n) \mid \!\nu)$.
Let the root node
be labelled with $(\tup{s} \to \FCom_e(\tup{r}_1,\ldots,\tup{r}_e) \mid \sigma)$.
Then
$T$ is a \emph{$(\DD,\RR)$-chain tree for $\tup{s} \sigma$} iff
the following conditions hold for any node of $T$,
where $(\tup{u} \to \FCom_m(\tup{v}_1,\ldots,\tup{v}_m) \mid \mu)$
is the label of the node:
\begin{itemize}
\itemsep=0.8pt
\item $\tup{u}\mu$ is in normal form w.r.t.~$\RR$;
\item if this node has the children
$(\tup{p}_1 \to \FCom_{m_1}(\ldots) \mid \delta_1), \ldots,
(\tup{p}_k \to \FCom_{m_k}(\ldots) \mid \delta_k)$, then there are
pairwise different $i_1,\ldots,i_k \in \{1,\ldots,m\}$ with
$\tup{v}_{i_j} \mu \itos \tup{p}_j \delta_j$ for all $j \in \{1,\ldots,k\}$.
\end{itemize}
\end{definition}
  A chain tree represents a
  rewrite
  sequence starting from some basic term.
  Each node captures exactly one rewrite step in this
  sequence. Its label consists of the DT corresponding to the applied rewrite rule
  and of the substitution with which the rule matches the term.
  A node has
  at most
  as many children as the arity of the compound symbol on the
  right-hand side of its DT. Each of its children captures one of the compound
  symbol's arguments and represents the remaining rewrite steps for this particular
  subterm. In the definition, we
  allow innermost
  rewrite
  steps to occur from
  compound symbols' subterms to their corresponding child node.
  This does not affect the overall cost of the computation represented
  by the chain tree. Indeed, the cost of every redex that appears as a
  result of rewriting is captured by the dependency tuple. Therefore,
  there is no need to ``count'' the cost of redexes that appear in one
  of a compound symbol's arguments: they have already been accounted for
  in another of its arguments.

  We illustrate this notion in Example~\ref{ex:chainTree}.
\eject

\begin{example}
\label{ex:chainTree}
    For $\RR$ from \exref{ex:size1} and $\DD\! = \DT(\RR)$ from
    \exref{ex:sizeDTs}, the following is a
    chain tree for the term
    $\tup{s} = \FSIZE(\FTree(\FZero, \FNil, \FNil))$:
\medskip

    \noindent
{\scalebox{0.88}{%
    \SetupBox\fsizenil{$\FSIZE(\FNil)$}%
\SetupBox\sizel{$\Fsize(l)$}%
\SetupBox\sizer{$\Fsize(r)$}%
\SetupBox\zeroBox{$\FZero$}%
\SetupBox\fpluszeroy{$\FPLUS(
  \colorbox{LightOranj}{\usebox\zeroBox},
  \colorbox{LightBleuh}{$y$})$}%
\begin{tikzpicture}[align=left, node distance=1.5cm and 1cm,
  hlterm/.style={rectangle, rounded corners=false, inner sep=2pt},
  chainlink/.style={draw, rectangle split, rounded corners,
  rectangle split parts=2, rectangle split horizontal=true}]
  \matrix [draw, rectangle, rounded corners, inner sep=2pt, nodes={anchor=base}]
  (root) {
    \node (lhs) {$\FSIZE(\FTree(v, l, r)) \to \FCom_3\big($}; &
    \node [hlterm, fill=LightOranj] (lsize) {$\FSIZE(l)$}; &[-1pt]
    \node {$,$}; &[2pt]
    \node [hlterm, fill=LightBleuh] (rsize) {$\FSIZE(r)$}; &[-1pt]
    \node {$,$}; &[2pt]
    \node [hlterm, fill=LightSalad] (pluss) {$\FPLUS(
      \colorbox{LightOranj}{\usebox\sizel},
      \colorbox{LightBleuh}{\usebox\sizer} )$}; &
    \node {$\big)$}; &
    \node (subst) {$\{v \mapsto \FZero; l \mapsto \FNil; r \mapsto \FNil\}$}; \\
  };
  \path (subst.north west) -- ++(0, 4pt) coordinate (line top);
  \path (subst.south west) -- ++(0, -4pt) coordinate (line bottom);
  \draw (line top) -- (line bottom);
  \node [below left=of lsize, chainlink] (left) {
    $\colorbox{LightOranj}{\usebox\fsizenil} \to \FCom_0$
    \nodepart{two} $\{\}$
  };
  \node [below=of rsize, chainlink] (right) {
    $\colorbox{LightBleuh}{\usebox\fsizenil} \to \FCom_0$
    \nodepart{two} $\{\}$
  };
  \node [below right=of pluss, chainlink, anchor=north] (plus) {
    $\colorbox{LightSalad}{\usebox\fpluszeroy} \to \FCom_0$
    \nodepart{two} $\{y \mapsto \FZero\}$
  };
  \draw[-Latex] (lsize.south) -- (left.north);
  \draw[-Latex] (rsize.south) -- (right.north);
  \draw[-Latex] (pluss.south) --
    (plus.north);
\end{tikzpicture}
 }
}\\
  The root node of the chain tree represents
  the rewrite step
\[
  \Fsize(\FTree(\FZero, \FNil, \FNil)) \ito
  \Fs(\Fplus(\Fsize(\FNil), \Fsize(\FNil)))
\]
  via the following DT:
\[
   \FSIZE(\FTree(v, l, r)) \to
   \FCom_3(\FSIZE(l), \FSIZE(r), \FPLUS(\Fsize(l), \Fsize(r)))
\]
  This node has three children, each corresponding to a redex
  that will be evaluated when rewriting
  $\Fs(\Fplus(\Fsize(\FNil), \Fsize(\FNil)))$ to normal form.
  Its first two children represent the reduction of the subterms at positions
  $1.1$ and $1.2$: $\Fsize(\FNil) \ito \FZero$, both with the
  corresponding DT $\FSIZE(\FNil) \to \FCom_0$.

  Its third child node represents the reduction of the rewritten subterm
  at position $1$ of the right-hand side of the rewrite rule,
  i.e., $\Fplus(\Fsize(\FNil), \Fsize(\FNil))$.
  In this node, we capture the rewrite step
  $\Fplus(\FZero, \FZero) \ito \FZero$ with the DT
  $\FPLUS(\FZero, y) \to \FCom_0$.
  In order to reach the term $\FPLUS(\FZero, \FZero)$
  from the sharp term
  $\FPLUS(\Fsize(\FNil), \Fsize(\FNil))$ that appears inside $\FCom_3$
  (and to reach a normal form w.r.t.\ $\ito$ in its arguments),
  we must first perform the rewrite steps
  $\Fsize(\FNil) \itos \FZero$ on each of its
  arguments.
  For the purposes of reaching the third child node,
  these rewrite steps are ``for free'', as they have already
  been accounted for by the two first children.
\end{example}

Analogous to the
derivation height $\Dh(t, \ito)$ for the number of rewrite steps
in the longest innermost rewrite sequence from a term $t$,
the notion of \emph{complexity}
$\Cplx(\tup{t})$
captures the maximum of the number of nodes of all chain trees
for a sharp term $\tup{t}$. As we shall see, this complexity
of a sharp term $\tup{t}$ w.r.t.\ chain trees
provides an upper bound on the
derivation height of its unsharped version $t$ for the
original TRS. One can lift $\Cplx$ to innermost runtime
complexity for Dependency Tuples as a function of term size,
and this notion for Dependency Tuples will
provide a bound on the innermost runtime complexity of the
original TRS.

\begin{definition}[$\Cplx$ \cite{DependencyTuple}]
Let $\DD$ be a set of DTs and $\RR$ be a TRS.
Let $\SSS \subseteq \DD$
and $\tup{s} \in \SharpTerms$. For a chain tree $T$,
$|T|_\SSS \in \mathbb{N} \cup \{\omega\}$
is the number of nodes in
$T$
labelled with a DT from $\SSS$. We define
$\Cplx_{\langle \DD, \SSS, \RR \rangle}(\tup{s}) = \sup \{ |T|_\SSS \mid
T \text{ is a } (\DD,\RR)\text{-chain tree for }
\tup{s} \}$.
For terms $\tup{s}$ without a $(\DD,\RR)$-chain tree,
we define $\Cplx_{\langle \DD, \SSS, \RR \rangle}(\tup{s}) = 0$.
\end{definition}
For automated complexity analysis with DTs, the following notion of
\emph{DT problems} is used as a characterisation of DTs
that we reduce in incremental proof steps
to a trivially solved problem.

\begin{definition}[DT Problem, Complexity of DT Problem \cite{DependencyTuple}]
\label{def:dt_problem}
Let $\RR$ be a TRS, $\DD$ be a set of DTs, $\SSS \subseteq \DD$.
Then $\langle \DD, \SSS, \RR \rangle$ is a DT problem.
Its complexity function is
$\irc{\langle \DD, \SSS, \RR \rangle}(n) =
\sup \{ \Cplx_{\langle \DD, \SSS, \RR \rangle}(\tup{t}) \mid
        t \in \BasicTerms^\RR, |t| \leq n \}$.
For any TRS $\RR$, the DT problem
$\langle \DT(\RR), \DT(\RR), \RR \rangle$ is called the
\emph{canonical DT problem} for $\RR$.
\end{definition}

For a DT problem $\langle \DD, \SSS, \RR \rangle$, the set $\DD$
contains all DTs that can be used in chain trees -- and whose complexity
we want to analyse.
$\SSS$ contains the DTs whose complexity remains to be
analysed. $\RR$ contains the rewrite rules for evaluating
the arguments of DTs. Here we focus on simplifying $\SSS$ (thus $\DD$ and $\RR$ are fixed during the process)
but techniques to simplify $\DD$ and $\RR$ are available as well
\cite{DependencyTuple,ava:mos:16}.

\begin{example}[\exref{ex:chainTree} continued]
Our chain tree from \exref{ex:chainTree}
for the term
$\tup{s} = \FSIZE(\FTree(\FZero, \FNil, \FNil))$
has 4 nodes. Thus, we can conclude that
$\Cplx_{\langle \DT(\RR), \DT(\RR), \RR \rangle}(\tup{s}) \geq 4$.
\end{example}

The main correctness statement in the sequential case
summarises our earlier intuitions. It has a special
case for \emph{confluent} TRSs, for which Dependency Tuples
capture innermost runtime complexity exactly. The reason that
confluence (intuitively: results of computations are deterministic) is required is that without confluence, there can also be
chain trees that do not correspond to real computations and
lead to higher complexities; for an example, see
\cite[Example 11]{DependencyTuple}.

\begin{theorem}[$\Cplx$ bounds Derivation Height for $\ito$
\cite{DependencyTuple}]
\label{thm:cplxSeq}
Let $\RR$ be a TRS, let $t = f(t_1,\ldots,t_n) \in \TT(\Signature,\VV)$
such that all $t_i$ are in normal form
(this includes all $t \in \BasicTerms^{\RR}$). Then we have $\Dh(t,\ito) \leq
\Cplx_{\langle \DT(\RR), \DT(\RR), \RR \rangle}(\tup{t})$.
If $\ito$ is
confluent,\footnote{The proofs for \thmref{thm:cplxSeq} and
\thmref{thm:canonical_dt_problem} from the literature
and for our new \thmref{thm:cplxPar} and
\thmref{thm:canonical_pdt_problem} require only the
property that the used rewrite relation has
\emph{unique normal forms (w.r.t.\ reduction)}
instead of confluence.
However, to streamline presentation, we follow the literature
\cite{DependencyTuple}
and state our theorems with confluence rather than the
property of unique normal forms for the rewrite relation.
Note that confluence of a relation is a sufficient condition for
unique normal forms, and confluence coincides with the unique normal form
property if the relation is terminating \cite{BaaderNipkow}.
Additionally, there is more readily available
tool support for confluence analysis.}
then  $\Dh(t,\ito) =
\Cplx_{\langle \DT(\RR), \DT(\RR), \RR \rangle}(\tup{t})$.
\end{theorem}

\thmref{thm:cplxSeq} implies the following link between
$\ircR$ and
$\irc{\langle \DT(\RR), \DT(\RR), \RR \rangle}$\journal{, which
also explains why $\langle \DT(\RR), \DT(\RR), \RR \rangle$ is
called the ``canonical'' DT problem for $\RR$}:

\begin{theorem}[Complexity Bounds for TRSs via Canonical DT
Problems \cite{DependencyTuple}]
\label{thm:canonical_dt_problem}
Let $\RR$ be a TRS with canonical DT problem
$\langle \DT(\RR), \DT(\RR), \RR \rangle$.
Then we have $\ircR(n) \leq \irc{\langle \DT(\RR), \DT(\RR), \RR \rangle}(n)$.
If $\ito$ is confluent, we have
$\ircR(n) = \irc{\langle \DT(\RR), \DT(\RR), \RR \rangle}(n)$.
\end{theorem}

In practice, the focus
is on
finding
asymptotic bounds  for $\ircR$.
For example, \exref{ex:size:irc} will show that
for our TRS $\RR$
from \exref{ex:size1}
we have $\ircR(n) \in \OO(n^2)$.

A DT problem $\langle \DD, \SSS, \RR \rangle$ is said to be \emph{solved}
iff $\SSS = \emptyset$: we always have $\irc{\langle \DD, \emptyset, \RR \rangle}(n) = 0$.
To simplify and finally solve DT problems in an incremental fashion,
complexity analysis techniques called \emph{DT processors} are used. A
DT processor takes a DT problem as input and returns a (hopefully
simpler) DT problem as well as an asymptotic complexity bound as an
output. The largest asymptotic complexity bound returned over this
incremental process is then also an upper bound for
\journal{$\irc{\langle \DT(\RR), \DT(\RR), \RR \rangle}(n)$
and hence also }$\ircR(n)$
\cite[Corollary 21]{DependencyTuple}.
In all examples that we present in \autoref{sec:irc} and \autoref{sec:para_complex},
  a single proof step with a DT processor suffices to
  solve the given DT problem.
  Thus, the complexity bound that is found by this proof
  step is trivially the largest bound among all proof steps
  and directly carries over to the original problem as an
  asymptotic bound. In general, \emph{several}
  proof steps using potentially different DT processors may be
  needed to find an asymptotic complexity bound for the input
  TRS, and then an explicit check for the largest bound is needed.

For our examples in \autoref{sec:irc} and
  \autoref{sec:para_complex}, we use the
reduction pair processor using polynomial interpretations
\cite{DependencyTuple}.
This DT processor applies a restriction
of polynomial interpretations to $\nat$ \cite{Lankford75} to infer
upper bounds on the number of times that DTs can occur
in a chain tree for terms of size at most~$n$.

\begin{definition}[Polynomial Interpretation, CPI]
A \emph{polynomial interpretation} $\Pol$ maps
every $n$-ary function symbol to a polynomial with variables
$x_1,\ldots,x_n$ and coefficients from $\nat$.
$\Pol$ extends
to terms via $\Pol(x)=x$ for
$x \in \VV$ and $\Pol(f(t_1,\ldots,t_n)) =
\Pol(f)(\Pol(t_1),\ldots,\Pol(t_n))$.
$\Pol$ induces an order $\succ_\Pol$ and a quasi-order
$\succsim_\Pol$ over terms where
$s \succ_\Pol t$ iff $\Pol(s) > \Pol(t)$
and
$s \succsim_\Pol t$ iff $\Pol(s) \geq \Pol(t)$
for all instantiations of variables with natural numbers.

\medskip
A \emph{complexity polynomial interpretation (CPI)} $\Pol$
is a polynomial interpretation where:
\begin{itemize}
\item
$\Pol(\FCom_n(x_1,\ldots,x_n)) = x_1 + \dots + x_n$, and
\item
for all
$f \in \ConSyms$,
$\Pol(f(x_1,\ldots,x_n)) = a_1\cdot x_1 + \dots + a_n \cdot x_n + b$
for some $a_i \in \{0,1\}$ and $b \in \nat$.
\end{itemize}
\end{definition}

The restriction for CPIs regarding constructor symbols
enforces that the interpretation of a constructor term $t$ (as an argument
of a term for which a chain tree is constructed) can
exceed its size $\tsize{t}$ only by at most a constant factor.
This is crucial for soundness.
Using a CPI, we can now define and state correctness of the
corresponding reduction pair processor
\cite[Theorem 27]{DependencyTuple}.

\begin{theorem}[Reduction Pair Processor with CPIs
  \cite{DependencyTuple}]
\label{thm:redpair}
Let $P = \langle \DD, \SSS, \RR \rangle$ be a DT problem,
let $\succsim$ and $\succ$ be induced by a CPI $\Pol$.
Let $k \in \nat$ be the maximal degree of all polynomials
$\Pol(\tup{f})$ for all $f \in \DefSyms$.
Let $\DD \cup \RR \subseteq {\succsim}$.
If $\SSS \cap {\succ} \neq \emptyset$,
the reduction pair processor returns the DT problem
$P' = \langle \DD, \SSS \setminus{\succ}, \RR \rangle$
and the complexity $\OO(n^k)$.
Then the reduction pair processor is sound, i.e.,
the maximum of the asymptotic upper bound it computes
and the complexity of $P'$ is indeed
an asymptotic
upper bound on the complexity of its input $P$;
see also \cite[Definition 17]{DependencyTuple}.
\end{theorem}

\begin{example}[\exref{ex:sizeDTs} continued]
\label{ex:size:irc}
For our running example,
consider the CPI
$\Pol$
with:
$\Pol(\FPLUS(x_1,x_2)) = \Pol(\Fsize(x_1)) = x_1,
\linebreak
\Pol(\FSIZE(x_1)) = 2 x_1 + x_1^2,
\Pol(\Fplus(x_1,x_2)) = x_1 + x_2,
\Pol(\FTree(x_1,x_2,x_3)) = 1 + x_2 + x_3,
\Pol(\FS(x_1)) = 1 + x_1,
\Pol(\FZero) = \Pol(\FNil) = 1$.
$\Pol$ orients all DTs in
$\SSS = \DT(\RR)$ with $\succ$ and all
rules in $\RR$ with $\succsim$.
Thus, with this CPI, the reduction pair processor
returns the solved DT problem
$\langle \DT(\RR), \emptyset, \RR \rangle$
and proves
$\ircR(n) \in \OO(n^2)$:
since the maximal degree of the CPI for a symbol $\tup{f}$ is 2,
the upper bound of $\OO(n^2)$ follows by
\thmref{thm:redpair}.
Polynomial interpretations such as those used in the
examples in this paper can be found automatically
using parametric interpretation templates for each function symbol and
SAT- or SMT-solvers \cite{satPolo,smtPolo}
to find the parameter values.

For example, we might use a parametric interpretation
$\Pol_p$ with
$\Pol_p(\FPLUS(x_1,x_2)) = p_0 + p_1 x_1 + p_2 x_2 +
p_3 x_1^2 + p_4 x_1 x_2 + p_5 x_2^2$,
$\Pol_p(\FS(x_1)) = p_6 + p_7 x_1$,
and $\Pol_p(\FCom_1(x_1)) = x_1$.
Here all $p_i$ are parameters that range over
$\mathbb{N}$, with $p_7$ additionally being restricted
to $\{0,1\}$. For the term constraint
$\FPLUS(\FS(x), y) \succ \FCom_1(\FPLUS(x, y))$,
we would get the following parametric polynomial
constraint that must be satisfied for all
$x,y \in \mathbb{N}$:
\[
p_0 + p_1 (p_6 + p_7 x) + p_2 y +
p_3 (p_6 + p_7 x)^2 + p_4 (p_6 + p_7 x) y + p_5 y^2 >
p_0 + p_1 x + p_2 y +
p_3 x^2 + p_4 x y + p_5 y^2
\]
We simplify the expression and get:
\[
 p_1 p_6 + p_1 p_7 x +
p_3 p_6^2 + 2 p_3 p_6 p_7 x + p_3 p_7^2 x^2 +
p_4 p_6 y + p_4 p_7 x y >
p_1 x + p_3 x^2 + p_4 x y
\]
We group parametric
coefficients for each monomial in $x,y$ together.
This yields:
\[
 (p_1 p_6 + p_3 p_6^2) \; + \;
 (p_1 p_7 + 2 p_3 p_6 p_7 - p_1) x \;
 + \; (p_3 p_7^2 - p_3) x^2 \; + \;
p_4 p_6 y \; + \; (p_4 p_7 - p_4) x y > 0
\]
The \emph{absolute positiveness criterion} \cite{positiveness} allows us
to reduce this $\exists\forall$ problem
to an $\exists$ problem such that a solution
for the latter problem is also a solution for the
former problem:
\[
 p_1 p_6 + p_3 p_6^2 > 0 \: \land \:
 p_1 p_7 + 2 p_3 p_6 p_7 - p_1 \geq 0 \: \land \:
 p_3 p_7^2 - p_3 \geq 0 \: \land \:
 p_4 p_6 \geq 0 \: \land \:
 p_4 p_7 - p_4 \geq 0
\]
This problem can now be passed to a constraint solver
for non-linear integer arithmetic, e.g., based on
SAT- or SMT-solving \cite{satPolo,smtPolo}.
In this example, the constraint solver might return a solution
with $p_1 = p_6 = p_7 = 1$ and $p_i = 0$ otherwise.
This solution for our constraint system allows us to refine $\Pol_p$ to
(part of) the above CPI $\Pol$ by replacing the parameters $p_i$
with the values returned from the constraint solver.
The full CPI $\Pol$ is obtained by considering all
term constraints simultaneously and passing the constraint
system to an integer constraint solver; for details,
we refer to \cite{satPolo}.
\end{example}

\section{Finding upper bounds for parallel complexity}
\label{sec:para_complex}

In this section we present our main contribution: an application of the
DT framework from innermost runtime complexity to \emph{parallel-innermost rewriting}.

The notion of parallel-innermost rewriting dates back at least
to \cite{parallelRewriting}. Informally, in a parallel-innermost
rewrite step, all innermost redexes are rewritten simultaneously.
This corresponds to executing all function calls in parallel using a
call-by-value strategy on a machine with unbounded parallelism \cite{BlellochFPCA95}.
In the literature \cite{LoopsUnderStrategies}, this strategy is also known as
``max-parallel-innermost rewriting''.

\begin{definition}[Parallel-Innermost Rewriting \cite{innermostOrdering}]
\label{def:pito}
A term $s$ \emph{rewrites innermost in parallel} to $t$ with a TRS $\RR$,
written $s \pito t$, iff $s \itop t$, and either
(a) $s \ito t$ with $s$ an innermost redex, or
(b) $s = f(s_1, \ldots, s_n)$, $t = f(t_1, \ldots, t_n)$, and for all
$1 \leq k \leq n$ either $s_k \pito t_k$ or $s_k = t_k$ is a normal
form.\footnote{The use of ``\,$\maxparallel$\,'' in the notation
  $\pito$ was suggested by van Oostrom \cite{vvO} to avoid
  confusion with the notation $\oldpto$, which is commonly
  used to denote rewriting 0 or more eligible parallel
  redexes \cite{BaaderNipkow}.
  In contrast, in this paper we require rewriting \emph{all}
  eligible (here: innermost) redexes in a step with $\pito$,
  and at least one such redex must be rewritten in such a step.}
\end{definition}

As common in parallel rewriting, the redexes that are rewritten
are at \emph{parallel positions} -- that is, no position is a prefix
of another. The literature \cite{LoopsUnderStrategies}
also contains definitions of
parallel rewrite strategies where an arbitrary selection of one or more
(not necessarily innermost) parallel redexes may be replaced in the
parallel rewrite step rather than all innermost redexes (which
are always parallel).
In the innermost case, such ``may'' parallel rewriting
would include (sequential) innermost rewriting: rewrite only one innermost
redex at a time. Thus, also the worst-case time complexity of ``may'' parallel
rewriting would be identical to that of sequential rewriting.
To capture the possible speed improvements enabled by parallel
rewriting on a (fictitious) machine with unbounded parallelism,
we have chosen a ``must'' parallel rewriting strategy where
\emph{all}
eligible
redexes (here: innermost) must be rewritten
(case (b) of \defref{def:pito} rewrites \emph{all}
arguments that are not in normal form).

\begin{example}[\exref{ex:size1} continued]
\label{ex:pito_sequence}
The TRS $\RR$ from \exref{ex:size1}
allows the following parallel-innermost rewrite sequence, where
innermost redexes are colourised in an analogous way to
\exref{ex:size1}:
\[
\begin{array}{rl}
&\colorbox{colorF}{\usebox\sizeoftwohightree}\\
\pito & \FS(\Fplus(\colorbox{colorB}{\usebox\sizeofnil}, \colorbox{colorE}{\usebox\sizeofonehightree}))\\
\pito & \FS(\Fplus(\FZero, \FS(\Fplus(\colorbox{colorD}{\usebox\sizeofnil}, \colorbox{colorA}{\usebox\sizeofnil}))))\\
\pito & \FS(\Fplus(\FZero, \FS(\colorbox{colorE}{\usebox\pluszerozero})))\\
\pito & \FS(\colorbox{colorF}{\usebox\pluszeroone})\\
\pito & \FS(\FS(\FZero))
\end{array}
\]

In the second and in the third step, two innermost steps
happen in parallel (which is not possible with standard
innermost rewriting: $\mathrel{\pito}\ \not\subseteq\ \mathrel{\ito}$).
An innermost rewrite sequence without parallel evaluation requires
two more steps to reach a normal form from this start term,
as shown in \exref{ex:size1}.
Note that switching from $\ito$ to $\pito$ in general
does not lead to such a ``speed-up'':
as we shall see in \thmref{thm:nopar},
the derivation height of a term does \emph{not} necessarily
decrease when $\pito$ is used instead of $\ito$.
\end{example}

Note that for all TRSs $\RR$, $\pito$ is terminating
iff $\ito$ is terminating \cite{innermostOrdering}.
\exref{ex:pito_sequence}
shows that
such an equivalence does \emph{not} hold for the derivation height
of a term.

The question now is: given a TRS $\RR$, how much of a speed-up might
we get by
a switch
from innermost to parallel-innermost rewriting?
To investigate,
we extend the notion of innermost
runtime complexity to parallel-innermost rewriting.

\begin{definition}[Parallel-Innermost Runtime Complexity $\pirc{}$]
\label{def:pirc}
For $n \in \mathbb{N}$, we define the
\emph{parallel-innermost runtime complexity} function
as \journal{the maximum of all derivations heights
of parallel executions from basic terms of size at most $n$:}
\[
\pircR(n) = \sup \{ \Dh(t, {\pito}) \mid t \in \BasicTerms^{\RR}, \tsize{t} \leq
n \}.
\]
\end{definition}

In the literature on parallel computing
\cite{BlellochFPCA95,HoffmannESOP15,BaillotPiESOP21},
the terms \emph{depth} or
\emph{span} are commonly used for the concept of the runtime
of a function on a machine with unbounded parallelism (``wall time''), corresponding
to the complexity measure of $\pircR$.
In contrast, $\ircR$ would describe the \emph{work} of a function
(``CPU time'').

In the following, given a TRS $\RR$, our goal shall be to infer
(asymptotic) upper bounds for $\pircR$ fully automatically.
Of course, an upper bound for (sequential) $\ircR$
is also an upper bound for $\pircR$.
We will now introduce
techniques to find upper bounds for $\pircR$
that are strictly tighter than these trivial bounds.

To find upper bounds for runtime complexity of parallel-innermost
rewriting, we can \emph{reuse} the notion of DTs
from \defref{def:dt} for sequential innermost rewriting along with
existing techniques~\cite{DependencyTuple} as illustrated in the following example.

\begin{example}
\label{ex:sizePDTs}
In the recursive $\Fsize$-rule, the two calls to $\Fsize(l)$ and
$\Fsize(r)$ happen \emph{in parallel} (they are \emph{structurally
independent})
and take place at \emph{parallel positions} in the term.
Thus, the cost (number of rewrite steps with $\pito$ until a normal form is reached) for these two
calls is not the \emph{sum}, but the \emph{maximum} of their
individual costs.
Regardless of which of these two calls has the higher cost,
we still need to add the cost for the call to $\Fplus$
on the results of the two calls:
$\Fplus$ starts evaluating only after both calls to $\Fsize$ have
finished, or equivalently, $\Fsize$ calls \emph{happen before} $\Fplus$.
With $\sigma$ as the used matcher for the rule
and with $t \downarrow$ as the (here unique)
normal form resulting from repeatedly rewriting a term $t$
with $\pito$ (the ``result'' of evaluating $t$), we have:
\[
\begin{array}{rl}
&\Dh(\Fsize(\FTree(v, l, r)) \sigma, \pito)\\
=&
1 + \max( \Dh(\Fsize(l) \sigma, \pito), \Dh(\Fsize(r) \sigma, \pito)) + \Dh(\Fplus(\Fsize(l) \sigma \!\downarrow, \Fsize(r) \sigma \!\downarrow), \pito)
\end{array}
\]
In the DT setting, we could
introduce a new symbol $\FComPar_n$
that explicitly expresses that its arguments
are evaluated in parallel. This symbol would
then be interpreted as the maximum
of its arguments in an extension of \thmref{thm:redpair}:
\begin{align*}
\FSIZE(\FTree(v, l, r)) & \to
   \FCom_2(\FComPar_2(\FSIZE(l), \FSIZE(r)), \FPLUS(\Fsize(l), \Fsize(r)))
\end{align*}
Although automation of the search for
polynomial interpretations extended by the maximum function is
readily available~\cite{maxpolo},
we would still have to extend the notion of
Dependency Tuples and also adapt all existing techniques
in the Dependency Tuple framework to work with $\FComPar_n$.

This is why we have chosen the following alternative approach,
which is equally powerful on theoretical level and
enables immediate reuse of
all techniques in the existing DT
framework \cite{DependencyTuple}.
Equivalently to the above, we can ``factor in'' the cost of
calling $\Fplus$ into the maximum function:
\[
\begin{array}{rl}
&\Dh(\Fsize(\FTree(v, l, r)) \sigma, \pito)\\
=&
\max( 1 + \Dh(\Fsize(l) \sigma, \pito) + \Dh(\Fplus(\Fsize(l) \sigma \!\downarrow, \Fsize(r) \sigma \!\downarrow), \pito),\\
&\phantom{\max(}
1 + \Dh(\Fsize(r) \sigma, \pito) +
\Dh(\Fplus(\Fsize(l) \sigma \!\downarrow, \Fsize(r) \sigma \!\downarrow), \pito))
\end{array}
\]
Intuitively, this would correspond to evaluating
$\Fplus(\ldots, \ldots)$
twice, in two parallel threads of execution,
which costs the same amount of (wall) time as evaluating
$\Fplus(\ldots, \ldots)$
once.
We can represent this maximum of the execution times of two threads
by introducing \emph{two} DTs for our recursive $\Fsize$-rule:
\[
\begin{array}{rcl}
\FSIZE(\FTree(v, l, r)) & \to &
   \FCom_2(\FSIZE(l), \FPLUS(\Fsize(l), \Fsize(r)))\\
\FSIZE(\FTree(v, l, r)) & \to &
   \FCom_2(\FSIZE(r), \FPLUS(\Fsize(l), \Fsize(r)))
\end{array}
\]
To express the cost of a concrete rewrite sequence, we would
non-deterministically choose the DT that
corresponds to the ``slower thread''.
\end{example}

In other words,
when a rule $\ell \to r$ is used,
the cost of the function call to the instance of $\ell$
is $1$ $+$ the sum
of the costs of the function calls in the resulting instance of $r$ \emph{that are in structural dependency with each other}.
The actual cost of the function call to the instance of $\ell$ in a concrete
rewrite sequence is the \emph{maximum} of all the possible costs
caused by such \emph{chains} of structural dependency
(based on the prefix order $>$ on positions of defined
function symbols in $r$).
Thus, \emph{structurally independent} function calls are considered in
separate DTs, whose non-determinism models the parallelism
of these function calls.

The notion of \emph{structural dependency} of function calls is captured by
\defref{def:structdeps}. Basically, it comes from the fact that a term
cannot be evaluated before all its subterms have been reduced to normal forms
(innermost rewriting/\emph{call by value}).
This induces a ``happens-before'' relation for the computation
\cite{Lamport78}.

\begin{definition}[Structural Dependency, Maximal Structural Dependency Chain Set $\msdcSym$]
  \label{def:structdeps}
For positions $\pi_1, \ldots, \pi_k$, we call $\lmsdc \pi_1, \dots, \pi_k \rmsdc$
a \emph{structural dependency chain} for a term $t$ iff
$\pi_1,\! \ldots,$  $\pi_k \in \PosDef(t)$ and $\pi_1 > \ldots > \pi_k$.
Here $\pi_i$ \emph{structurally depends on} $\pi_j$ in $t$ iff $i > j$.
A structural dependency chain $\lmsdc \pi_1, \dots, \pi_k \rmsdc$
for a term $t$ is \emph{maximal} iff
$k = 0$ and $\PosDef(t) = \emptyset$, or $k > 0$ and
$\forall \pi \in \PosDef(t), \big(\pi \ngtr \pi_1 \wedge (\pi_1 > \pi
\Rightarrow \pi \in \{\pi_2, \ldots, \pi_k\})\big)$.
We write $\msdcSym(t)$ for the set of all maximal structural
dependency chains for $t$.
\end{definition}

In the formula specifying \emph{maximal} structural
dependency chains for $k > 0$, the first conjunct states that
the first element $\pi_1$ must be the position of an innermost
defined symbol, and the second conjunct states that all
positions of defined symbols above $\pi_1$ must be part of the
chain as well.
In other words, it is not possible to add further elements anywhere
in a maximal structural dependency chain and obtain a larger
structural dependency chain.
Note that $\msdcSym(t) \neq \emptyset$ always holds:
if $\PosDef(t) = \emptyset$, then
$\msdcSym(t) = \{ \lmsdc \rmsdc \}$.
Note also that since we consider maximality of structural dependency
chains w.r.t.~subset inclusion rather than cardinality of the sets of
their elements, $\msdcSym(t)$ may contain
structural dependency chains with different numbers of elements,
as we shall now see in \exref{ex:chains}.

\begin{example}
\label{ex:chains}
Let $t = \FS(\Fplus(\Fsize(\FNil), \Fplus(\Fsize(x),\FZero)))$.
In our running example,
$t$ has the
 following
structural dependencies:
$\msdcSym(t) = \{ \lmsdc 1.1, 1 \rmsdc, \lmsdc 1.2.1, 1.2, 1 \rmsdc \}$.
The
chain $\lmsdc 1.1, 1 \rmsdc$
corresponds to the nesting of
$t|_{1.1}=\Fsize(\FNil)$ below $t|_{1}=\Fplus(\Fsize(\FNil),
\Fplus(\Fsize(x),\FZero))$, so the evaluation of
$t|_{1}$ will have to wait at least until $t|_{1.1}$ has been
fully evaluated.

If
$\pi$ structurally depends on $\tau$ in a
term $t$,
neither $t|_\tau$ nor $t|_\pi$ need to \emph{be} a redex.
Rather, $t|_\tau$ could be \emph{instantiated} to a
redex and
an instance of $t|_\pi$ could become a redex after
its subterms, including the instance of $t|_\tau$, have been
evaluated.
\end{example}

We thus revisit the notion of DTs
as \emph{Parallel Dependency Tuples}, which now embed
structural
dependencies in addition to the algorithmic dependencies
already captured in DTs.

\begin{definition}[Parallel Dependency Tuples $\DTpar$,
Canonical Parallel DT Problem]
\label{def:pdt}
For a rewrite rule $\ell \to r$, we define the set of its
\emph{Parallel Dependency Tuples (PDTs)} $\DTpar(\ell \to r)$:
$\DTpar(\ell \to r) = \{ \tup{\ell} \to \FCom_k(\tup{r|}_{\pi_1},\ldots,\tup{r|}_{\pi_k})
 \mid
 \msdc{r}{\pi_1,\ldots, \pi_k} \}$.
For a TRS $\RR$, let $\DTpar(\RR) = \bigcup_{\ell \to r \in \RR}
\DTpar(\ell \to r)$.

The \emph{canonical parallel DT problem} for $\RR$ is
$\langle \DTpar(\RR), \DTpar(\RR), \RR \rangle$.
\end{definition}

\begin{example}
For our recursive $\Fsize$-rule
$\ell \to r$,
we have
$\PosDef(r) = \{ 1, 1.1, 1.2 \}$
and
$\msdcSym(r) = \{ \lmsdc 1.1, 1 \rmsdc, \lmsdc 1.2, 1 \rmsdc \}$.
With $\mathit{r}|_1 = \Fplus(\Fsize(l), \Fsize(r))$,
$\mathit{r}|_{1.1} = \Fsize(l)$, and $\mathit{r}|_{1.2} = \Fsize(r)$,
we get the PDTs
from \exref{ex:sizePDTs}. For
the rule $\Fsize(\FNil) \to \FZero$,
we have $\msdcSym(\FZero) = \{ \lmsdc \rmsdc \}$,
so we get
$\DTpar(\Fsize(\FNil) \to \FZero) = \{ \FSIZE(\FNil) \to \FCom_0 \}$.
\end{example}

Our goal is now to prove that with the canonical
  PDT problem for $\RR$ as a starting point, we can
  \emph{reuse} the existing Dependency Tuple Framework
  to find bounds on \emph{parallel-}innermost runtime complexity,
  even though the DT Framework was originally introduced
  only with sequential innermost rewriting in mind.
  This allows for reuse both on theory level and on
  implementation level.
  A crucial step towards this goal is
our main correctness statement:

\begin{theorem}[$\Cplx$ bounds Derivation Height for $\pito$]
\label{thm:cplxPar}
Let $\RR$ be a TRS, let $t = f(t_1,\ldots,t_n) \in \TT(\Signature,\VV)$
such that all $t_i$ are in normal form
(e.g., when $t \in \BasicTerms^{\RR}$). Then we have $\Dh(t,\pito) \leq
\Cplx_{\langle \DTpar(\RR), \DTpar(\RR), \RR \rangle}(\tup{t})$.

If $\pito$ is confluent, then  $\Dh(t,\pito) =
\Cplx_{\langle \DTpar(\RR), \DTpar(\RR), \RR \rangle}(\tup{t})$.
\end{theorem}

To prove \thmref{thm:cplxPar}, we need some further definitions
and lemmas.
Intuitively, the notion of
\emph{maximal parallel argument normal form} of a term
$t$
captures the result of reducing all
its arguments
to such normal forms
that
the number of $\pito$ steps at root position
will have
``worst-case cost''.

\begin{definition}[Argument Normal Form \cite{DependencyTuple}, Maximal Parallel Argument Normal Form]
A term $t$ is an \emph{argument normal form} iff $t \in \VV$ or
$t = f(t_1,\ldots,t_n)$ and all $t_i$ are in normal form.
A term $\maxanf{t}$ is a \emph{maximal parallel argument normal form} of a
term $t$ iff $\maxanf{t}$ is an argument normal form such that
$t \pitosbelow \maxanf{t}$ and for all argument normal forms
$t'$ with $t \pitosbelow t'$, we have $\Dh(t',\pito) \leq \Dh(\maxanf{t},\pito)$. Here $u \pitosbelow v$ denotes a rewrite sequence
with $\pito$ where all steps are at positions $> \varepsilon$.
\end{definition}

\begin{example}
  Consider the TRS $\{ \Fa \to \Fb, \Fa \to \Fc, \Ff(\Fb) \to \Fd, \Ff(\Fc) \to \Ff(\Fd), \Ff(\Fd) \to \Fd \}$.
  The terms $\Ff(\Fb)$ and $\Ff(\Fc)$ are both in argument normal form,
  and we have $\Ff(\Fa) \pitosbelow \Ff(\Fb)$ and $\Ff(\Fa) \pitosbelow \Ff(\Fc)$.
  The term $\Ff(\Fc)$ is a maximal parallel argument normal form of $\Ff(\Fa)$
  because for any other term $t$ that satisfies both of these criteria
  (argument normal form and $\Ff(\Fa) \pitosbelow t$), we always have
  $\Dh(t, \pito) \leq \Dh(\Ff(\Fc), \pito)$.
\end{example}

The following lemma is adapted to the parallel
setting from \cite{DependencyTuple}.

\begin{lemma}[Parallel Derivation Heights of Nested Subterms]
\label{lem:nested}
Let $t$ be a term, let $\RR$ be a TRS such that all
reductions of $t$ with $\pito$ are finite. Then
\begin{align*}
\Dh(t,\pito) &\leq
\max
\{
 \sum_{1 \leq i \leq k} \Dh(\maxanf{t|_{\pi_i}},\pito)
 \mid  \msdc{t}{\pi_1, \dots, \pi_k}
\}
\end{align*}

If $\pito$ is confluent, then we additionally have:
\begin{align*}
\Dh(t,\pito) &=
\max
\{
 \sum_{1 \leq i \leq k} \Dh(\maxanf{t|_{\pi_i}},\pito)
 \mid
 \msdc{t}{\pi_1, \dots, \pi_k}
\}
\end{align*}
\end{lemma}

\begin{proof}
By induction on the term size $|t|$.
If $|t| = 1$, the statement follows immediately
since $\maxanf{t}\ = t$.
Now consider the case $|t| > 1$. Let $n$ be the arity
of the root symbol of $t$. In (parallel-)innermost
rewriting, a rewrite step at the root of $t$ requires
that the arguments of $t$ have been rewritten to normal
forms. Since rewriting of arguments takes place in parallel
(case (b) of \defref{def:pito} applies), we have
\[
\Dh(t, \pito) \leq \Dh(\maxanf{t}, \pito) +
   \max \{\quad \Dh(t|_j, \pito) \quad \mid 1 \leq j \leq n \}
\]
If $\pito$ is confluent, then $\maxanf{t}$ is uniquely determined and
we have equality in the previous as well as in the next (in)equalities.

\medskip
As $|t_j| < |t|$, we can apply the induction hypothesis:
\begin{align*}
\Dh(t, \pito) &\leq \Dh(\maxanf{t}, \pito) \\
&+ \max \left\{
  \max \left\{
         \sum_{1 \leq i \leq m} \Dh(\maxanf{t|_{j.\tau_i}},\pito)
        \, \middle|\, \msdc{t|_j}{\tau_1, \dots, \tau_m}
  \right\}
\,\middle|\,  1 \leq j \leq n \right\}
\end{align*}
Equivalently:
\begin{align*}
\Dh(t, \pito) &\leq
\max \left\{
  \Dh(\maxanf{t}, \pito) +
  \sum_{1 \leq i \leq m} \Dh(\maxanf{t|_{j.\tau_i}},\pito)
  ~\middle|~
  \begin{gathered}
    1 \leq j \leq n \\
    \msdc{t|_j}{\tau_1, \dots, \tau_m}
  \end{gathered}
\right\} \\
&=
\max \left\{
  \sum_{1 \leq i \leq k} \Dh(\maxanf{t|_{\pi_i}},\pito)
  ~\middle|~ \msdc{t}{\pi_1, \dots, \pi_k}
\right\}
\end{align*}

For the last equality, consider that
the maximal structural dependency chains $\pi_1, \dots, \pi_k$ of
$t$ can have two forms. If the root of $t$ is a defined symbol,
we have $\msdc{t}{j.\tau_1, \dots, j.\tau_m, \varepsilon}$.
Otherwise $\Dh(\maxanf{t}, \pito) = 0$ and thus
$\msdc{t}{j.\tau_1, \dots, j.\tau_m}$.
\end{proof}

We now can proceed with the proof of~\thmref{thm:cplxPar}.
\begin{proof}[of~\thmref{thm:cplxPar}]
  \begin{itemize}
  \item   As the first case, consider $\Dh(t, \pito) = \omega$.
Since $t$ is in argument normal form, the first rewrite step
from $t$ must occur at the root. Thus, there are
$\ell_1 \to r_1 \in \RR$ and a substitution $\sigma_1$
such that $t = \ell_1\sigma_1 \pito r_1\sigma_1$ and
$\Dh(r_1\sigma_1,\pito) = \omega$. Hence, there is a minimal subterm
$r_1\sigma_1|_{\pi_1}$ of $r_1\sigma_1$ such that
$\Dh(r_1\sigma_1|_{\pi_1},\pito) = \omega$ and all proper subterms of
$r_1\sigma_1|_{\pi_1}$ terminate w.r.t.\ $\pito$. As $\sigma_1$
must instantiate all variables with normal forms, we have
$\pi_1 \in \PosDef(r_1)$, i.e., $r_1\sigma_1|_{\pi_1} =
r_1|_{\pi_1}\sigma_1$.
In the infinite $\pito$-reduction of $r_1|_{\pi_1}\sigma_1$,
all arguments are again reduced to normal forms first, and
we get a term $t'$ with $\Dh(t', \pito) = \omega$.
Since $t'$ is in argument normal form, the first rewrite step from
$t'$ must occur at the root. Thus, there are
$\ell_2 \to r_2 \in \RR$ and a substitution $\sigma_2$
such that $t' = \ell_2\sigma_2 \pito r_2\sigma_2$ and
$\Dh(r_2\sigma_2,\pito) = \omega$.
This argument can be continued \emph{ad infinitum}, giving rise to an
infinite path of chain tree nodes
\[
(\tup{\ell}_1 \to \FCom_{n_1}(\ldots,\tup{r_1|}_{\pi_1},\ldots) \mid
\sigma_1), \qquad
(\tup{\ell}_2 \to \FCom_{n_2}(\ldots,\tup{r_2|}_{\pi_2},\ldots) \mid
\sigma_2), \qquad \dots
\]
Thus, $\tup{\ell}_1\sigma_1 = \tup{t}$ has an infinite chain
tree, and
$\Cplx_{\langle \DTpar(\RR), \DTpar(\RR), \RR \rangle}(\tup{t}) =
\omega$.

\item  Now consider the case where $\Dh(t,\pito) \in \nat$.
We use induction on $\Dh(t,\pito)$.

If $\Dh(t,\pito) = 0$,
the term $t$ is in normal form w.r.t.\ $\RR$. Thus, $\tup{t}$ is in
normal form w.r.t.\ $\DTpar(\RR) \cup \RR$, and
$\Cplx_{\langle \DTpar(\RR), \DTpar(\RR), \RR \rangle}(\tup{t}) = 0$.

If $\Dh(t,\pito) > 0$, since $t$ is in argument normal form, there are
$\ell \to r \in \RR$ and a substitution $\sigma$
such that $t = \ell\sigma \pito r\sigma = u$ and
\begin{align}
\Dh(t, \pito) &= 1 + \Dh(u, \pito) \label{eq:topstep}
\end{align}
As $\sigma$
must instantiate all variables with normal forms, we have
that $u|_\pi = r\sigma|_\pi$ is in normal form for all
$\pi \in \PosDef(u) \setminus \PosDef(r)$. For these positions $\pi$,
$\maxanf{u|_\pi} = u|_\pi$ and $\Dh(u|_\pi, \pito) = 0$.
From \lemref{lem:nested}, we get:
\begin{align}
&\quad\:\Dh(u,\pito) \notag\\
& \leq
\max
\left\{
 \sum_{1 \leq i \leq k} \hspace*{-1ex} \Dh(\maxanf{u|_{\pi_i}},\pito)
 ~\middle|~ \msdc{u}{\pi_1, \dots, \pi_k}
\right\}\notag\\[1.5ex]
& =
\max
\left\{
 \sum_{1 \leq i \leq j} \hspace*{-1ex} \Dh(\maxanf{u|_{\pi_i}},\pito) +
 \hspace*{-2ex}
 \sum_{j+1 \leq i \leq k}
 \hspace*{-2ex}
 \Dh(\maxanf{u|_{\pi_i}},\pito)
 ~\middle|~
  \begin{gathered}
    \msdc{u}{\pi_1, \dots, \pi_k} \\
    \pi_1, \dots, \pi_j \in \PosDef(u) \setminus \PosDef(r) \\
    \pi_{j+1}, \dots, \pi_k \in \PosDef(r)
  \end{gathered}
\right\}\notag\\[1.5ex]
& =
\max
\left\{
 \sum_{1 \leq i \leq j} \hspace*{-1ex} \Dh(u|_{\pi_i},\pito) \,\:\; +
 \hspace*{-2ex}
 \sum_{j+1 \leq i \leq k}
 \hspace*{-2ex}
 \Dh(\maxanf{u|_{\pi_i}},\pito)
 ~\middle|~
  \begin{gathered}
    \msdc{u}{\pi_1, \dots, \pi_k} \\
    \pi_1, \dots, \pi_j \in \PosDef(u) \setminus \PosDef(r) \\
    \pi_{j+1}, \dots, \pi_k \in  \PosDef(r)
 \end{gathered}
\right\}\notag\\[1.5ex]
& =
\max
\left\{ \hspace*{4ex}
 \sum_{j+1 \leq i \leq k} \Dh(\maxanf{u|_{\pi_i}},\pito)
 ~\middle|~
  \begin{gathered}
    \msdc{u}{\pi_1, \dots, \pi_k} \\
    \pi_1, \dots, \pi_j \in \PosDef(u) \setminus \PosDef(r) \\
    \pi_{j+1}, \dots, \pi_k \in  \PosDef(r)
  \end{gathered}
\right\}\notag
\label{eq:maxanf}\\[-2ex]
& \hspace*{25ex}
\end{align}
Note that $\Dh(\maxanf{u|_{\pi}},\pito) \leq
\Dh(u|_{\pi},\pito) < \Dh(t,\pito)$ holds for all $\pi \in
\PosDef(r)$. Thus, with the induction hypothesis,
\eqref{eq:topstep} and \eqref{eq:maxanf}, we get:

\eject

\hbox{}
\vspace*{-15mm}

\begin{align}
&\quad\:\Dh(t, \pito)\notag\\
& = 1 + \Dh(u, \pito)\notag\\
& \leq
1 +
\max
\left\{
 \sum_{j+1 \leq i \leq k}
  \hspace*{-1ex}
  \Cplx_{\langle \DTpar(\RR), \DTpar(\RR), \RR
                                       \rangle}(\tup{\maxanf{u|_{\pi_i}}})
  ~\middle|~
  \begin{gathered}
    \msdc{u}{\pi_1, \dots, \pi_k} \\
    \pi_1, \dots, \pi_j \in \PosDef(u) \setminus \PosDef(r) \\
    \pi_{j+1}, \dots, \pi_k \in  \PosDef(r)
  \end{gathered}
\right\}
\label{eq:ind}
\end{align}
Let $\lmsdc \pi_1,\ldots,\pi_k \rmsdc$ be an arbitrary maximal structural
dependency chain for $r$. Then there exists a corresponding
chain tree for $\tup{t}$ whose root node is
$(\tup{\ell} \to \FCom_{k}(\tup{\maxanf{r_1|_{\pi_1}}},\ldots,
   \tup{\maxanf{r_k|_{\pi_k}}}) \mid \sigma)$
and where the children of the root node are maximal
chain trees for
$\tup{\maxanf{u|_{\pi_1}}},\ldots,\tup{\maxanf{u|_{\pi_k}}}$.
This follows because for all $1 \leq i \leq k$, we have
$r|_{\pi_i}\sigma = u|_{\pi_i}$ and so
$\tup{r|_{\pi_i}}\sigma \itos \tup{\maxanf{u|_{\pi_i}}}$.
Together with \eqref{eq:ind}, this gives
$\Dh(t,\pito) \leq
\Cplx_{\langle \DTpar(\RR), \DTpar(\RR), \RR \rangle}(\tup{t})$,
and for confluent $\pito$ we also get
$\Dh(t,\pito) =
\Cplx_{\langle \DTpar(\RR), \DTpar(\RR), \RR \rangle}(\tup{t})$.
\end{itemize}
\end{proof}

From \thmref{thm:cplxPar}, the soundness of our approach to parallel
complexity analysis via the DT framework follows analogously to
\cite{DependencyTuple}:

\begin{theorem}[Parallel Complexity Bounds for TRSs via Canonical Parallel DT
Problems]
\label{thm:canonical_pdt_problem}
Let $\RR$ be a TRS with canonical parallel DT problem
$\langle \DTpar(\RR), \DTpar(\RR), \RR \rangle$.
Then we have $\pircR(n) \leq \irc{\langle \DTpar(\RR), \DTpar(\RR), \RR \rangle}(n)$.

If $\pito$ is confluent, we have
$\pircR(n) = \irc{\langle \DTpar(\RR), \DTpar(\RR), \RR \rangle}(n)$.
\end{theorem}

This theorem implies that we can reuse arbitrary techniques to find
upper bounds
for \emph{sequential} complexity in the
DT framework also to find upper bounds for \emph{parallel}
complexity, without requiring any modification to the framework.
To analyse parallel complexity of a TRS $\RR$
  instead of sequential
  complexity, we need to make only a single adjustment:
  the input for the DT framework is now the
  canonical \emph{parallel} DT problem for $\RR$
  (\defref{def:pdt})
  instead of the canonical DT problem
  (\defref{def:dt_problem}).

Specifically, with \thmref{thm:canonical_pdt_problem}
we can use
the existing reduction pair processor with CPIs
(\thmref{thm:redpair}) in the
DT framework
to get upper bounds for $\pircR$.

\begin{example}[\exref{ex:sizePDTs} continued]
\label{ex:size_tight_bound}
For our TRS $\RR$ computing the $\Fsize$ function on trees, we get the
set $\DTpar(\RR)$ with the following
PDTs:
\[
\begin{array}{rcl}
\FPLUS(\FZero, y) & \to &  \FCom_0 \\
\FPLUS(\FS(x), y) & \to & \FCom_1(\FPLUS(x, y)) \\
\FSIZE(\FNil) & \to & \FCom_0 \\
\FSIZE(\FTree(v, l, r)) & \to &
   \FCom_2(\FSIZE(l), \FPLUS(\Fsize(l), \Fsize(r))) \\
\FSIZE(\FTree(v, l, r)) & \to &
   \FCom_2(\FSIZE(r), \FPLUS(\Fsize(l), \Fsize(r)))
\end{array}
\]
The
interpretation $\Pol$
from \exref{ex:size:irc} implies
$\pircR(n) \in \OO(n^2)$.
This bound is tight: consider $\Fsize(t)$ for a comb-shaped tree $t$
where the first argument of $\FTree$ is always $\FZero$ and the
third is always $\FNil$.
The function $\Fplus$, which needs time linear in its first argument,
is called linearly often on data linear in the size of the start term.
Due to the structural dependencies, these calls do not happen in parallel
(so call $k+1$ to $\Fplus$ must wait for call $k$).
\end{example}

\begin{example}
\label{ex:doubles}
Note that $\pircR(n)$ can be asymptotically lower than $\ircR(n)$, for instance for the TRS $\RR$ of Section~\ref{sec:example} with the following rules:
\[
\begin{array}{rcl@{\hspace*{5ex}}|@{\hspace*{5ex}}rcl}
\Fdoubles(\FZero) & \to & \FNil &
  \Fd(\FZero) & \to & \FZero \\
\Fdoubles(\FS(x)) & \to & \FCons(\Fd(\FS(x)), \Fdoubles(x)) &
  \Fd(\FS(x)) & \to & \FS(\FS(\Fd(x)))
\end{array}
\]
The upper bound $\ircR(n) \in \mathcal{O}(n^2)$ is tight:
from
$\Fdoubles(\FS(\FS(\ldots\FS(\FZero)\ldots)))$, we get
linearly many calls to the linear-time function $\Fd$ on arguments
of size linear in the start term.
However, the Parallel Dependency Tuples in this example are:
\[
\begin{array}{rcl@{\hspace*{5ex}}|@{\hspace*{5ex}}rcl}
\FDOUBLES(\FZero) & \to & \FCom_0 & \FD(\FZero) & \to & \FCom_0 \\
\FDOUBLES(\FS(x)) & \to & \FCom_1(\FD(\FS(x))) & \FD(\FS(x)) & \to & \FCom_1(\FD(x)) \\
\FDOUBLES(\FS(x)) & \to & \FCom_1(\FDOUBLES(x))
\end{array}
\]
Then the following polynomial
interpretation, which orients all DTs with $\succ$ and all
rules from $\RR$ with $\succsim$ \journal{so that the reduction
pair processor returns a solved DT problem},
proves $\pircR(n) \in \mathcal{O}(n)$:
$\Pol(\FDOUBLES(x_1)) = \Pol(\Fd(x_1)) = 2 x_1,
\Pol(\FD(x_1)) = x_1,
\Pol(\Fdoubles(x_1)) = \Pol(\FZero) =\linebreak
\Pol(\FCons(x_1,x_2)) = \Pol(\FNil) = 1,
\Pol(\FS(x_1)) = 1 + x_1$.
\end{example}

Interestingly enough, Parallel Dependency Tuples also allow us to
identify TRSs that have \emph{no} potential for parallelisation
by parallel-innermost rewriting.

\begin{theorem}[Absence of Parallelism by PDTs]
\label{thm:nopar}
Let $\RR$ be a TRS
such that for all rules $\ell \to r \in \RR$,
$|\msdcSym(r)| = 1$.
Then:
\begin{enumerate}
\item[(a)] $\DTpar(\RR) = \DT(\RR)$;
\item[(b)] for all basic terms $t_0$ and rewrite sequences
$t_0 \pito t_1 \pito t_2 \pito \dots$, also
$t_0 \ito  t_1 \ito  t_2 \ito  \dots$ holds (i.e., from basic terms,
$\pito$ and $\ito$ coincide);
\item[(c)] $\pirc{\RR}(n) = \irc{\RR}(n)$.
\end{enumerate}
\end{theorem}
\begin{proof}
Let $\RR$ be a TRS such that for all rules $\ell \to r \in \RR$,
$|\msdcSym(r)| = 1$.

\medskip
We prove part (a) by showing that for each rule $\ell \to r \in \RR$,
we have $\DTpar(\ell \to r) = \{ \DT(\ell \to r) \}$.
Let $\ell \to r \in \RR$.
By construction of $\DTpar$, we get
from $|\msdcSym(r)| = 1$ that $|\DTpar(\ell \to r)| = 1$.
$|\msdcSym(r)| = 1$ implies that $\PosDef(r)$ is ordered by
the prefix order $>$ on positions.
Thus, by using
the extension of $>$
to
the lexicographic order on positions $\gtrdot$ used as an ingredient for the
construction of $\DT(\ell \to r)$, we obtain the result for
part (a).

\medskip
We now prove part (b).
Let $t_0$ be a basic term for $\RR$ with a rewrite sequence
$t_0 \pito t_1 \pito t_2 \pito \dots$.
We show by induction over $i$ that for all $t_i$,
$t_i$ contains at most one innermost redex.

For the base case, consider that the basic term $t_0$
contains only a single occurrence of a defined symbol,
at the root.
Thus, if $t_0$ is a redex, it is also the unique innermost
redex in $t_0$.

For the induction step, assume that $t_i$ has at most one innermost redex.
If $t_i$ has no redex, it is a normal form, and we are done.
Otherwise, $t_i$ has exactly one innermost redex at position $\tau$,
and in the parallel-innermost rewrite step $t_i \pito t_{i+1}$
a rule $\ell \to r$ with matcher $\sigma$ replaces $t_i|_\tau = \sigma(\ell)$
by $\sigma(r)$.
The premise $|\msdcSym(r)| = 1$ implies that there is exactly one
(empty or non-empty) maximal structural dependency chain
$\msdc{r}{\pi_1, \ldots, \pi_k}$.

Since the rewrite step $t_i \pito t_{i+1}$ uses
(parallel-)\emph{innermost} rewriting, $\sigma(x)$
is in normal form for all variables $x$.
Thus, potential redexes in term $t_{i+1}$ can only be at positions
$\tau . \pi_1, \ldots, \tau . \pi_k$.
As $\lmsdc \pi_1, \ldots, \pi_k \rmsdc$ is a structural
dependency chain, we have $\pi_1 > \dots > \pi_k$, which implies
$\tau . \pi_1 > \dots > \tau . \pi_k$. Thus, the term $t_{i+1}$
has at most one innermost redex $\tau \pi_i$.
This concludes part (b).

\medskip
Part (c) follows directly from part (b) and the definitions of
$\pirc{\RR}(n)$ and $\irc{\RR}(n)$.
\end{proof}

Thus, for TRSs $\RR$ where
\thmref{thm:nopar} applies,
no rewrite rule
can introduce
parallel redexes, and
specific analysis techniques for
$\pircR$
are not needed.

\section{From parallel DTs to innermost rewriting}
\label{sec:dt_to_irc}
As we have seen in the previous section, we can transform a TRS $\RR$ with
parallel-innermost rewrite relation to a DT problem whose
complexity provides an upper bound of $\pircR$ (or,
for
confluent $\pito$, corresponds exactly to $\pircR$).
However,
DTs are only one of many
available techniques
to find bounds for $\ircR$. Other techniques
include, e.g.,
Weak Dependency Pairs \cite{Hirokawa08IJCAR},
usable replacement maps \cite{HirokawaMoser14},
the Combination Framework \cite{ava:mos:16},
a transformation to complexity
problems for integer transition systems \cite{naa:fro:bro:fuh:gie:17},
amortised complexity analysis \cite{MoserS20},
or techniques for finding \emph{lower} bounds \cite{LowerBounds}.
Thus, can we benefit also from other
techniques for (sequential) innermost complexity to analyse parallel
complexity?

\medskip
In this section, we answer the question in the affirmative, via a
generic transformation from Dependency Tuple problems back to rewrite
systems whose innermost complexity can then be analysed using
arbitrary existing techniques.

\medskip
We use
\emph{relative rewriting},
which allows for labelling some of the rewrite rules such that
their use does not contribute to the derivation height of a term.
In other words, rewrite steps with these rewrite rules are
``for free'' from the perspective of complexity.
Existing state-of-the-art tools like \aprove~\cite{aprove-tool}
and \tct~\cite{tct} are
able to find bounds on (innermost) runtime complexity of such
rewrite systems.

\begin{definition}[Relative Rewriting]
\label{def:rel}
For two TRSs $\RRA$ and $\RRB$,  $\RRA/\RRB$ is a
 \emph{relative TRS}.
Its \emph{rewrite relation} $\to_{\RRA/\RRB}$ is
  $\to^*_{\RRB} \circ \to_{\RRA} \circ \to^*_{\RRB}$, i.e.,  rewriting with
$\RRB$ is allowed before and after each $\RRA$-step.
We define the \emph{innermost} rewrite relation by
$s \someito{\RRA/\RRB} t$ iff
  $s \to^*_{\RRB} s' \to_{\RRA} s'' \to^*_{\RRB} t$  for some terms $s',
  s''$ such that the
proper subterms of the
redexes of each step with $\to_\RRB$ or $\to_\RRA$ are in normal form
w.r.t.\ $\RRA \cup \RRB$.

The set $\BasicTerms^{\RRA/\RRB}$ of basic terms for a
relative TRS $\RRA/\RRB$ is
$\BasicTerms^{\RRA/\RRB} = \BasicTerms^{\RRA\cup\RRB}$.
The notion of innermost runtime complexity extends to
relative TRSs in the natural way:
$\irc{\RRA/\RRB}(n) =
\sup \{ \Dh(t, {\someito{\RRA/\RRB}}) \mid t \in \BasicTerms^{\RRA/\RRB}, \tsize{t}
\leq n \}$
\end{definition}

The rewrite relation $\someito{\RRA/\RRB}$ is essentially
the same as $\someito{\RRA \cup \RRB}$, but only steps using
rules from $\RRA$ count towards the complexity;
steps
using rules from $\RRB$ have no cost.
This can be useful, e.g., for representing that built-in functions
from programming languages modelled as recursive functions
have constant cost.

\begin{example}
\label{ex:size_relative}
Consider a variant of \exref{ex:size1} where
$\Fplus(\FS(x), y) \to \FS(\Fplus(x, y))$ is moved
to $\RRB$,
but all other rules are elements of $\RRA$.
Then
$\RRA/\RRB$ would provide a modelling of the $\Fsize$ function that is
closer to the
OCaml
function from \autoref{sec:intro}.
Let $\FS^n(\FZero)$ denote the term obtained by $n$-fold application
of $\FS$ to $\FZero$ (e.g., $\FS^2(\FZero) = \FS(\FS(\FZero))$).
Although $\Dh(\Fplus(\FS^n(\FZero),\FS^m(\FZero)), \someito{\RRA\cup\RRB})\linebreak =
n+1$, we would then get $\Dh(\Fplus(\FS^n(\FZero),\FS^m(\FZero)),
\someito{\RRA/\RRB}) = 1$, corresponding to a machine model where the
time of evaluating addition for integers is constant.
\end{example}

Note the similarity of a relative TRS and a Dependency Tuple
problem: only certain rewrite steps count towards the
analysed complexity.
We make use of this observation for the following transformation.

\begin{definition}[Relative TRS for a Dependency Tuple Problem, $\detup$]
\label{def:to_relative}
Let $\langle \DD, \SSS, \RR \rangle$ be a Dependency Tuple problem.
We define the corresponding relative TRS:
\[\detup(\langle \DD, \SSS, \RR \rangle) =
\SSS/((\DD \setminus \SSS) \cup \RR).
\]
\end{definition}

In other words, we omit the information that steps with our
dependency tuples can happen only on top level (possibly below
constructors $\FCom_n$, but above $\someto{\RR}$ steps). (As we shall see in \thmref{thm:lower_rel}, this information can be recovered.)

\medskip
The following example is taken from the
\emph{Termination Problem Data Base (TPDB)} \cite{tpdb},
a collection of examples used at the
annual \emph{Termination and Complexity Competition (termCOMP)}
\cite{termcomp,termcompWiki} (see also \autoref{sec:expe}):

\begin{example}[TPDB, \texttt{HirokawaMiddeldorp\_04/t002}]
\label{ex:to_relative}
Consider the following TRS $\RR$ from category
\texttt{Innermost\_Runtime\_Complexity}
of the TPDB:
\[
\begin{array}{rcl@{\hspace*{5ex}}|@{\hspace*{5ex}}rcl}
\Fleq(\Fzero, y) & \to & \Ftrue
& \Fif(\Ftrue, x, y) & \to & x\\
\Fleq(\Fs(x), \Fzero) & \to & \Ffalse
& \Fif(\Ffalse, x, y) & \to & y\\
\Fleq(\Fs(x), \Fs(y)) & \to & \Fleq(x, y)
& \Fminus(x, \Fzero) & \to & x\\
\Fmod(\Fzero, y) & \to & \Fzero
& \Fminus(\Fs(x), \Fs(y)) & \to & \Fminus(x, y)\\
\Fmod(\Fs(x), \Fzero) & \to & \Fzero\\
\Fmod(\Fs(x), \Fs(y)) & \to & \multicolumn{4}{l}{\Fif(\Fleq(y, x), \Fmod(\Fminus(\Fs(x), \Fs(y)), \Fs(y)), \Fs(x))
}
\end{array}
\]
This TRS has the following PDTs $\DTpar(\RR)$:
\[
\begin{array}{rcl@{\hspace*{5ex}}|@{\hspace*{5ex}}rcl}
\tup{\Fleq}(\Fzero, y) & \to & \FCom_0
& \tup{\Fif}(\Ftrue, x, y) & \to & \FCom_0\\
\tup{\Fleq}(\Fs(x), \Fzero) & \to & \FCom_0
& \tup{\Fif}(\Ffalse, x, y) & \to & \FCom_0\\
\tup{\Fleq}(\Fs(x), \Fs(y)) & \to & \FCom_1(\tup{\Fleq}(x, y))
& \tup{\Fminus}(x, \Fzero) & \to & \FCom_0\\
\tup{\Fmod}(\Fzero, y) & \to & \FCom_0
& \tup{\Fminus}(\Fs(x), \Fs(y)) & \to & \FCom_1(\tup{\Fminus}(x, y))\\
\tup{\Fmod}(\Fs(x), \Fzero) & \to & \FCom_0\\
\tup{\Fmod}(\Fs(x), \Fs(y)) & \to &
                                    \multicolumn{4}{l}{\FCom_2(\tup{\Fleq}(y,x), \tup{\Fif}(\Fleq(y, x), \Fmod(\Fminus(\Fs(x), \Fs(y)), \Fs(y)), \Fs(x)))}\\
\tup{\Fmod}(\Fs(x), \Fs(y)) & \to &
                                    \multicolumn{4}{l}{\FCom_3(\tup{\Fminus}(\Fs(x), \Fs(y)), \tup{\Fmod}(\Fminus(\Fs(x), \Fs(y)), \Fs(y)),}\\
&& \multicolumn{4}{r}{\tup{\Fif}(\Fleq(y, x), \Fmod(\Fminus(\Fs(x), \Fs(y)), \Fs(y)), \Fs(x)))}
\end{array}
\]

The canonical parallel DT problem \journal{for $\RR$} is
$\langle \DTpar(\RR), \DTpar(\RR), \RR \rangle$.
We get the relative TRS
$\detup(\langle \DTpar(\RR), \DTpar(\RR), \RR \rangle) = \DTpar(\RR)/\RR$.
\end{example}

\begin{remark}
One of the reviewers suggested that the use of
$\Fif$ in \exref{ex:to_relative} indicated that TRSs with
innermost rewriting were inherently unable to
provide a faithful representation of the
evaluation strategy for conditional statements used in
programming languages with call-by-value evaluation.
The reason was that in the recursive $\Fmod$ rule,
the function calls in the subterm
$\Fmod(\Fminus(\Fs(x), \Fs(y)), \Fs(y))$
are evaluated also if the call to
$\Fleq(y, x)$ rewrites to $\Ffalse$ so that the
``then''-branch of the conditional evaluation would never
be needed. In contrast, a language like
OCaml, C++, Rust, \ldots would
evaluate the ``then''-branch only after the $\Fleq(y, x)$ had
evaluated to $\Ftrue$.

There are several ways of dealing with this modelling issue
in term rewriting.
One way is to impose a context-sensitive rewrite strategy
\cite{context-sensitive}
that ``freezes'' the second and third argument of $\Fif$.
Like this, the evaluation of these arguments is delayed until
after the result of evaluating the first argument is known,
and only the appropriate branch is evaluated.

\medskip
Another way that does not necessitate a different rewrite
strategy is to reorganise the rewrite rules with a dedicated
symbol $\Fcond$ for the \emph{specific} conditional expression
rather than using generic $\Fif$ rules.
In \exref{ex:to_relative},
we could replace the recursive $\Fmod$ rule
with the following rules (and potentially remove the
$\Fif$ rules, which would no longer be needed):

\[
\begin{array}{rcl}
\Fmod(\Fs(x), \Fs(y)) & \to & \Fcond(\Fleq(y, x), x, y)\\
\Fcond(\Ftrue, x, y)  & \to & \Fmod(\Fminus(\Fs(x), \Fs(y)), \Fs(y))\\
\Fcond(\Ffalse, x, y) & \to & \Fs(x)
\end{array}
\]

In this way, an innermost rewrite sequence would require
first evaluating the instance of $\Fleq(y, x)$ to normal form,
and depending on the result $\Ftrue$ or $\Ffalse$, the
corresponding $\Fcond$ rule will evaluate (only) the
corresponding branch of the earlier $\Fif$ statement.
We refrained from doing so in \exref{ex:to_relative} to
keep the link with an existing TRS from the literature
that reflects different modelling choices and
that our method should be able to analyse as well.
\end{remark}

From the definition of relative TRS, we are now able to prove an upper bound (Theorem~\ref{thm:upper_rel}) as well as a lower bound (Theorem~\ref{thm:lower_rel}) for parallel complexities.

\begin{theorem}[Upper Complexity Bounds for $\detup(\langle \DD, \SSS,
\RR \rangle)$ from
\label{thm:upper_rel}
  $\langle \DD, \SSS, \RR \rangle$]
Let $\langle \DD, \SSS, \RR \rangle$ be a
DT problem.
Then
\begin{enumerate}
\item[(a)] for all
$\tup{t} \in \SharpTerms$ with $t \in \BasicTerms^\RR$,
we have $\Cplx_{\hspace*{-1pt}\langle \DD, \SSS, \RR \rangle}(\tup{t})
\,{\leq} \Dh(\tup{t}, \itodetup)$,
and
\item[(b)] $\irc{\langle \DD, \SSS, \RR \rangle}(n) \,{\leq}
\,\irc{\detupTRS}(n)$.
\end{enumerate}
\end{theorem}

\begin{proof}
  We first show part (a) of the statement.
  For a DT Problem
  $\langle \DD, \SSS, \RR \rangle$ and a term
  $\tup{t} \in \SharpTerms$, consider an arbitrary
  chain tree $T$. We will show that if $|T|_\SSS = n$, then also
  $\Dh(\tup{t},\itodetup) \geq n$.
  We consider two cases.
  \begin{itemize}
    \item First, $n = \omega$.
    The set of all dependency tuples is finite, thus a finite term
    can only have a finite number of immediate successors.
    Therefore, $T$ is finitely branching, so there must be an infinite path
    with infinitely many nodes of the form
    $(\tup{u}_1 \to \FCom_{n_1}(\ldots,\tup{v_1},\ldots) \mid
    \sigma_1),
    (\tup{u}_2 \to \FCom_{n_2}(\ldots,\tup{v_2},\ldots) \mid
    \sigma_2), \dots$ such that
    $\tup{u}_1 \to \FCom_{n_1}(\ldots,\tup{v_1},\ldots),
    \tup{u}_2 \to \FCom_{n_2}(\ldots,\tup{v_2},\ldots),\linebreak
    \ldots \in \DD$.
    For infinitely many $i_1 < i_2 < i_3 < \dots$, we also have
    $\tup{u}_i \to \FCom_{n_i}(\ldots,\tup{v_i},\ldots) \in \SSS$,
    and for all $i$, we have
    $\tup{v}_i\sigma_i \itos \tup{u}_{i+1}\sigma_{i+1}$.
    Then we also have a corresponding infinite rewrite sequence
    \begin{align*}
    \tup{t} = \tup{u}_1\sigma_1 &\itosdetup
    C_1[\tup{u}_{i_1}\sigma_{i_1}]
    \someito{\SSS}
    C_1[\FCom_{n_{i_1}}(\ldots,\tup{v}_{i_1},\ldots)\sigma_{i_1}]\\
    &\someitos{(\DD\setminus\SSS)\cup\RR}\quad\;\;
    C_2[\tup{u}_{i_2}\sigma_{i_2}]
    \someito{\SSS}
    C_2[\FCom_{n_{i_2}}(\ldots,\tup{v}_{i_2},\ldots)\sigma_{i_2}]\\
    &\someitos{(\DD\setminus\SSS)\cup\RR}\quad\;\;
    \dots
    \end{align*}

    \noindent
    for some contexts $C_1,C_2,\ldots$
    (which result from rewrite steps with rules from $\DD$).

    \item Now consider the case $n \in \nat$. We use induction.
    For $n = 0$, the statement trivially holds.
    For the induction step, let $n > 0$.

\smallskip
    The (potentially infinite)
    chain tree $T$ has $m$ subtrees $T'_i$ with roots
    $(\tup{u}_i \to \FCom_{q_i}(\tup{v}_{i,1}, \ldots, \tup{v}_{i,q_i})
    \mid \sigma_i)$ such that
    $\tup{u}_i \to \FCom_{q_i}(\tup{v}_{i,1}, \ldots, \tup{v}_{i,q_i}) \in
    \SSS$ and the path in the chain tree from the root to
    $(\tup{u}_i \to \FCom_{q_i}(\tup{v}_{i,1}, \ldots, \tup{v}_{i,q_i})
    \mid \sigma_i)$
    has no outgoing edges from a node with a DT in $\SSS$.

\smallskip
    We show two separate statements in the induction step,
    which together
    let us conclude that
    $\Dh(\tup{t}, \itodetup) \geq |T'_1|_\SSS + \dots + |T'_m|_\SSS  = n$:
    \begin{gather}
    \text{For each $T'_i$, the term $\tup{u}_i\sigma_i$ has
    $\Dh(\tup{u}_i\sigma_i, \itodetup) \geq
    |T'_i|_\SSS$.} \label{belowS}\\[1.5ex]
    \begin{split}
    \text{There are contexts $C_1,\ldots,C_m$ such that\qquad}\\[-3pt]
    \tup{t} \itosdetup
    \FCom_m(C_1[\tup{u_1}\sigma_1],\ldots,C_m[\tup{u_m}\sigma_m]).
    \end{split}
    \label{aboveS}
    \end{gather}

    On \eqref{belowS}:
    Let $i \in \{1,\ldots,m\}$ be arbitrary and fixed, let
    $u = u_i$, let $\sigma = \sigma_i$, let $T'=T'_i$ (to ease
    notation).
    $T'$ is a chain tree for $u\sigma$ and its root is
    $(\tup{u} \to \FCom_{q}(\tup{v}_{1}, \ldots, \tup{v}_{q})
    \mid \sigma)$.
    Let this node have children
    $N_1 = (\tup{w}_1 \to \FCom_{r_1}(\ldots) \mid \mu_1)$,
    $\ldots$,
    $N_q = (\tup{w}_q \to \FCom_{r_q}(\ldots) \mid \mu_q)$.
    For the corresponding trees $T''_j$ with $N_j$ at the root,
    we have $|T''_j|_\SSS < |T'|_\SSS \leq n$ by construction,
    so the induction hypothesis is applicable to the terms
    $\tup{w}_j\mu_j$, and we get
    $\Dh(\tup{w}_j\mu_j, \itodetup)\linebreak \geq |T''_j|_\SSS$ for all
    $1 \leq j \leq q$. We construct a rewrite sequence with
    $\itodetup$ using at least $1 + |T''_1|_\SSS + \dots +
    |T''_q|_\SSS = |T'|_\SSS$ steps with a rule from $\SSS$ as follows:
    \begin{align*}
      \tup{u}\sigma &\someito{\SSS}
      \FCom_q(\tup{v}_{1}\sigma, \ldots, \tup{v}_{q}\sigma)\\
      & \someitos{\RR}
      \FCom_q(\tup{w}_{1}\mu_1, \ldots, \tup{v}_{q}\sigma)\\
      & \someitos{\RR}
      \dots\\
      & \someitos{\RR}
      \FCom_q(\tup{w}_{1}\mu_1, \ldots, \tup{w}_{q}\mu_q)
    \end{align*}

    \noindent
    With this rewrite sequence, we obtain \eqref{belowS} using the induction hypothesis:
    \begin{align*}
      &\quad\;\Dh(\tup{u}\sigma, \itodetup)\\
      &\geq
      1 + \Dh(\tup{w}_{1}\mu_1, \itodetup) + \dots +
      \Dh(\tup{w}_{q}\mu_q, \itodetup)\\
      &\geq 1 + |T''_1|_\SSS + \dots + |T''_q|_\SSS\\
      &= |T'|_\SSS
    \end{align*}

    On \eqref{aboveS}:
    Let the root of $T$ be
    $(\tup{\ell} \to \FCom_p(\tup{r}_1,\ldots,\tup{r}_p) \mid \nu)$.
    With a construction similar to the one used in the case
    $n = \omega$, we get:
    \begin{align*}
      \tup{t} = \tup{\ell}\nu &\someito{\DD}\qquad\;\:\:
      \FCom_p(\tup{r}_1\nu,\ldots,\tup{r}_p\nu)\\
      &\someitos{(\DD\setminus\SSS)\cup\RR}
      \FCom_p(C_1[\tup{u_1}\sigma_1],\ldots,\tup{r}_p\nu)\\
      &\someitos{(\DD\setminus\SSS)\cup\RR}
      \dots\\
      &\someitos{(\DD\setminus\SSS)\cup\RR}
      \FCom_p(C_1[\tup{u_1}\sigma_1],\ldots,C_m[\tup{u_m}\sigma_m])
    \end{align*}

    \noindent
    for some contexts $C_1,\ldots,C_m$ (which result from rewrite
    steps with rules from $\DD$). Note that here it suffices
    to reduce only in those
    subterms with a symbol $\tup{f}$ at their root
    that are on a path to one of the
    $C_i[\tup{u_i}\sigma_i]$, and depending on the tree structure,
    each $\tup{r}_j\nu$ may yield 0 or more of these $m$ terms
    (note that $p$ and $m$ are not necessarily equal).

\smallskip
    This concludes the induction step and hence the overall proof
    of part (a).
  \end{itemize}

  Part (b) follows from part (a), as shown in the following:
  \begin{align*}
    \irc{\langle \DD, \SSS, \RR \rangle}(n)
    &= \sup \{ \Cplx_{\langle \DD, \SSS, \RR \rangle}(\tup{t}) \mid
       t \in \BasicTerms^\RR, \tsize{t} \leq n \} & \text{by \defref{def:dt_problem}}\\
    &\leq \sup \{ \Dh(\tup{t}, {\itodetup}) \mid
       t \in \BasicTerms^\RR, \tsize{t} \leq n \} & \text{by part (a)}\\
    &\leq \sup \{ \Dh(s, {\itodetup}) \mid
       s \in \BasicTerms^{\RR\cup\DD}, \tsize{s} \leq n \}\\
    &= \irc{\detupTRS}(n)
  \end{align*}

  \vspace*{-6mm}
\end{proof}

\begin{example}[\exref{ex:to_relative} continued]
\label{ex:to_relative_upper}
For the relative TRS $\DTpar(\RR)/\RR$
from
\exref{ex:to_relative},
the
tool \aprove\ uses a transformation to
integer transition systems~\cite{naa:fro:bro:fuh:gie:17} followed by an application of the
complexity analysis tool \cofloco~\cite{CoFloCo,CoFloCoFM16}
to find
a bound $\irc{\DTpar(\RR)/\RR}(n) \in \OO(n)$ and
to deduce the bound $\pirc{\RR}(n) \in \OO(n)$ for the original TRS $\RR$
from the TPDB. In contrast, using the techniques of~\autoref{sec:para_complex}
without the transformation to a relative
TRS from \defref{def:to_relative}, \aprove\ finds only a bound $\pirc{\RR}(n) \in \OO(n^2)$.
\end{example}

Intriguingly, we can use our transformation
from \defref{def:to_relative} not only for
finding upper bounds, but also for \emph{lower} bounds on $\pirc{\RR}$.

\begin{theorem}[Lower Complexity Bounds for $\detup(\langle \DD, \SSS, \RR \rangle)$ from
  $\langle \DD, \SSS, \RR \rangle$]
\label{thm:lower_rel}
Let $\langle \DD, \SSS, \RR \rangle$ be a
DT
problem.
Then
\begin{enumerate}
\item[(a)] there is a type assignment
s.t.\ for all
$\ell \to r \in \DD \cup \RR$, $\ell$ and $r$ get the same type, and
for all well-typed $t \in \BasicTerms^{\DD \cup \RR}$,
$\Cplx_{\langle \DD, \SSS, \RR \rangle}(\tup{t}) \geq
\Dh(t, \itodetup)$, and
\item[(b)] $\irc{\langle \DD, \SSS, \RR \rangle}(n) \geq
\irc{\detupTRS}(n)$.
\end{enumerate}
\end{theorem}

\begin{proof}
We first consider the proof for part (a).

\smallskip
We use the following (many-sorted first-order
monomorphic) type assignment $\Theta$ with two sorts $\sbot$
and $\sDT$, where the arities of the symbols are respected (note that here all
arguments of a given symbol have the same type):
\begin{align*}
\Theta(f) & = \sbot \times \dots \times \sbot \to \sbot
              \text{ for $f$ in $\symsof{\RR}$}\\
\Theta(\tup{f}) & = \sbot \times \dots \times \sbot \to \sDT
                    \text{ for $\tup{f}$ a sharp symbol}\\
\Theta(\FCom_k) & = \sDT \times \dots \times \sDT \to \sDT
\end{align*}
With this type assignment, for all rules $\ell \to r \in \DD \cup
\RR$, $\ell$ and $r$ are well typed and have the same type: if $\ell
\to r \in \RR$, then all occurring symbols have the same result type
$\sbot$, which carries over to $\ell$ and $r$. And if $\ell \to r \in \DD$, then $\ell$ and $r$ have type
$\sDT$. To see that $\ell$ and $r$ are well typed, consider that every
term $\ell$ has the shape $\tup{f}(s_1,\ldots,s_n)$, where $\tup{f}$
has result type $\sDT$ and expects all arguments to have type $\sbot$,
while all $s_i$ contain only subterms of type $\sbot$. Similarly,
$r$ has the shape $\FCom_k(\tup{f}_1(t_{1,1},\ldots,t_{1,n_1}),
\ldots, \tup{f}_k(t_{k,1},\ldots,t_{k,n_k}))$.
$\FCom_k$ has result type $\sDT$ and expects all arguments to have
type $\sDT$. This is the case since all $\tup{f_i}$ have result type
$\sDT$. And all $\tup{f_i}$, which are right below the root, expect
their arguments $t_{i,j}$ to have result type $\sbot$. This is the
case by construction.

In the following, we consider basic terms that
are well typed according to $\Theta$ as start terms.
For our relative TRS $\SSS/((\DD\setminus\SSS)\cup\RR)$,
we have
the following
two kinds of well-typed basic terms that we need to consider:

\begin{description}
\item[Case 1:]
$t = f(t_1,\ldots,t_n)$ with $f \in \DefSyms$ and $t_1,\ldots,t_n
\in \TT(\ConSyms,\VV)$. This term and all its subterms are of type
$\sbot$. Thus, this term can be rewritten by rules from $\RR$, but not
by rules from $\DD$ (and $\SSS$), which all have type $\sDT$. As
rewriting preserves the type of terms, $t$ is a normal form w.r.t.\ the relations
$\someito{\Theta(\SSS/((\DD\setminus\SSS)\cup\RR))}$ and
$\someito{\SSS/((\DD\setminus\SSS)\cup\RR)}$, and
$\Dh(t,\someito{\SSS/((\DD\setminus\SSS)\cup\RR)}) = 0$.
Since $\Cplx_{\langle \DD, \SSS, \RR \rangle}(s) \geq 0$
regardless of $s$, the claim follows for this case.

\medskip

\item[Case 2:]
$t = \tup{f}(t_1,\ldots,t_n)$ with $f \in \DefSyms$ and
$t_1,\ldots,t_n \in \TT(\ConSyms,\VV)$.
If $t$ is a normal form, there is no tree, and
$\Dh(t, \someito{\SSS/((\DD\setminus\SSS)\cup\RR)}) = 0 = \Cplx_{\langle \DD, \SSS, \RR \rangle}(t)$.

Otherwise, we can convert any ${\someito{\SSS \cup ((\DD\setminus\SSS)\cup\RR)}} =
{\someito{\DD \cup \RR}}$ rewrite
sequence to a $(\DD,\RR)$-chain tree $T$ for $t$ such that
$\Dh(t, \someito{\SSS/((\DD\setminus\SSS)\cup\RR)}) = |T|_\SSS$,
including any rewrite sequence that witnesses
$\Dh(t, \someito{\SSS/((\DD\setminus\SSS)\cup\RR)})$
in the following way:

\smallskip
As $t$ is a basic term, the first step in the rewrite sequence
rewrites at the root of the term. Since only rules from
$\DD$ are applicable to terms with $\tup{f}$ at the root,
this step uses a DT $\tup{s} \to \FCom_k(\ldots)$ from $\DD$.
With $\sigma$ as the used matcher for the rewrite step,
we obtain the root node $(\tup{s} \to \FCom_k(\ldots) \mid \sigma)$.

\smallskip
Now assume that we have a partially constructed chain tree $T'$ for the
rewrite sequence so far, which we have represented up until
the term $s$ that resulted from a $\someito{\DD}$ step.

If there are no further $\someito{\DD}$ steps in the rewrite sequence,
we have completed our chain tree $T = T'$
as the remaining $\someito{\RR}$
suffix of the rewrite sequence does not contribute
to $\Dh(t,\someito{\SSS/((\DD\setminus\SSS)\cup\RR)})$
(only steps using rules from $\SSS \subseteq \DD$ are counted).

Otherwise, our remaining rewrite sequence has the shape
$s \someitos{\RR} u \someito{\DD} v \someitom{\DD \cup \RR} \dots$
for some $m \in \nat \cup \{ \omega \}$.
The step $u \someito{\DD} v$ takes place at position $\pi$,
using the DT $\tup{p} \to \FCom_l(\tup{q}_1,
\ldots,\tup{q}_l) \in \DD$ and the matcher $\mu$.

We can reorder the rewrite steps
$s \someitos{\RR} u$ by advancing all $\someito{\RR}$ steps
at positions $\tau > \pi$, yielding
$s \someitos{\RR,>\pi} s' \someitos{\RR,\not>\pi} u$.
(This reordering is possible in our innermost setting.)
Here $\someito{\RR,>\pi}$ denotes an innermost rewrite step
using rules from $\RR$ at a position $\tau > \pi$, and
$\someito{\RR,\not>\pi}$ denotes an innermost rewrite step
using rules from $\RR$ at a position $\tau' \not > \pi$.
Now we change our remaining rewrite sequence to
$s \someitos{\RR,>\pi} s' \someito{\DD} u' \someitos{\RR,\not>\pi} v \someitom{\DD \cup \RR} \dots$.
Let $\FCom_k(\tup{q'}_1, \ldots,\tup{q'}_k) \delta = s|_\pi$.
Since the $s' \someito{\DD} u'$ rewrite step has not been encoded yet,
there is a $j$ such that $\tup{q'}_j \delta \someito{\RR} \tup{p}\mu$
has not yet been used in the construction.
Therefore, there exists a node $N = (\tup{p'} \to \FCom_k(\tup{q'}_1,
\ldots,\tup{q'}_k) \mid \delta)$
where for some $j$, the subterm $\tup{q'}_j\delta$ in the DT of $N$
has not yet been used for this purpose in the
construction before.
We encode the subsequence $s \someitos{\RR,>\pi} s' \someito{\DD} u'$
by adding the node $(\tup{p} \to \FCom_l(\tup{q}_1,
\ldots,\tup{q}_l) \mid \mu)$ to $T'$ as a child to $N$.

We obtain the chain tree $T''$, which we extend further
by encoding the rewrite sequence
$u' \someitos{\RR,\not>\pi} v \someitom{\DD \cup \RR} \dots$
following the same procedure.

Since our construction adds a node with a DT from $\SSS$ in
the first component of the label
whenever the rewrite sequence uses a rule from $\SSS$, we
have
$\Dh(t, \someito{\SSS/((\DD\setminus\SSS)\cup\RR)}) = |T|_\SSS$
as desired.
This concludes the proof for part (a).

We now prove part (b).

Innermost runtime complexity is known to be a persistent
property w.r.t.\ type introduction \cite{irc_persistent}.
For our relative TRS $\detupTRS$, this means that we may introduce
an arbitrary (many-sorted first-order monomorphic) type assignment
$\Theta$ for all symbols in the considered signature
such that the rules in $\RR\cup\DD$ are well typed.
We obtain a typed relative TRS $\Theta(\detupTRS)$, and
$\irc{\Theta(\detupTRS)}(n) = \irc{\detupTRS}(n)$
holds. Thus, only basic terms that are well typed according to
$\Theta$
need to be considered as start terms for $\irc{\detupTRS}$.
We write $\Theta(\BasicTerms^{\DD\cup\RR})$ for the set of well-typed
basic terms for $\detupTRS$.
\smallskip

We use the type assignment $\Theta$ from part (a)
to restrict the set of basic terms as start terms.
With this type assignment, we obtain:
\begin{align*}
&\phantom{{}={}}\irc{\detupTRS}(n)\\
&= \irc{\Theta(\detupTRS)}(n)
   & \text{by \cite{irc_persistent}}\\
&= \sup \{ \Dh(t, {\itodetup}) \mid
t \in \Theta(\BasicTerms^{\RR\cup\DD}), \tsize{t} \leq n \}
   & \text{by \defref{def:rel}}\\
&\leq \sup \{ \Cplx_{\langle \DD, \SSS, \RR \rangle}(\tup{t}) \mid
t \in \Theta(\BasicTerms^{\RR\cup\DD}), \tsize{t} \leq n \}
   & \text{by part (a)}\\
&\leq \sup \{ \Cplx_{\langle \DD, \SSS, \RR \rangle}(\tup{t}) \mid
t \in \BasicTerms^{\RR\cup\DD}, \tsize{t} \leq n \}
   & \text{drop types $\Rightarrow$ more start terms}\\
&\leq \irc{\langle \DD, \SSS, \RR \rangle}(n)
\end{align*}

\end{description}
\end{proof}

\thmref{thm:upper_rel}
and \thmref{thm:lower_rel}
hold
regardless of
whether the original
DT problem was obtained from a TRS with sequential or with parallel
evaluation.
So while this
kind of connection between DT (or DP) problems and
relative rewriting may be
folklore in the community, its application
to convert a TRS whose \emph{parallel} complexity is sought
to a TRS with the same \emph{sequential} complexity is new.
\smallskip

\begin{example}[\exref{ex:to_relative_upper} and \exref{ex:non-overlapping} continued]
\label{ex:lower}
We continue \exref{ex:to_relative_upper}.
\thmref{thm:lower_rel}
implies that a lower bound for
$\irc{\DTpar(\RR)/\RR}(n)$ of the relative TRS $\DTpar(\RR)/\RR$
from \exref{ex:to_relative} carries over to
$\langle \DTpar(\RR), \DTpar(\RR), \RR \rangle$ and,
presuming that $\pito$ is confluent, also to
$\pirc{\RR}(n)$ of the original TRS $\RR$ from the
TPDB.
\aprove uses rewrite lemmas \cite{LowerBounds} to find
the lower bound $\irc{\DTpar(\RR)/\RR}(n) \in \Omega(n)$.
Together with \exref{ex:to_relative_upper}, we have
automatically inferred
that this complexity bound is \emph{tight} if we can also prove
confluence of $\pito$:
$\pirc{\RR}(n) \in \Theta(n)$.
We shall see the missing confluence proof for $\pito$ in
\autoref{sec:confluence}.
\end{example}

Note that \thmref{thm:canonical_pdt_problem} requires confluence of
$\pito$ to derive lower bounds for $\pirc{\RR}$
from lower complexity bounds of the canonical parallel DT problem.
So
to use \thmref{thm:lower_rel} to search for \emph{lower}
complexity bounds with existing techniques~\cite{LowerBounds},
we need a criterion for confluence of parallel-innermost rewriting.
\autoref{sec:confluence} shall be dedicated to proposing
  two sufficient syntactic criteria for confluence that can
  be checked automatically.

\section{Confluence of parallel-innermost rewriting}
\label{sec:confluence}

Recall that confluence of a relation $\to$ means that
whenever we have $t_1 \to^* t_2$ and $t_1 \to^* t_3$, there is also some
$t_4$ with $t_2 \to^* t_4$ and $t_3 \to^* t_4$. In other words,
if $\to$ is confluent, any non-determinism
between steps with $\to$ that \emph{temporarily}
leads to different
outcomes can always be undone to reach a common object.
Methods for analysis of confluence have been
a topic of interest for many years (see, e.g.,
\cite{KB70,Rosen73,Huet80} for early work), motivated
both by applications in theorem proving and as a
topic of study in its own right. In recent years,
the development of automated tools for confluence analysis
of (sequential) term rewriting has flourished.
This is witnessed by the Confluence Competition
\cite{coco}, which has been running annually since 2012
to compare state-of-the-art tools for automated confluence analysis.
However, we are not aware of any tool at the
Confluence Competition that currently supports the analysis
of parallel-innermost rewriting.

As an alternative, it would be tempting to use an
existing tool for confluence analysis of standard
rewriting, possibly restricted to innermost
rewriting, as a decidable sufficient criterion for confluence
of parallel-innermost rewriting.
However, the following example shows that this approach
is in general not sound.

\begin{example}[Confluence of $\ito$ \journal{or $\someto{\RR}$} does not Imply Confluence of
$\pito$]
\label{ex:no_confluence}
To see that we cannot prove confluence of $\pito$ just by
using a standard off-the-shelf tool for
confluence analysis of innermost or full rewriting \cite{coco}, consider the TRS
$\RR = \{\Fa \to \Ff(\Fb,\Fb), \Fa \to \Ff(\Fb,\Fc), \Fb \to \Fc,
\Fc \to \Fb\}$. For this TRS, both $\ito$ and $\someto{\RR}$ are confluent.
However, $\pito$ is not confluent: we can rewrite both
$\Fa \pito \Ff(\Fb,\Fb)$ and $\Fa \pito \Ff(\Fb,\Fc)$,
yet there
is no term $v$ such that $\Ff(\Fb,\Fb) \pitos v$ and
$\Ff(\Fb,\Fc) \pitos v$. The reason is that the only possible rewrite
sequences with $\pito$ from these terms are
$\Ff(\Fb,\Fb) \pito \Ff(\Fc,\Fc) \pito \Ff(\Fb,\Fb) \pito \dots$
and
$\Ff(\Fb,\Fc) \pito \Ff(\Fc,\Fb) \pito \Ff(\Fb,\Fc) \pito \dots$,
with no terms in common.
\end{example}

Thus, in general a confluence proof for $\someto{\RR}$ or $\ito$
does not imply confluence for $\pito$. Yet it seems that other criteria on
$\ito$ or $\someto{\RR}$ may be sufficient: intuitively, the reason
for non-confluence for $\pito$ in \exref{ex:no_confluence}
is the non-termination of $\ito$.
\begin{proposition}
\label{conj:terminating_pito_ito}
Let $\RR$ be a TRS whose innermost rewrite relation $\ito$ is
terminating. Then $\ito$ is confluent iff $\pito$ is
confluent.
\end{proposition}

The next proposition is motivated by applying techniques
  for proving confluence of $\pito$, as developed in this paper,
  to proving confluence of $\ito$.

\begin{proposition}
\label{conj:pito_confluent_implies_ito_confluent}
Let $\RR$ be a (not necessarily terminating) TRS.
If $\pito$ is confluent, then $\ito$ is confluent.
\end{proposition}

We
conjectured \propref{conj:terminating_pito_ito} first in
our informal extended abstract \cite[Conjecture 1]{iwc2022}
and \propref{conj:pito_confluent_implies_ito_confluent}
in the preliminary conference version of this paper
\cite[Conjecture 1]{lopstr2022}.
We are grateful to van Oostrom for providing proofs for
both statements \cite{vvO} and closing our conjectures.

\subsection{Confluence of $\pito$ for non-overlapping rules}

We recall the standard notions of \emph{uniformly confluent} and
\emph{deterministic} relations, which are special cases of
confluent relations. In the following, we will use these
notions to identify sufficient
criteria for confluence of parallel-innermost term rewriting.

\begin{definition}
A relation $\to$ is
\emph{uniformly confluent} iff $s \to t$ and $s \to u$ imply that
 $t = u$ or that there exists an object $v$ with $t \to v$ and $u \to v$, and
$\to$ is \emph{deterministic} iff for every $s$ there is at most one
$t$ with $s \to t$.
\end{definition}

Let us work towards a first sufficient criterion for confluence of parallel-innermost rewriting.
Confluence means:
if a term $s$ can be rewritten to two
different terms $t_1$ and $t_2$ in 0 or more steps,
it is
always possible to rewrite $t_1$ and $t_2$ in 0 or more steps
to
the same
term $u$.
For
$\pito$, the redexes that get rewritten are
fixed: all
innermost redexes simultaneously.
Thus,
$s$ can
rewrite to two \emph{different} terms $t_1$ and $t_2$
only if at least one of these redexes can be rewritten in two
different ways using
$\ito$.

Towards a sufficient criterion for confluence of parallel-innermost
rewriting, we introduce the following standard definitions used in
confluence analysis:

\begin{definition}[Unifier, Most General Unifier, see also \cite{BaaderNipkow}]
Two terms $s$ and $t$ \emph{unify} iff there exists a substitution
$\sigma$ (called a \emph{unifier} of $s$ and $t$)
such that $s\sigma = t\sigma$. A unifier $\sigma$ of $s$ and $t$
is called a \emph{most general unifier} of $s$ and $t$ iff
for all unifiers $\delta$ of $s$ and $t$ there exists
some substitution $\delta'$ such that
$(s\sigma)\delta' = s \delta = t \delta = (t\sigma)\delta'$.
\end{definition}

Critical pairs capture the local non-determinism that arises if
a given redex may be rewritten by different rewrite rules
or at different positions in the same redex. They are defined
with the help of most general unifiers to determine which
instances of left-hand sides may lead to local non-determinism.

\begin{definition}[Critical Pair, Critical Peak, Critical Overlay \cite{KB70}, see also \cite{BaaderNipkow}]
For a given TRS $\RR$, let $\ell \to r, u \to v \in \RR$
be rules whose variables have been renamed apart and let
$\pi$ be a position in $u$ such that $\ell|_\pi \notin \VV$.
If $\pi = \varepsilon$, we require that $\ell \to r$ and $u \to v$
are not variants of the same rule, i.e., that we cannot obtain $\ell \to r$
by renaming variables in $u \to v$.
If $\ell$ and $u|_\pi$ unify with most general unifier $\sigma$,
then we call $\cp{u\sigma[r\sigma]_\pi}{v\sigma}$ a
\emph{critical pair},
resulting from the \emph{critical peak} (i.e., local non-determinism)
between the steps $u\sigma \someto{\RR} v\sigma$ and
$u\sigma = u\sigma[\ell\sigma]_\pi \someto{\RR} u\sigma[r\sigma]_\pi$.
If $\pi = \varepsilon$, we call the critical pair
$\cp{u\sigma[r\sigma]_\pi}{v\sigma} = \cp{r\sigma}{v\sigma}$
a \emph{critical overlay}, and we may write
$\ocp{r\sigma}{v\sigma}$.
\end{definition}

A critical peak is the concrete local non-determinism for rewriting
a term in two different ways (with overlapping redexes, or using
different rewrite rules). It results in a critical pair that describes
this non-determinism in an abstract way. Finite TRSs have only
finitely many critical pairs, which is very useful for analysis of
confluence.

\begin{example}
\label{ex:critical_pairs}
Consider the (highly artificial, but illustrative) TRS
$\RR = \{ \Ff(\Fa) \to \Fb, \Ff(x) \to \Fc, \Fa \to \Fd  \}$.
The rewrite relation $\someto{\RR}$ of this TRS is not
confluent: for example, we can rewrite $\Ff(\Fa) \someto{\RR} \Fb$
using the first rule and $\Ff(\Fa) \someto{\RR} \Fc$ using the
second rule, and neither $\Fb$ nor $\Fc$ can be rewritten any further.

\medskip
We have the following critical pairs/overlays for $\RR$:
\[
\begin{array}{cr}
\ocp{\Fb}{\Fc} & \text{from the first and second rule}\\
\ocp{\Fc}{\Fb} & \text{from the first and second rule}\\
\cp{\Ff(\Fd)}{\Fb} & \text{from the first and third rule}
\end{array}
\]
The left-hand sides of the first and the second rule,
$\Ff(\Fa)$ and $\Ff(x)$, unify at the root position $\varepsilon$
with the most general unifier $\sigma = \{ x \mapsto \Fa \}$.
Thus, these two rules have a critical pair, and since the
unification was at root position, this critical pair is also
a critical overlay. If we instantiate both rules using $\sigma$, we get
$\Ff(\Fa) \to \Fb$ and $\Ff(\Fa) \to \Fc$. The right-hand sides
of these instantiated rules are $\Fb$ and $\Fc$, and they are
the components of the critical overlays $\ocp{\Fb}{\Fc}$
and $\ocp{\Fc}{\Fb}$.

Note that every critical \emph{overlay} comes together with
its mirrored version: if two different rules
$\ell_1 \to r_1$ and $\ell_2 \to r_2$
unify at the root position (and can thus both be used for rewriting a
redex $\ell_1 \sigma = \ell_2 \sigma$ at the root), there is a
symmetry and thus a choice which of the terms $r_1 \sigma$
and $r_2 \sigma$ to write on the left and which one on the right
of the critical overlay. This is why the first and second rule
together produce \emph{two} critical overlays.

Now let us consider our first and our third rule.
The left-hand side $\Ff(\Fa)$ of the first rule has at its position $1$
the non-variable subterm $\Fa$ that unifies with the left-hand side
of the third rule, $\Fa$, using the identity substitution as the
most general unifier. Thus, we consider the instantiated
left-hand side of the first rule, $\Ff(\Fa)$, as the redex that
may be rewritten either
at position $1$ with the third rule, to $\Ff(\Fd)$, or
at the root with the first rule, to $\Fb$.
This leads to the critical pair $\cp{\Ff(\Fd)}{\Fb}$.

Note that critical pairs that are not critical overlays do not have
the symmetry mentioned earlier: the ``inside'' rewrite step at position
$\pi > \varepsilon$ is always written on the left.
\end{example}

\begin{definition}[Non-Overlapping]
A TRS $\RR$ is \emph{non-overlapping} iff $\RR$ has no critical pairs.
\end{definition}

A sufficient criterion that a given redex has a unique result from a
rewrite step is given in the following.

\begin{lemma}[\cite{BaaderNipkow}, Lemma 6.3.9]
\label{lem:unique_reduct}
If a TRS $\RR$ is non-overlapping,
$s \someto{\RR} t_1$ and
$s \someto{\RR} t_2$ with the redex of both rewrite steps at the same
position, then $t_1 = t_2$.
\end{lemma}

In other words, for non-overlapping TRSs,
rewriting a specific redex has a deterministic result.
In a parallel-innermost rewrite step $s \pito t$,
we rewrite all innermost redexes in $s$ at the same time, so the choice
of redexes to use is also deterministic.
Together, this means that in a rewrite step $s \pito t$, the term $t$
is uniquely determined for $s$, so the relation $\pito$ is deterministic
as well.

With the above reasoning, this lemma directly gives us a sufficient criterion for confluence of
\emph{parallel-innermost} rewriting by determinism.

\begin{corollary}[Confluence of Parallel-Innermost Rewriting]
\label{cor:confluence}
If a TRS $\RR$ is non-overlapping, then $\pito$ \journal{is deterministic and hence}
confluent.
\end{corollary}

\begin{remark}
The reasoning behind \corref{cor:confluence} can be generalised
to \emph{arbitrary} parallel rewrite strategies
where the redexes that are rewritten are fixed,
such as (max-)parallel-\emph{outermost} rewriting
\cite{LoopsUnderStrategies}.
\end{remark}

\begin{remark}
Note that in contrast to similar confluence criteria for
full rewriting $\someto{\RR}$ \cite{Rosen73,Huet80},
here $\RR$ is \emph{not} required to be left-linear (i.e., $\RR$
may have rewrite rules where the left-hand side has more than one
occurrence of the same variable).

\corref{cor:confluence} is similar to a result by Gramlich for (uniform) confluence
of (sequential) innermost rewriting $\ito$ for non-overlapping TRSs
that are not necessarily left-linear \cite[Lemma 3.2.1, Corollary 3.2.2]{GramlichPhD}.
For parallel-innermost rewriting, we have the stronger property that $\pito$
is even deterministic: for innermost rewriting with $\ito$, there is still the
non-deterministic choice between different innermost redexes, whereas
the used redexes for a step with $\pito$ are uniquely determined.
\end{remark}

\begin{example}
\label{ex:non-overlapping}
The TRSs $\RR$ from \exref{ex:size1}, \exref{ex:doubles}, and
\exref{ex:to_relative} are all non-overlapping, and by
\corref{cor:confluence} their parallel-innermost rewrite
relations $\pito$ are confluent.
Thus, also the tight complexity bound
$\pirc{\RR}(n) \in \Theta(n)$ in \exref{ex:lower} is confirmed.
\end{example}

So, in those cases we can actually use this sequence of
transformations from a parallel-innermost TRS via a DT problem to an
innermost (relative) TRS to analyse both upper and lower bounds for
the original. Conveniently, these cases correspond to
programs
with deterministic small-step semantics, our motivation for this work!

\begin{example}
\label{ex:max}
\corref{cor:confluence}
already fails for such natural examples as a TRS with the
following rules to compute the maximum function on
natural numbers:
\[
\begin{array}{rcl}
\Fmax(\FZero,x) & \to & x \\
\Fmax(x,\FZero) & \to & x \\
\Fmax(\FS(x),\FS(y)) & \to & \FS(\Fmax(x,y)) \\
\end{array}
\]
Here we can arguably see immediately that the overlap
between the first and the second rule, at root position,
is harmless: if both rules are applicable to the same redex, the
result of a rewrite step with either rule will be the same
($\Fmax(\FZero,\FZero) \pito \FZero$).
Indeed, the resulting critical pair has the form
$\cp{\Fmax(\FZero,\FZero)}{\Fmax(\FZero,\FZero)}$,
with both components of the critical pair the same.
\end{example}

\subsection{Confluence of $\pito$ with trivial innermost critical overlays}

Critical pairs like $\cp{\Fmax(\FZero,\FZero)}{\Fmax(\FZero,\FZero)}$
where both components are identical are also called \emph{trivial}.
If a critical pair is trivial, it means that the non-determinism in the
choice of rules or positions in the redex does not lead to a non-determinism in the result
of the two rewrite steps described abstractly by the critical pair.
For the purposes of confluence, such critical pairs are always harmless.
For example, for analysing confluence of \exref{ex:max}, the critical pair
$\cp{\Fmax(\FZero,\FZero)}{\Fmax(\FZero,\FZero)}$ can be ignored.

For innermost rewriting, critical pairs resulting from an overlap between
a redex with another redex as a subterm can be ignored as well: the outer
redex is not enabled for an innermost rewrite step. For example, in
\exref{ex:critical_pairs}, the critical pair $\cp{\Ff(\Fd)}{\Ff(\Fb)}$
can be ignored. The reason is that for \emph{innermost} rewriting, here
everything is deterministic: only the rewrite step
$\Ff(\Fa) \ito \Ff(\Fd)$ is innermost, but the rewrite step
$\Ff(\Fa) \someto{\RR} \Fb$ is not.
Thus, for confluence of innermost rewriting it suffices to consider
critical \emph{overlays}, which describe rewriting a redex at the
\emph{same} position of a term using two different rewrite rules,
and we can ignore all other critical pairs.

Moreover, for innermost rewriting, we can even ignore those
critical overlays that can result only from non-innermost rewrite steps,
such as the two critical overlays $\ocp{\Fb}{\Fc}$ and $\ocp{\Fc}{\Fb}$
in \exref{ex:critical_pairs}: the term $\Ff(\Fa)$ that causes these
critical overlays in a critical peak cannot be rewritten
\emph{innermost} at the root because the subterm $\Fa$ must be rewritten
first.

This considerations give rise to the following definitions
(see, e.g., \cite{GramlichPhD}):

\begin{definition}[Trivial critical pair]
Critical pairs of the form $\cp{t}{t}$ are called \emph{trivial}.
\end{definition}

\begin{definition}[Innermost critical overlay]
A critical overlay $\ocp{s}{t}$
resulting from a critical peak $\ell \sigma = u \sigma \someto{\RR} s$
and $\ell \sigma = u \sigma \someto{\RR} t$
is called \emph{innermost} for $\RR$ iff we also have
$\ell \sigma = u \sigma \ito s$ and $\ell \sigma = u \sigma \ito t$.
\end{definition}

Gramlich combines the above observations on critical pairs
in his PhD thesis \cite{GramlichPhD} to the following
sufficient criterion for (uniform) confluence of innermost rewriting:

\begin{theorem}[\cite{GramlichPhD}, Theorem 3.5.6]
\label{thm:ito_confluent_with_cps}
Let $\RR$ be a TRS such that all innermost critical overlays of $\RR$
are trivial. Then $\ito$ is uniformly confluent and hence confluent.
\end{theorem}

We can apply a similar reasoning to \thmref{thm:ito_confluent_with_cps}
to get a stronger criterion for $\pito$ being deterministic.

\begin{theorem}[Parallel-innermost confluence from
    only trivial innermost critical overlays]
\label{thm:pito_confluent_with_cps}
Let $\RR$ be a TRS such that all innermost critical overlays of $\RR$
are trivial. Then $\pito$ is deterministic and hence confluent.
\end{theorem}

\thmref{thm:pito_confluent_with_cps} subsumes \corref{cor:confluence}:
if there are no critical pairs at all, then there are also no non-trivial
innermost critical overlays.

\begin{proof}
We prove that the relation $\pito$ is deterministic if
all innermost critical overlays of $\RR$ are trivial.
To this end, we will show that if $t_0 \pito t_1$ and
$t_0 \pito t_2$, we have $t_1 = t_2$.

Assume $t_0 \pito t_1$ and $t_0 \pito t_2$ for some terms $t_0, t_1, t_2$.
Since the rewrite step is parallel-innermost,
all used redexes are fixed: all the innermost redexes.
Let $s_0$ be an arbitrary innermost redex in $t_0$.
If $s_0 \ito s_1$ and $s_0 \ito s_2$ implies $s_1 = s_2$,
the statement $t_1 = t_2$ would directly follow.

Thus, assume $s_0 \ito s_1$ and $s_0 \ito s_2$.
As these rewrite steps are innermost, rewriting must take place
at the root of $s_0$.
Assume that two different rules $\ell \to r, u \to v \in \RR$
with variables renamed apart are used for these rewrite steps
(otherwise the claim follows directly). Let
$\delta, \theta$ be substitutions such that
$s_0 = u \delta \ito v \delta = s_1$ and
$s_0 = \ell \theta \ito r \theta = s_2$.

There are corresponding critical peaks
$u\sigma \ito v\sigma$ and $\ell\sigma \ito r\sigma$
with $\sigma$ a most general unifier of $u$ and $\ell$.
Since the rewrite steps are innermost, the critical peak
gives rise to innermost critical overlays
$\ocp{v\sigma}{r\sigma}$ and $\ocp{r\sigma}{v\sigma}$.

As $\sigma$ is a most general unifier of $u$ and $\ell$,
we
have $s_0 = u \delta = u \sigma \delta'$ and
$s_0 = \ell \theta = \ell\sigma \theta'$ with
$\delta'(x) = \theta'(x)$ for all variables
$x \in \VV(u\delta) \cup \VV(\ell\theta)$.
Thus, $s_1 = v \delta = v \sigma\delta'$ and
$s_2 = r \theta = r \sigma\theta' = r\sigma\delta'$.
As the critical overlays are trivial by precondition of
our theorem, we have $v \sigma = r \sigma$ and hence also
$s_1 = v \sigma\delta' = r\sigma\delta' = s_2$.
This concludes our proof.
\end{proof}

\begin{example}
Since the only innermost critical overlay
$\ocp{\Fmax(\FZero,\FZero)}{\Fmax(\FZero,\FZero)}$
for the TRS $\RR$ from
\exref{ex:max} is trivial, $\pito$ is confluent.
\end{example}

\begin{example}
Since none of the critical overlays in the TRS $\RR$ from
\exref{ex:critical_pairs} is innermost, $\pito$ is confluent.
\end{example}

\begin{remark}
From \thmref{thm:pito_confluent_with_cps} one may be tempted
to claim that $\pito$ is deterministic iff $\ito$ is uniformly
confluent. The ``$\Rightarrow$'' direction clearly holds.
But for the ``$\Leftarrow$'' direction consider the following
counterexample:
$\RR = \{ \Fa \to \Fb, \Fa \to \Fc, \Fb \to \Fd, \Fc \to \Fd \}$.
The relation $\ito$ is uniformly confluent, but $\pito$ with
$\Fa \pito \Fb$ and $\Fa \pito \Fc$ is not deterministic.
\end{remark}

With \corref{cor:confluence} and
\thmref{thm:pito_confluent_with_cps} we have proposed two
sufficient criteria for proving
confluence of parallel-innermost rewriting for given TRSs
that can be automated using syntactic checks only, without
any search problems.
These criteria specifically capture TRSs corresponding to
deterministic programs.

\section{Implementation and experiments}
\label{sec:expe}
\emph{Implementation.}
We have implemented the contributions of this paper in the
automated termination and complexity analysis tool
\aprove~\cite{aprove-tool}. We added or modified
\journal{about 730} lines of Java code,
including
\begin{itemize}
\itemsep=3pt
\item the framework of parallel-innermost rewriting;
\item the generation of parallel DTs (\thmref{thm:canonical_pdt_problem});
\item a processor to convert them to TRSs with the same complexity
(\thmref{thm:upper_rel}, \thmref{thm:lower_rel});
\item the confluence tests of~\corref{cor:confluence}
  and of \thmref{thm:pito_confluent_with_cps}.
\end{itemize}
As far as we are aware, this is the first implementation of a fully
automated inference of complexity bounds for parallel-innermost
rewriting.
A preliminary implementation of our techniques in
  \aprove\ participated successfully in the new demonstration
  category ``Runtime Complexity: TRS Parallel Innermost''
  at termCOMP 2022 and 2023.

\medskip
This implementation is now part of the \aprove\ release versions and
can be downloaded or used via a web interface~\cite{aprove-tool}. The
input format is an extension of the human-readable text format that was
used to represent TRSs in early versions of the TPDB.
For example, a file \texttt{size.trs} for \exref{ex:size1} would
have the content shown in \autoref{fig:size.trs}.

\begin{figure}[h]
\vspace*{2mm}
\begin{verbatim}
(GOAL COMPLEXITY)
(STARTTERM CONSTRUCTORBASED)
(STRATEGY PARALLELINNERMOST)
(VAR v l r x y)
(RULES
  size(Tree(v, l, r)) -> S(plus(size(l), size(r)))
  size(Nil) -> Zero
  plus(Zero, y) -> y
  plus(S(x), y) -> S(plus(x, y))
)
\end{verbatim} \vspace*{-4mm}
\caption{Input file for \exref{ex:size1}.}
\label{fig:size.trs}
\end{figure}

In this format, we can designate a rewrite rule as a
\emph{relative} rule, as in \exref{ex:size_relative},
by writing ``\verb!->=!'' instead of ``\verb!->!''.

\emph{Experiments.}
To demonstrate the effectiveness of our implementation, we have
considered the 663 TRSs from category
\texttt{Runtime\_Complexity\_Innermost\_Rewriting} of the
TPDB, version 11.2~\cite{tpdb}.\footnote{Version 11.3 of the
  TPDB was released in July 2022, but does not contain changes
  over version 11.2 for the category
  \texttt{Runtime\_Complexity\_Innermost\_Rewriting}.}
This category
of the TPDB
is the benchmark collection used at
termCOMP to compare tools that
infer complexity bounds for runtime complexity of innermost
rewriting, $\ircR$.
To get meaningful results, we first applied \thmref{thm:nopar} to
exclude TRSs $\RR$ where $\pircR(n) = \ircR(n)$ trivially holds.
We obtained
294 TRSs with potential for parallelism as our benchmark set.
We conducted our experiments
on the \starexec\ compute cluster \cite{starexec} in
the \texttt{all.q} queue. The timeout per example and tool
configuration was set to 300 seconds.
Our experimental data with analysis times
and all examples are available online~\cite{evalPageJournal}.

As remarked earlier, we always have $\pircR(n) \leq \ircR(n)$,
so an upper bound for $\ircR(n)$ is always a legitimate
upper bound for $\pircR(n)$.
Thus,
we include
upper bounds for $\ircR$
found by the state-of-the-art tools \aprove\ and \tct~\cite{versionPage,tct} from
termCOMP 2021\footnote{For analysis of $\ircR$, both tools
  participated in termCOMP 2022 with their versions from
  termCOMP 2021.}
as a ``baseline'' in our evaluation.
We compare with several configurations of \aprove\ and \tct\ that use
the
techniques of this paper for $\pircR$:
``\aprove\ $\pircR$ Section~3'' also uses
\thmref{thm:canonical_pdt_problem} to produce canonical
parallel DT problems as input for the DT framework.
``\aprove\ $\pircR$ Sections~3~\&~4'' additionally uses
the transformation from \defref{def:to_relative}
to convert a TRS $\RR$ to a relative TRS
$\DTpar(\RR)/\RR$
and then to analyse $\irc{\DTpar(\RR)/\RR}(n)$ (for lower bounds
only together with a confluence proof
either via \corref{cor:confluence} \journal{or via
\thmref{thm:pito_confluent_with_cps}, as indicated where relevant}).
We also extracted each of the TRSs $\DTpar(\RR)/\RR$
and used the files as inputs
for the analysis of $\irc{\DTpar(\RR)/\RR}$ by \aprove\ and
\tct\ from termCOMP 2021.
``\aprove\ $\pircR$ Section~4''
and ``\tct\ $\pircR$ Section~4'' provide the results
for $\pircR$ obtained by analysing $\irc{\DTpar(\RR)/\RR}$
(for lower bounds, only where $\pito$ had been proved confluent).

\begin{table}[!h]
\vspace*{-2mm}
\begin{center}
\tabcolsep=5pt
\caption{Upper bounds for runtime complexity of (parallel-)innermost rewriting}
\label{table:ui}
{\scalebox{0.95}{
\begin{tabular}{|l||c|c|c|c|c||c|}
\hline
Tool & $\Oh(1)$ & $\leq\Oh(n)$ & $\leq\Oh(n^2)$ & $\leq\Oh(n^3)$
     & $\leq\Oh(n^{\geq 4})$ & avg.\ time (s) \\\hline\hline
\tct\ $\ircR$
& 4 & 32 & 51 & 62 & 67 & 202.9 \\ 
  \aprove\ $\ircR$
     & \textbf{5} & 50 & 111 & 123 & 127 & 193.8 \\ 
    \hline
\aprove\ $\pircR$ Section \ref{sec:para_complex}
& \textbf{5} & \textbf{70} & \textbf{125} & 139 & 141 & 222.0 \\ 
\aprove\ $\pircR$ Sections \ref{sec:para_complex} \&
  \ref{sec:dt_to_irc}
& \textbf{5} & \textbf{70} & \textbf{125} & \textbf{140} & \textbf{142} & 211.5 \\ 
\hline
\tct\ $\pircR$ Section \ref{sec:dt_to_irc}
& 4 & 46 & 66 & 79 & 80 & \textbf{189.7} \\ 
\aprove\ $\pircR$ Section \ref{sec:dt_to_irc}
& \textbf{5} & 64 & 99 & 108 & 108 & 219.8 \\ 
\hline
\end{tabular} } }
\end{center}
\end{table}

\autoref{table:ui}
gives an overview over our experimental results for upper bounds.
For each configuration, we state the number of examples for
which the corresponding asymptotic complexity bound
was
inferred.
A
column ``$\leq \Oh(n^k)$'' means that the corresponding tools proved a bound
$\leq \Oh(n^k)$ (e.g., the configuration
``\aprove\ $\ircR$'' proved constant or linear upper bounds in
50
cases). Maximum values in a column are highlighted in \textbf{bold}.
We observe that upper complexity bounds improve in a
noticeable number of cases, e.g., linear bounds on $\pircR$ can now be inferred
for
70 TRSs rather than for
50 TRSs (using upper bounds on $\ircR$
as an over-approximation), an improvement by 40\%.
Note that this does \emph{not} indicate deficiencies in the existing tools for $\ircR$,
which had not been designed with analysis of $\pircR$ in mind -- rather,
it shows that specialised techniques for analysing $\pircR$ are a worthwhile
subject of investigation.
Note also that \exref{ex:size:irc} and
\exref{ex:size_tight_bound} show that even for TRSs
with potential for parallelism, the actual parallel and sequential
complexity may still be asymptotically identical,
which further highlights the need for dedicated analysis techniques for $\pircR$.

Improvements from $\ircR$ to $\pircR$ can be drastic:
for example, for the TRS \texttt{TCT\_12/recursion\_10}, the
bounds found by \aprove\ change from an upper bound of sequential
complexity of $\OO(n^{10})$ to a (tight) upper bound for parallel
complexity of $\OO(n)$. This TRS models a specific recursion
structure, with rules
$\{ \Ff_0(x) \to \Fa \} \cup
 \{ \Ff_i(x) \to \Fg_i(x,x), \; \Fg_i(\Fs(x), y) \to \Fb(\Ff_{i-1}(y),
 \Fg_i(x,y)) \mid 1 \leq i \leq 10 \}$, and is highly amenable to parallelisation.
This TRS resembles a classical ``syntactically vectorisable loop''. In such cases, vectorisation indeed accelerates the program by this order of magnitude.
Our analysis captures this
kind of acceleration, however we should keep in mind that apart
from vectorisation, even perfect ``thread-based'' parallelisation
on $N$ cores does not achieve $N$ times acceleration, due to the
cost of context switching. Our complexity result should be taken
as an estimation of ``parallelism potential''.

We observe that adding the
techniques from \autoref{sec:dt_to_irc} to the techniques from
\autoref{sec:para_complex} leads to only few examples for which better
upper bounds can be found (one of them is \exref{ex:to_relative_upper}).

\begin{table}[!h]
\tabcolsep=5pt
\begin{center}
\caption{Lower bounds for runtime complexity of parallel-innermost
rewriting}
\label{table:li}
\begin{tabular}{|l||c||c|c|c|c|}
\hline
Tool
& confluent
& $\geq\Omega(n)$ & $\geq\Omega(n^2)$ & $\geq\Omega(n^3)$
  & $\geq\Omega(n^{\geq 4})$ \\\hline\hline
\begin{tabular}{@{}l@{}}
  \aprove\ $\pircR$ Sections \ref{sec:para_complex} \&
    \ref{sec:dt_to_irc},\\
  \quad confluence by \corref{cor:confluence}
\end{tabular}
  & 165
& 116 & \textbf{22} & \textbf{5} & \textbf{1} \\ 
\begin{tabular}{@{}l@{}}
  \aprove\ $\pircR$ Sections \ref{sec:para_complex} \&
    \ref{sec:dt_to_irc},\\
  \quad confluence by \thmref{thm:pito_confluent_with_cps}
\end{tabular}
& \textbf{190}
& 133 & \textbf{22} & \textbf{5} & \textbf{1} \\ 
\hline
\begin{tabular}{@{}l@{}}
\tct\ $\pircR$ Section \ref{sec:dt_to_irc},\\
  \quad confluence by \corref{cor:confluence}
\end{tabular}
& 165
& 112 & 0 & 0 & 0 \\ 
\begin{tabular}{@{}l@{}}
\tct\ $\pircR$ Section \ref{sec:dt_to_irc},\\
  \quad confluence by \thmref{thm:pito_confluent_with_cps}
\end{tabular}
& \textbf{190}
& 131 & 0 & 0 & 0 \\ 
\hline
\begin{tabular}{@{}l@{}}
\aprove\ $\pircR$ Section \ref{sec:dt_to_irc},\\
  \quad confluence by \corref{cor:confluence}
\end{tabular}
& 165
& 140 & 21 & \textbf{5} & \textbf{1} \\ 
\begin{tabular}{@{}l@{}}
\aprove\ $\pircR$ Section \ref{sec:dt_to_irc},\\
  \quad confluence by \thmref{thm:pito_confluent_with_cps}
\end{tabular}
& \textbf{190}
& \textbf{159} & 21 & \textbf{5} & \textbf{1} \\ 
\hline
\end{tabular}
\end{center}
\end{table}
\begin{table}[h]
\tabcolsep=5pt
\begin{center}
\caption{Tight bounds for runtime complexity of parallel-innermost
rewriting}
\label{table:tight}
\begin{tabular}{|l||c|c|c|c||c|}
\hline
Tool & $\Theta(1)$ & $\Theta(n)$ & $\Theta(n^2)$ & $\Theta(n^3)$ & Total \\\hline\hline
\begin{tabular}{@{}l@{}}
  \aprove\ $\pircR$ Sections \ref{sec:para_complex} \&
    \ref{sec:dt_to_irc},\\
  \quad confluence by \corref{cor:confluence}
\end{tabular}
& \textbf{5} & 27 & 0 & \textbf{3} & 35 \\ 
\begin{tabular}{@{}l@{}}
  \aprove\ $\pircR$ Sections \ref{sec:para_complex} \&
    \ref{sec:dt_to_irc},\\
  \quad confluence by \thmref{thm:pito_confluent_with_cps}
\end{tabular}
& \textbf{5} & 33 & 0 & \textbf{3} & 41 \\ 
\hline
\begin{tabular}{@{}l@{}}
\tct\ $\pircR$ Section \ref{sec:dt_to_irc},\\
  \quad confluence by \corref{cor:confluence}
\end{tabular}
& 4 & 19 & 0 & 0 & 23 \\
\begin{tabular}{@{}l@{}}
\tct\ $\pircR$ Section \ref{sec:dt_to_irc},\\
  \quad confluence by \thmref{thm:pito_confluent_with_cps}
\end{tabular}
& 4 & 22 & 0 & 0 & 26 \\
\hline
\begin{tabular}{@{}l@{}}
\aprove\ $\pircR$ Section \ref{sec:dt_to_irc},\\
  \quad confluence by \corref{cor:confluence}
\end{tabular}
& \textbf{5} & 33 & 0 & \textbf{3} & 41 \\
\begin{tabular}{@{}l@{}}
\aprove\ $\pircR$ Section \ref{sec:dt_to_irc},\\
  \quad confluence by \thmref{thm:pito_confluent_with_cps}
\end{tabular}
& \textbf{5} & \textbf{41} & 0 & \textbf{3} & \textbf{49} \\
\hline
\end{tabular}
\end{center}\vspace*{-4mm}
\end{table}

\autoref{table:li} shows our results for lower bounds on $\pircR$.
Here we evaluated only
configurations including
\defref{def:to_relative} to make inference techniques
for lower bounds of $\ircR$ applicable to $\pircR$.
The reason is that a lower
bound on $\ircR$ is not necessarily also a lower bound for $\pircR$
(the whole \emph{point} of performing innermost rewriting in parallel is to
reduce the asymptotic complexity!), so
using results by tools that compute lower bounds on $\ircR$ for
comparison would not make sense.
\journal{As a precondition for applying our approach to lower
bounds inference for $\pircR$ for a TRS $\RR$, we also need to
find a proof for confluence of $\pito$.
The confluence criterion of \corref{cor:confluence}
is applicable to 165 TRSs in our benchmark set, about 56.1\%
of the benchmark set. Our new contribution
\thmref{thm:pito_confluent_with_cps} proves confluence
of $\pito$ for a superset of 190 TRSs, about 64.6\% of our
benchmark set.
This indicates that the search for more powerful criteria
for proving confluence of parallel-innermost term rewriting
is worthwhile.}

Regarding lower bounds, we
observe that non-trivial lower
bounds can be inferred for
140
out of the
165 examples proved confluent via
\corref{cor:confluence}, and for 159 out of
the 190 examples proved confluent via
\thmref{thm:pito_confluent_with_cps}.
This shows that
our transformation from \autoref{sec:dt_to_irc}
has practical value since it
produces relative TRSs
that are generally amenable to analysis by existing program analysis
tools. It also shows that the more powerful confluence
  analysis by \thmref{thm:pito_confluent_with_cps} improves
  also the inference of lower complexity bounds.
Finally, \autoref{table:tight} shows that for overall
49 TRSs, the bounds that were found are asymptotically
\emph{precise}.\footnote{Unfortunately, the implementation of our
  confluence check used in the experiments for the conference
  version \cite{lopstr2022} had a bug that caused it to consider
  some TRSs confluent where \corref{cor:confluence} was not
  applicable (specifically, overlaps due to certain
  non-trivial critical overlays were not detected as such).
  Overall, 21 TRSs were considered confluent even though
  \corref{cor:confluence} could not be applied, and
  18 of the lower bounds claimed in the experimental
  evaluation of \cite{lopstr2022} were affected.
  This bug has now been fixed.}

However, the nature of the benchmark set also plays
a significant role for assessing the applicability of criteria
for proving confluence.
Therefore, we also considered COPS \cite{cops},
the benchmark collection used in the
\emph{Confluence Competition}~\cite{coco}.
As a benchmark collection for our second experiment to assess our
confluence analysis,
we downloaded the \journal{577} unsorted unconditional TRSs
of COPS.\footnote{The download took place on \journal{29 May 2023}.}
While the TRSs in this subset of COPS are usually analysed for confluence of
full rewriting, we analysed whether the TRSs are confluent for
parallel-innermost rewriting (which currently does not have a
dedicated category in COPS).
Our implementation determined that 60 of the 577
TRSs (about 10.4\%) are non-overlapping, which implies parallel-innermost
confluence by \corref{cor:confluence}.
In contrast, \thmref{thm:pito_confluent_with_cps}
finds 274 confluence proofs (about 47.5\%), a significantly
better result.

Still, the success rate for this benchmark set is
significantly lower than for the examples from the TPDB.
This
is not surprising: COPS collects TRSs that provide a challenge
to confluence analysis tools, whereas the analysed subset of the TPDB
contains TRSs which are interesting specifically for runtime complexity
analysis and often correspond to programs with deterministic results.

\emph{Runtime of the analysis.}
\autoref{table:ui} shows the (mean) average time used by the respective
tool configurations to analyse their inputs for (parallel-)innermost
runtime complexity (the search for upper and lower bounds was run
concurrently).
With the used timeout of 300 seconds, all configurations needed
between 180 and 240 seconds on average per example. It may perhaps
be surprising that even the fastest configuration used over
180 seconds per example.
The likely reason is that determining the asymptotic runtime complexity
of a rewrite system is an optimisation problem: as long as the highest
lower bound and the lowest upper bound found so far do not coincide,
there is the possibility that applying further techniques in the
analysis may lead to tighter bounds. Thus, complexity analysis tools
like \aprove\ and \tct\ are usually configured to exhaust most of the
available time by a search for upper and lower bounds, finishing their
search only when
(a) the current upper and lower bound coincide and can be reported
    as the best possible result,
(b) all available techniques for the given configuration have been tried, or
(c) the timeout is almost reached and the current result can be reported
    as the best result obtained within the given timeout.
Indeed, in our experiments with a timeout of 300 seconds, both
\aprove\ and \tct\ have runtimes between 290 and 300 seconds on many
benchmarks.

Regarding the confluence-only analysis, on most TRSs in our
collection both used criteria usually return a result within a few
milliseconds. We believe that this very quick result is due to the
fact that our criteria are purely syntactic and do not involve any
search problems.
For \corref{cor:confluence}, the highest runtime on our benchmark
suite was observed for \texttt{Frederiksen\_Glenstrup/int}
(183 rewrite rules)
with 76 ms on a computer with an Intel Core i7-10750H CPU @ 2.60GHz.
For \thmref{thm:pito_confluent_with_cps}, the highest runtime
was observed for
\texttt{Transformed\_CSR\_04/LISTUTILITIES\_complete\_noand\_GM}
(407 rewrite rules)
with 135 ms on the same computer. Our implementation is not
particularly optimised, so we anticipate that the criteria can
be made to scale to larger examples as well.

\section{Related work, conclusion, and future work}
\label{sec:related}
\label{sec:conclusion}

\emph{Related work.}
We provide pointers to work on automated analysis of (sequential) innermost runtime complexity of
TRSs at the start of \autoref{sec:dt_to_irc}, and we
discuss the apparent absence of work on confluence
of parallel-innermost rewriting in \autoref{sec:confluence}.
We now focus on
automated techniques
for
complexity analysis of
parallel/concurrent computation.

Our
notion of
parallel complexity follows a large
tradition of static \emph{cost analysis}, notably
for
concurrent programming. The two notable
works~\cite{AlbertLCTES11,AlbertTOCL18} address
async/finish programs where tasks are explicitly launched.
The authors propose several metrics such as the total number of spawned
tasks (in any execution of the program) and
a notion of
parallel complexity that is roughly the same as ours. They provide
static analyses that build on
techniques for estimating costs
of imperative languages with functions calls~\cite{Albert2012}, and/or
recurrence equations.
Recent approaches for the Pi
Calculus~\cite{BaillotPiESOP21,BaillotPiCONCUR21} compute
the
\emph{span} (our parallel complexity)
through a new typing system. Another type-based calculus for the same
purpose has been proposed with session types~\cite{HoffmannICFP18}.

\medskip
For logic programs, which -- like TRSs -- express an implicit parallelism,
parallel complexity can be inferred
using recurrence solving~\cite{GallagherLOPSTR19}.

The tool \raml~\cite{raml}
derives bounds on the worst-case evaluation cost of
first-order functional programs with list and pair
constructors as well as pattern matching and both sequential and
parallel composition \cite{HoffmannESOP15}.
They use two typing derivations with specially annotated types, one for the \emph{work}
and one for the \emph{depth} (parallel complexity).
Our
setting is more flexible w.r.t.\ the shape of user-defined
data structures (we allow for tree constructors of arbitrary arity),
and our analysis
deals with both data structure and control in an integrated manner.

\medskip
\emph{Conclusion and future work.}
We have defined parallel-innermost
runtime complexity for TRSs and proposed an approach to its
automated analysis.
Our approach allows for finding both upper and lower bounds and builds
on existing techniques and tools. Our experiments on the TPDB indicate that our
approach is practically usable, and we are
confident that it captures the potential parallelism of programs with pattern matching.

Parallel rewriting is a topic of active research, e.g., for GPU-based
massively parallel rewrite engines \cite{gpu_trs,gpu_trs_journal}.
Here our work could
be useful to determine which functions to evaluate on the GPU.
More generally, parallelising
compilers which need to determine which function calls should be
compiled into parallel code may benefit from an analysis of
parallel-innermost runtime complexity such as ours.

DTs have been used \cite{lctrsComplexity} in runtime complexity analysis of
\emph{Logically Constrained TRSs (LCTRSs)}\ \cite{lctrs13,lctrsTocl}, an extension of
TRSs by built-in data types from SMT theories
(integers, arrays, \ldots).
This work could be extended to parallel rewriting.
Moreover, analysis of
\emph{derivational complexity} \cite{derivational89}
of parallel-innermost term rewriting can be a promising direction. Derivational complexity
considers the length of rewrite sequences from arbitrary start terms, e.g.,
$\Fd(\Fd(\dots(\Fd(\FS(\FZero)))\dots))$ in our motivating example (\ref{sec:example}, \exref{ex:doubles}),
which can have longer derivations than basic terms of the same size.
Finally, towards
automated parallelisation
we aim to
infer
complexity bounds w.r.t.\ term
\emph{height}
(terms = trees!), as suggested in~\cite{wst16trs}.

\medskip
For \emph{confluence analysis},
an obvious next step
towards
more powerful criteria
would be to adapt the classic
confluence criterion by Knuth and Bendix \cite{KB70}.
By this criterion, a TRS $\RR$ has a confluent rewrite
relation $\someto{\RR}$ if it is terminating and for each of
its critical pairs $\cp{t_1}{t_2}$, there exists some term
$s$ such that $t_1 \sometos{\RR} s$ and $t_2 \sometos{\RR} s$
(i.e., $t_1$ and $t_2$ are \emph{joinable}).
Termination can in many cases be proved automatically
by modern termination analysis tools \cite{termcomp},
and for terminating TRSs, it is decidable whether
two terms are joinable (by rewriting them to all possible
normal forms and checking whether a common normal form
has been reached).
This criterion has been adapted for confluence of
innermost rewriting \cite[Theorem 3.5.8]{GramlichPhD}
via critical overlays $\ocp{t_1}{t_2}$.
A promising next step would be to investigate if
further modifications are needed for proving confluence
of parallel-innermost rewriting.

\medskip
Towards handling larger and more difficult inputs,
it would be worth investigating
to what extent confluence criteria for (parallel-)innermost
rewriting can be made \emph{compositional}
\cite{CompositionalConfluence22,CompositionalConfluence24}
or integrated into the \emph{Confluence Framework}
\cite{confident}.
This would allow
combinations of different confluence criteria to work together
by focusing on different parts of a TRS for the overall
confluence proof.

In a different direction, formal \emph{certification}
of the proofs
found using the techniques in this paper would be highly desirable.
Unfortunately, automated tools for program verification such as
\aprove\ are not immune to logical errors in their programming.
For TRSs, many proof techniques for properties such as
termination, complexity bounds, and confluence have been
formalised in trusted proof assistants such as
\coq\ or \isabelleHOL\ \cite{cime,color,ceta,nijn}.
Based on these formalisations, proof certifiers have been
created to check proof traces (for a given property).
Such proof traces are usually
generated by automated tools specialised in \emph{finding}
proofs, such as \aprove, whose source code has not
been formally verified. The certifier then either verifies that
the proof trace indeed correctly instantiates the formalised
proof techniques for the given TRS, or it points out that this
was not the case (ideally with a pointer to the specific step
in the proof trace that could not be verified).

\subsection*{Acknowledgements}
We thank Vincent van Oostrom and the anonymous reviewers
  of earlier and the current versions of this paper for helpful discussions and comments.


\begin{thebibliography}{10}
\providecommand{\url}[1]{\texttt{#1}}
\providecommand{\urlprefix}{URL }
\expandafter\ifx\csname urlstyle\endcsname\relax
  \providecommand{\doi}[1]{doi:\discretionary{}{}{}#1}\else
  \providecommand{\doi}{doi:\discretionary{}{}{}\begingroup
  \urlstyle{rm}\Url}\fi
\providecommand{\eprint}[2][]{\url{#2}}

\bibitem{versionArxivJournal}
Baudon T, Fuhs C, Gonnord L.
\newblock On Complexity Bounds and Confluence of Parallel Term Rewriting, 2024.
\newblock \urlprefix\url{https://arxiv.org/abs/2305.18250}.

\bibitem{BaillotPiESOP21}
Baillot P, Ghyselen A.
\newblock Types for Complexity of Parallel Computation in Pi-Calculus.
\newblock In: Yoshida N (ed.), Programming Languages and Systems - 30th
  European Symposium on Programming, {ESOP} 2021, Held as Part of the European
  Joint Conferences on Theory and Practice of Software, {ETAPS} 2021,
  Luxembourg City, Luxembourg, March 27 - April 1, 2021, Proceedings, volume
  12648 of \emph{Lecture Notes in Computer Science}. Springer, 2021 pp. 59--86.
\newblock \urlprefix\url{https://doi.org/10.1007/978-3-030-72019-3\_3}.

\bibitem{BaillotPiCONCUR21}
Baillot P, Ghyselen A, Kobayashi N.
\newblock {Sized Types with Usages for Parallel Complexity of Pi-Calculus
  Processes}.
\newblock In: Haddad S, Varacca D (eds.), 32nd International Conference on
  Concurrency Theory, {CONCUR} 2021, August 24-27, 2021, Virtual Conference,
  volume 203 of \emph{LIPIcs}. Schloss Dagstuhl - Leibniz-Zentrum f{\"{u}}r
  Informatik, 2021 pp. 34:1--34:22.
\newblock \urlprefix\url{https://doi.org/10.4230/LIPIcs.CONCUR.2021.34}.

\bibitem{GallagherLOPSTR19}
Klemen M, L{\'{o}}pez{-}Garc{\'{\i}}a P, Gallagher JP, Morales JF, Hermenegildo
  MV.
\newblock A General Framework for Static Cost Analysis of Parallel Logic
  Programs.
\newblock In: Gabbrielli M (ed.), Logic-Based Program Synthesis and
  Transformation - 29th International Symposium, {LOPSTR} 2019, Porto,
  Portugal, October 8-10, 2019, Revised Selected Papers, volume 12042 of
  \emph{Lecture Notes in Computer Science}. Springer, 2019 pp. 19--35.
\newblock \urlprefix\url{https://doi.org/10.1007/978-3-030-45260-5\_2}.

\bibitem{AlbertTOCL18}
Albert E, Correas J, Johnsen EB, Pun KI, Rom\'{a}n-D\'{\i}ez G.
\newblock Parallel Cost Analysis.
\newblock \emph{ACM Trans. Comput. Logic}, 2018.
\newblock \textbf{19}(4):31:1--31:37.
\newblock \urlprefix\url{https://doi.org/10.1145/3274278}.

\bibitem{HoffmannESOP15}
Hoffmann J, Shao Z.
\newblock Automatic Static Cost Analysis for Parallel Programs.
\newblock In: Vitek J (ed.), Programming Languages and Systems - 24th European
  Symposium on Programming, {ESOP} 2015, Held as Part of the European Joint
  Conferences on Theory and Practice of Software, {ETAPS} 2015, London, UK,
  April 11-18, 2015. Proceedings, volume 9032 of \emph{Lecture Notes in
  Computer Science}. Springer, 2015 pp. 132--157.
\newblock \urlprefix\url{https://doi.org/10.1007/978-3-662-46669-8\_6}.

\bibitem{HoffmannICFP18}
Das A, Hoffmann J, Pfenning F.
\newblock Parallel complexity analysis with temporal session types.
\newblock \emph{Proc. {ACM} Program. Lang.}, 2018.
\newblock \textbf{2}({ICFP}):91:1--91:30.
\newblock \urlprefix\url{https://doi.org/10.1145/3236786}.

\bibitem{parallelRewriting}
Vuillemin J.
\newblock Correct and Optimal Implementations of Recursion in a Simple
  Programming Language.
\newblock \emph{J. Comput. Syst. Sci.}, 1974.
\newblock \textbf{9}(3):332--354.
\newblock \urlprefix\url{https://doi.org/10.1016/S0022-0000(74)80048-6}.

\bibitem{innermostOrdering}
Fern{\'{a}}ndez M, Godoy G, Rubio A.
\newblock Orderings for Innermost Termination.
\newblock In: Giesl J (ed.), Term Rewriting and Applications, 16th
  International Conference, {RTA} 2005, Nara, Japan, April 19-21, 2005,
  Proceedings, volume 3467 of \emph{Lecture Notes in Computer Science}.
  Springer, 2005 pp. 17--31.
\newblock \urlprefix\url{https://doi.org/10.1007/978-3-540-32033-3\_3}.

\bibitem{wst16trs}
Alias C, Fuhs C, Gonnord L.
\newblock {Estimation of Parallel Complexity with Rewriting Techniques}.
\newblock In: Proceedings of the 15th Workshop on Termination (WST 2016). 2016
  pp. 2:1--2:5.
\newblock \urlprefix\url{https://hal.archives-ouvertes.fr/hal-01345914}.

\bibitem{gpu_trs}
van Eerd J, Groote JF, Hijma P, Martens J, Wijs A.
\newblock Term Rewriting on {GPUs}.
\newblock In: Hojjat H, Massink M (eds.), Fundamentals of Software Engineering
  - 9th International Conference, {FSEN} 2021, Virtual Event, May 19-21, 2021,
  Revised Selected Papers, volume 12818 of \emph{Lecture Notes in Computer
  Science}. Springer, 2021 pp. 175--189.
\newblock \urlprefix\url{https://doi.org/10.1007/978-3-030-89247-0\_12}.

\bibitem{gpu_trs_journal}
van Eerd J, Groote JF, Hijma P, Martens J, Osama M, Wijs A.
\newblock Innermost many-sorted term rewriting on GPUs.
\newblock \emph{Sci. Comput. Program.}, 2023.
\newblock \textbf{225}:102910.
\newblock \urlprefix\url{https://doi.org/10.1016/j.scico.2022.102910}.

\bibitem{lopstr2022}
Baudon T, Fuhs C, Gonnord L.
\newblock Analysing Parallel Complexity of Term Rewriting.
\newblock In: Villanueva A (ed.), Logic-Based Program Synthesis and
  Transformation - 32nd International Symposium, {LOPSTR} 2022, Tbilisi,
  Georgia, September 21-23, 2022, Proceedings, volume 13474 of \emph{Lecture
  Notes in Computer Science}. Springer, 2022 pp. 3--23.
\newblock \urlprefix\url{https://doi.org/10.1007/978-3-031-16767-6\_1}.

\bibitem{DependencyTuple}
Noschinski L, Emmes F, Giesl J.
\newblock Analyzing Innermost Runtime Complexity of Term Rewriting by
  Dependency Pairs.
\newblock \emph{J. Autom. Reason.}, 2013.
\newblock \textbf{51}(1):27--56.
\newblock \urlprefix\url{https://doi.org/10.1007/s10817-013-9277-6}.

\bibitem{Lankford75}
Lankford DS.
\newblock Canonical algebraic simplification in computational logic.
\newblock Technical Report ATP-25, University of Texas, 1975.

\bibitem{satPolo}
Fuhs C, Giesl J, Middeldorp A, Schneider{-}Kamp P, Thiemann R, Zankl H.
\newblock {SAT} Solving for Termination Analysis with Polynomial
  Interpretations.
\newblock In: Marques{-}Silva J, Sakallah KA (eds.), Theory and Applications of
  Satisfiability Testing - {SAT} 2007, 10th International Conference, Lisbon,
  Portugal, May 28-31, 2007, Proceedings, volume 4501 of \emph{Lecture Notes in
  Computer Science}. Springer, 2007 pp. 340--354.
\newblock \urlprefix\url{https://doi.org/10.1007/978-3-540-72788-0\_33}.

\bibitem{smtPolo}
Borralleras C, Lucas S, Oliveras A, Rodr{\'{\i}}guez{-}Carbonell E, Rubio A.
\newblock {SAT} Modulo Linear Arithmetic for Solving Polynomial Constraints.
\newblock \emph{J. Autom. Reason.}, 2012.
\newblock \textbf{48}(1):107--131.
\newblock \urlprefix\url{https://doi.org/10.1007/s10817-010-9196-8}.

\bibitem{aprove-tool}
Giesl J, Aschermann C, Brockschmidt M, Emmes F, Frohn F, Fuhs C, Hensel J, Otto
  C, Pl{\"{u}}cker M, Schneider{-}Kamp P, Str{\"{o}}der T, Swiderski S,
  Thiemann R.
\newblock Analyzing Program Termination and Complexity Automatically with
  {AProVE}.
\newblock \emph{J. Autom. Reason.}, 2017.
\newblock \textbf{58}(1):3--31.
\newblock Web interface and download:
  \url{https://aprove.informatik.rwth-aachen.de/},
  \urlprefix\url{https://doi.org/10.1007/s10817-016-9388-y}.

\bibitem{tct}
Avanzini M, Moser G, Schaper M.
\newblock {TcT}: {T}yrolean {C}omplexity {Tool}.
\newblock In: Chechik M, Raskin J (eds.), Tools and Algorithms for the
  Construction and Analysis of Systems - 22nd International Conference, {TACAS}
  2016, Held as Part of the European Joint Conferences on Theory and Practice
  of Software, {ETAPS} 2016, Eindhoven, The Netherlands, April 2-8, 2016,
  Proceedings, volume 9636 of \emph{Lecture Notes in Computer Science}.
  Springer, 2016 pp. 407--423.
\newblock Web interface and download:
  \url{https://www.uibk.ac.at/en/theoretical-computer-science/research/software/tct/},
  \urlprefix\url{https://doi.org/10.1007/978-3-662-49674-9\_24}.

\bibitem{lctrs13}
Kop C, Nishida N.
\newblock Term Rewriting with Logical Constraints.
\newblock In: Fontaine P, Ringeissen C, Schmidt RA (eds.), Frontiers of
  Combining Systems - 9th International Symposium, FroCoS 2013, Nancy, France,
  September 18-20, 2013. Proceedings, volume 8152 of \emph{Lecture Notes in
  Computer Science}. Springer, 2013 pp. 343--358.
\newblock \urlprefix\url{https://doi.org/10.1007/978-3-642-40885-4\_24}.

\bibitem{lctrsTocl}
Fuhs C, Kop C, Nishida N.
\newblock Verifying Procedural Programs via Constrained Rewriting Induction.
\newblock \emph{{ACM} Trans. Comput. Log.}, 2017.
\newblock \textbf{18}(2):14:1--14:50.
\newblock \urlprefix\url{https://doi.org/10.1145/3060143}.

\bibitem{thesisKop}
Kop C.
\newblock Higher Order Termination.
\newblock Ph.D. thesis, VU Amsterdam, 2012.

\bibitem{esop24lcstrs}
Guo L, Kop C.
\newblock Higher-Order LCTRSs and Their Termination.
\newblock In: Weirich S (ed.), Programming Languages and Systems - 33rd
  European Symposium on Programming, {ESOP} 2024, Held as Part of the European
  Joint Conferences on Theory and Practice of Software, {ETAPS} 2024,
  Luxembourg City, Luxembourg, April 6-11, 2024, Proceedings, Part {II}, volume
  14577 of \emph{Lecture Notes in Computer Science}. Springer, 2024 pp.
  331--357.
\newblock \urlprefix\url{https://doi.org/10.1007/978-3-031-57267-8\_13}.

\bibitem{BaaderNipkow}
Baader F, Nipkow T.
\newblock Term rewriting and all that.
\newblock Cambridge Univ.\ Press, 1998.
\newblock ISBN 978-0-521-45520-6.

\bibitem{Hirokawa08IJCAR}
Hirokawa N, Moser G.
\newblock Automated Complexity Analysis Based on the Dependency Pair Method.
\newblock In: Armando A, Baumgartner P, Dowek G (eds.), Automated Reasoning,
  4th International Joint Conference, {IJCAR} 2008, Sydney, Australia, August
  12-15, 2008, Proceedings, volume 5195 of \emph{Lecture Notes in Computer
  Science}. Springer, 2008 pp. 364--379.
\newblock \urlprefix\url{https://doi.org/10.1007/978-3-540-71070-7\_32}.

\bibitem{HirokawaMoser14}
Hirokawa N, Moser G.
\newblock Automated Complexity Analysis Based on Context-Sensitive Rewriting.
\newblock In: Dowek G (ed.), Rewriting and Typed Lambda Calculi - Joint
  International Conference, {RTA-TLCA} 2014, Held as Part of the Vienna Summer
  of Logic, {VSL} 2014, Vienna, Austria, July 14-17, 2014. Proceedings, volume
  8560 of \emph{Lecture Notes in Computer Science}. Springer, 2014 pp.
  257--271.
\newblock \urlprefix\url{https://doi.org/10.1007/978-3-319-08918-8\_18}.

\bibitem{ava:mos:16}
Avanzini M, Moser G.
\newblock A combination framework for complexity.
\newblock \emph{Information and Computation}, 2016.
\newblock \textbf{248}:22--55.
\newblock \urlprefix\url{https://doi.org/10.1016/j.ic.2015.12.007}.

\bibitem{naa:fro:bro:fuh:gie:17}
Naaf M, Frohn F, Brockschmidt M, Fuhs C, Giesl J.
\newblock Complexity Analysis for Term Rewriting by Integer Transition Systems.
\newblock In: Dixon C, Finger M (eds.), Frontiers of Combining Systems - 11th
  International Symposium, FroCoS 2017, Bras{\'{\i}}lia, Brazil, September
  27-29, 2017, Proceedings, volume 10483 of \emph{Lecture Notes in Computer
  Science}. Springer, 2017 pp. 132--150.
\newblock \urlprefix\url{https://doi.org/10.1007/978-3-319-66167-4\_8}.

\bibitem{LowerBounds}
Frohn F, Giesl J, Hensel J, Aschermann C, Str{\"{o}}der T.
\newblock Lower Bounds for Runtime Complexity of Term Rewriting.
\newblock \emph{J. Autom. Reason.}, 2017.
\newblock \textbf{59}(1):121--163.
\newblock \urlprefix\url{https://doi.org/10.1007/s10817-016-9397-x}.

\bibitem{MoserS20}
Moser G, Schneckenreither M.
\newblock Automated amortised resource analysis for term rewrite systems.
\newblock \emph{Sci. Comput. Program.}, 2020.
\newblock \textbf{185}.
\newblock \urlprefix\url{https://doi.org/10.1016/j.scico.2019.102306}.

\bibitem{DependencyPairs}
Arts T, Giesl J.
\newblock Termination of Term Rewriting Using Dependency Pairs.
\newblock \emph{Theoretical Computer Science}, 2000.
\newblock \textbf{236}:133--178.

\bibitem{positiveness}
Hong H, Jakus D.
\newblock Testing Positiveness of Polynomials.
\newblock \emph{J. Autom. Reason.}, 1998.
\newblock \textbf{21}(1):23--38.
\newblock \urlprefix\url{https://doi.org/10.1023/A:1005983105493}.

\bibitem{BlellochFPCA95}
Blelloch GE, Greiner J.
\newblock Parallelism in Sequential Functional Languages.
\newblock In: Williams J (ed.), Proceedings of the seventh international
  conference on Functional programming languages and computer architecture,
  {FPCA} 1995, La Jolla, California, USA, June 25-28, 1995. {ACM}, 1995 pp.
  226--237.
\newblock \urlprefix\url{https://doi.org/10.1145/224164.224210}.

\bibitem{LoopsUnderStrategies}
Thiemann R, Sternagel C, Giesl J, Schneider{-}Kamp P.
\newblock Loops under Strategies ... Continued.
\newblock In: Kirchner H, Mu{\~{n}}oz CA (eds.), Proceedings International
  Workshop on Strategies in Rewriting, Proving, and Programming, {IWS} 2010,
  Edinburgh, UK, 9th July 2010, volume~44 of \emph{{EPTCS}}. 2010 pp. 51--65.
\newblock \urlprefix\url{https://doi.org/10.4204/EPTCS.44.4}.

\bibitem{vvO}
van Oostrom V.
\newblock Remarks on the full parallel innermost strategy, 2023.
\newblock \urlprefix\url{http://www.javakade.nl/research/pdf/fpi.pdf}.

\bibitem{maxpolo}
Fuhs C, Giesl J, Middeldorp A, Schneider{-}Kamp P, Thiemann R, Zankl H.
\newblock Maximal Termination.
\newblock In: Voronkov A (ed.), Rewriting Techniques and Applications, 19th
  International Conference, {RTA} 2008, Hagenberg, Austria, July 15-17, 2008,
  Proceedings, volume 5117 of \emph{Lecture Notes in Computer Science}.
  Springer, 2008 pp. 110--125.
\newblock \urlprefix\url{https://doi.org/10.1007/978-3-540-70590-1\_8}.

\bibitem{Lamport78}
Lamport L.
\newblock Time, Clocks, and the Ordering of Events in a Distributed System.
\newblock \emph{Commun. {ACM}}, 1978.
\newblock \textbf{21}(7):558--565.
\newblock \urlprefix\url{https://doi.org/10.1145/359545.359563}.

\bibitem{tpdb}
Wiki.
\newblock {T}ermination {P}roblems {D}ata{B}ase ({TPDB}).
\newblock \url{http://termination-portal.org/wiki/TPDB}.

\bibitem{termcomp}
Giesl J, Rubio A, Sternagel C, Waldmann J, Yamada A.
\newblock The Termination and Complexity Competition.
\newblock In: Beyer D, Huisman M, Kordon F, Steffen B (eds.), Tools and
  Algorithms for the Construction and Analysis of Systems - 25 Years of
  {TACAS:} TOOLympics, Held as Part of {ETAPS} 2019, Prague, Czech Republic,
  April 6-11, 2019, Proceedings, Part {III}, volume 11429 of \emph{Lecture
  Notes in Computer Science}. Springer, 2019 pp. 156--166.
\newblock \urlprefix\url{https://doi.org/10.1007/978-3-030-17502-3\_10}.

\bibitem{termcompWiki}
Wiki.
\newblock The {I}nternational {T}ermination {C}ompetition ({TermComp}).
\newblock \url{http://termination-portal.org/wiki/Termination_Competition}.

\bibitem{context-sensitive}
Lucas S.
\newblock Context-sensitive Rewriting.
\newblock \emph{{ACM} Comput. Surv.}, 2021.
\newblock \textbf{53}(4):78:1--78:36.
\newblock \urlprefix\url{https://doi.org/10.1145/3397677}.

\bibitem{CoFloCo}
Flores{-}Montoya A, H{\"{a}}hnle R.
\newblock Resource Analysis of Complex Programs with Cost Equations.
\newblock In: Garrigue J (ed.), Programming Languages and Systems - 12th Asian
  Symposium, {APLAS} 2014, Singapore, November 17-19, 2014, Proceedings, volume
  8858 of \emph{Lecture Notes in Computer Science}. Springer, 2014 pp.
  275--295.
\newblock \urlprefix\url{https://doi.org/10.1007/978-3-319-12736-1\_15}.

\bibitem{CoFloCoFM16}
Flores{-}Montoya A.
\newblock Upper and Lower Amortized Cost Bounds of Programs Expressed as Cost
  Relations.
\newblock In: Fitzgerald JS, Heitmeyer CL, Gnesi S, Philippou A (eds.), {FM}
  2016: Formal Methods - 21st International Symposium, Limassol, Cyprus,
  November 9-11, 2016, Proceedings, volume 9995 of \emph{Lecture Notes in
  Computer Science}. 2016 pp. 254--273.
\newblock \urlprefix\url{https://doi.org/10.1007/978-3-319-48989-6\_16}.

\bibitem{irc_persistent}
Avanzini M, Felgenhauer B.
\newblock Type introduction for runtime complexity analysis.
\newblock In: WST~'14. 2014 pp. 1--5.
\newblock Available from
  \url{http://www.easychair.org/smart-program/VSL2014/WST-proceedings.pdf}.

\bibitem{KB70}
Knuth DE, Bendix PB.
\newblock Simple Word Problems in Universal Algebras.
\newblock In: Leech J (ed.), Computational Problems in Abstract Algebra.
  Pergamon Press, 1970 pp. 263--297.

\bibitem{Rosen73}
Rosen BK.
\newblock Tree-Manipulating Systems and Church-Rosser Theorems.
\newblock \emph{J. {ACM}}, 1973.
\newblock \textbf{20}(1):160--187.
\newblock \urlprefix\url{https://doi.org/10.1145/321738.321750}.

\bibitem{Huet80}
Huet GP.
\newblock Confluent Reductions: Abstract Properties and Applications to Term
  Rewriting Systems: Abstract Properties and Applications to Term Rewriting
  Systems.
\newblock \emph{J. {ACM}}, 1980.
\newblock \textbf{27}(4):797--821.
\newblock \urlprefix\url{https://doi.org/10.1145/322217.322230}.

\bibitem{coco}
Community.
\newblock The International {C}onfluence {C}ompetition ({CoCo}).
\newblock \url{http://project-coco.uibk.ac.at/}.

\bibitem{iwc2022}
Baudon T, Fuhs C, Gonnord L.
\newblock On Confluence of Parallel-Innermost Term Rewriting.
\newblock In: Winkler S, Rocha C (eds.), Proceedings of the 11th International
  Workshop on Confluence. 2022 pp. 31--36.
\newblock
  \urlprefix\url{http://cl-informatik.uibk.ac.at/iwc/2022/proceedings.pdf}.

\bibitem{GramlichPhD}
Gramlich B.
\newblock Termination and confluence: properties of structured rewrite systems.
\newblock Ph.D. thesis, Kaiserslautern University of Technology, Germany, 1996.
\newblock \urlprefix\url{https://d-nb.info/949807389}.

\bibitem{starexec}
Stump A, Sutcliffe G, Tinelli C.
\newblock StarExec: {A} Cross-Community Infrastructure for Logic Solving.
\newblock In: Demri S, Kapur D, Weidenbach C (eds.), Automated Reasoning - 7th
  International Joint Conference, {IJCAR} 2014, Held as Part of the Vienna
  Summer of Logic, {VSL} 2014, Vienna, Austria, July 19-22, 2014. Proceedings,
  volume 8562 of \emph{Lecture Notes in Computer Science}. Springer, 2014 pp.
  367--373.
\newblock \url{https://www.starexec.org/},
  \urlprefix\url{https://doi.org/10.1007/978-3-319-08587-6\_28}.

\bibitem{evalPageJournal}
\urlprefix\url{https://www.dcs.bbk.ac.uk/~carsten/eval/parallel_complexity_journal/}.

\bibitem{versionPage}
TCT, version from the Termination and Complexity Competitions 2020 -- 2022.
\newblock
  \urlprefix\url{https://www.starexec.org/starexec/secure/details/solver.jsp?id=29575}.

\bibitem{cops}
Hirokawa N, Nagele J, Middeldorp A.
\newblock Cops and {CoCoWeb}: Infrastructure for Confluence Tools.
\newblock In: Automated Reasoning - 9th International Joint Conference, {IJCAR}
  2018, Held as Part of the Federated Logic Conference, FloC 2018, Oxford, UK,
  July 14-17, 2018, Proceedings, volume 10900 of \emph{LNCS}. Springer, 2018
  pp. 346--353.
\newblock See also: \url{https://cops.uibk.ac.at/},
  \urlprefix\url{https://doi.org/10.1007/978-3-319-94205-6\_23}.

\bibitem{AlbertLCTES11}
Albert E, Arenas P, Genaim S, Zanardini D.
\newblock {Task-level analysis for a language with async/finish parallelism}.
\newblock In: Vitek J, Sutter BD (eds.), Proceedings of the {ACM}
  {SIGPLAN/SIGBED} 2011 conference on Languages, compilers, and tools for
  embedded systems, {LCTES} 2011, Chicago, IL, USA, April 11-14, 2011. {ACM},
  2011 pp. 21--30.
\newblock \urlprefix\url{https://doi.org/10.1145/1967677.1967681}.

\bibitem{Albert2012}
Albert E, Arenas P, Genaim S, Puebla G, Zanardini D.
\newblock Cost analysis of object-oriented bytecode programs.
\newblock \emph{Theor. Comput. Sci.}, 2012.
\newblock \textbf{413}(1):142--159.
\newblock \urlprefix\url{https://doi.org/10.1016/j.tcs.2011.07.009}.

\bibitem{raml}
Hoffmann J, Aehlig K, Hofmann M.
\newblock Resource Aware {ML}.
\newblock In: Madhusudan P, Seshia SA (eds.), Computer Aided Verification -
  24th International Conference, {CAV} 2012, Berkeley, CA, USA, July 7-13, 2012
  Proceedings, volume 7358 of \emph{Lecture Notes in Computer Science}.
  Springer, 2012 pp. 781--786.
\newblock \urlprefix\url{https://doi.org/10.1007/978-3-642-31424-7\_64}.

\bibitem{lctrsComplexity}
Winkler S, Moser G.
\newblock Runtime Complexity Analysis of Logically Constrained Rewriting.
\newblock In: Fern{\'{a}}ndez M (ed.), Logic-Based Program Synthesis and
  Transformation - 30th International Symposium, {LOPSTR} 2020, Bologna, Italy,
  September 7-9, 2020, Proceedings, volume 12561 of \emph{Lecture Notes in
  Computer Science}. Springer, 2020 pp. 37--55.
\newblock \urlprefix\url{https://doi.org/10.1007/978-3-030-68446-4\_2}.

\bibitem{derivational89}
Hofbauer D, Lautemann C.
\newblock Termination Proofs and the Length of Derivations.
\newblock In: Dershowitz N (ed.), Rewriting Techniques and Applications, 3rd
  International Conference, RTA-89, Chapel Hill, North Carolina, USA, April
  3-5, 1989, Proceedings, volume 355 of \emph{Lecture Notes in Computer
  Science}. Springer, 1989 pp. 167--177.
\newblock \urlprefix\url{https://doi.org/10.1007/3-540-51081-8\_107}.

\bibitem{CompositionalConfluence22}
Shintani K, Hirokawa N.
\newblock Compositional Confluence Criteria.
\newblock In: Felty AP (ed.), 7th International Conference on Formal Structures
  for Computation and Deduction, {FSCD} 2022, August 2-5, 2022, Haifa, Israel,
  volume 228 of \emph{LIPIcs}. Schloss Dagstuhl - Leibniz-Zentrum f{\"{u}}r
  Informatik, 2022 pp. 28:1--28:19.
\newblock \urlprefix\url{https://doi.org/10.4230/LIPIcs.FSCD.2022.28}.

\bibitem{CompositionalConfluence24}
Shintani K, Hirokawa N.
\newblock Compositional Confluence Criteria.
\newblock \emph{Log. Methods Comput. Sci.}, 2024.
\newblock \textbf{20}(1).

\bibitem{confident}
Guti{\'{e}}rrez R, V{\'{\i}}tores M, Lucas S.
\newblock Confluence Framework: Proving Confluence with CONFident.
\newblock In: Villanueva A (ed.), Logic-Based Program Synthesis and
  Transformation - 32nd International Symposium, {LOPSTR} 2022, Tbilisi,
  Georgia, September 21-23, 2022, Proceedings, volume 13474 of \emph{Lecture
  Notes in Computer Science}. Springer, 2022 pp. 24--43.
\newblock \urlprefix\url{https://doi.org/10.1007/978-3-031-16767-6\_2}.

\bibitem{cime}
Contejean E, Courtieu P, Forest J, Pons O, Urbain X.
\newblock Automated Certified Proofs with CiME3.
\newblock In: Schmidt{-}Schau{\ss} M (ed.), Proceedings of the 22nd
  International Conference on Rewriting Techniques and Applications, {RTA}
  2011, May 30 - June 1, 2011, Novi Sad, Serbia, volume~10 of \emph{LIPIcs}.
  Schloss Dagstuhl - Leibniz-Zentrum f{\"{u}}r Informatik, 2011 pp. 21--30.
\newblock \urlprefix\url{https://doi.org/10.4230/LIPIcs.RTA.2011.21}.

\bibitem{color}
Blanqui F, Koprowski A.
\newblock CoLoR: a Coq library on well-founded rewrite relations and its
  application to the automated verification of termination certificates.
\newblock \emph{Math. Struct. Comput. Sci.}, 2011.
\newblock \textbf{21}(4):827--859.
\newblock \urlprefix\url{https://doi.org/10.1017/S0960129511000120}.

\bibitem{ceta}
Thiemann R, Sternagel C.
\newblock Certification of Termination Proofs Using CeTA.
\newblock In: Berghofer S, Nipkow T, Urban C, Wenzel M (eds.), Theorem Proving
  in Higher Order Logics, 22nd International Conference, TPHOLs 2009, Munich,
  Germany, August 17-20, 2009. Proceedings, volume 5674 of \emph{Lecture Notes
  in Computer Science}. Springer, 2009 pp. 452--468.
\newblock \urlprefix\url{https://doi.org/10.1007/978-3-642-03359-9\_31}.

\bibitem{nijn}
van~der Weide N, Vale D, Kop C.
\newblock Certifying Higher-Order Polynomial Interpretations.
\newblock In: Naumowicz A, Thiemann R (eds.), 14th International Conference on
  Interactive Theorem Proving, {ITP} 2023, July 31 to August 4, 2023,
  Bia{\l}ystok, Poland, volume 268 of \emph{LIPIcs}. Schloss Dagstuhl -
  Leibniz-Zentrum f{\"{u}}r Informatik, 2023 pp. 30:1--30:20.
\newblock \urlprefix\url{https://doi.org/10.4230/LIPIcs.ITP.2023.30}.

\end{thebibliography}
\end{document}